\documentclass{article}

\usepackage[paperwidth=7in, paperheight=10in, textwidth=5.8in, textheight=8.2in]{geometry}

\usepackage{latexsym}
\usepackage{amsmath}
\usepackage{amsthm}
\usepackage{mathtools}
\usepackage{forloop}
\usepackage[ruled,vlined,linesnumbered]{algorithm2e}
\usepackage{nicefrac}
\usepackage{paralist}
\usepackage{xcolor} 
\usepackage{tabularx}
\usepackage{xifthen}
\usepackage{makecell}
\usepackage{tikz}
\usetikzlibrary{arrows.meta, positioning, quotes}
\usepackage{authblk}

\usepackage{epsfig}
\usepackage{amssymb}
\usepackage[hidelinks]{hyperref}
\usepackage{graphicx}

\usepackage{natbib}
\usepackage{enumitem}
\usepackage{float}
\usepackage{graphicx,subfig}
\usepackage{amsmath}
\usepackage{mathtools}
\usepackage{booktabs}
\usepackage{nicefrac}
\usepackage{siunitx}
\usepackage{caption}

\newtheorem{theorem}{Theorem}
\newtheorem{lemma}[theorem]{Lemma}
\newtheorem{corollary}[theorem]{Corollary}
\newtheorem{proposition}[theorem]{Proposition}

\newcommand{\eps}{\varepsilon}
\newcommand{\ie}{{i.e.}}
\newcommand{\ExchangeAlg}{{\textsc{Exchange-Candidate}}\xspace}
\newcommand{\algdt}{\textsc{{RepeatedGreedy}}\xspace}
\newcommand{\algrd}{\textsc{{SampleGreedy}}\xspace}
\newcommand{\algfantom}{\textsc{{Fantom}}\xspace}
\newcommand{\characteristic}{{\mathbf{1}}}
\newcommand{\nnR}{{\bR_{\geq 0}}}
\DeclareMathOperator*{\argmax}{arg\,max}

\newcommand{\OPT}{\text{OPT}} 
\newcommand{\ground}{\cN}
\newcommand{\reals}{\mathbb{R}}
\newcommand{\AlgSeq}{\textsc{seqDPP}\xspace}
\newcommand{\AlgLocal}{\textsc{Local-Search}\xspace}
\newcommand{\AlgSampling}{\textsc{Sample-Streaming}\xspace}
\newcommand{\msg}{\textsc{Max-SampleGreedy}\xspace}
\newcommand{\CondProb}[2]{\mathrm{Pr} \left[ #1 \, \middle|  \, #2 \right] } 
\DeclarePairedDelimiter{\evdel}{[}{]} 
\newcommand{\E}{\mathbb{E}\evdel}
\renewcommand{\epsilon}{\varepsilon}
\usepackage{forloop}
\newcommand{\defcal}[1]{\expandafter\newcommand\csname c#1\endcsname{{\mathcal{#1}}}}
\newcommand{\defbb}[1]{\expandafter\newcommand\csname b#1\endcsname{{\mathbb{#1}}}}
\newcounter{calBbCounter}
\forLoop{1}{26}{calBbCounter}{
	\edef\letter{\Alph{calBbCounter}}
	\expandafter\defcal\letter
	\expandafter\defbb\letter
}

\usepackage[ruled, linesnumbered]{algorithm2e}
\usepackage{etoolbox}

\theoremstyle{TH}%
\newtheorem{manualtheoreminner}{Proposition}
\newenvironment{manualtheorem}[1]{%
	\renewcommand\themanualtheoreminner{#1}%
	\manualtheoreminner
}{\endmanualtheoreminner}

\newtheorem{reduction}{Reduction}
\newtheorem{observation}{Observation}

\title{The Power of Subsampling in \\ Submodular Maximization\footnote{The results presented in this paper previously appeared in COLT 2017 \citep{FHK17} and NeurIPS 2018 \citep{FKK18} in the form of extended abstracts. We note, however, that some results of~\cite{FHK17} were not included in the current paper since they are unrelated to its central theme. Nevertheless, for completeness, we reference these results while discussing the related work.}}

\author[1]{Christopher Harshaw}
\author[2]{Ehsan Kazemi}
\author[3]{Moran Feldman}
\author[4]{Amin Karbasi}
\affil[1]{Yale University, Department of Computer Science}
\affil[2]{Google, Z{\"u}rich}
\affil[3]{University of Haifa, Department of Computer Science}
\affil[4]{Yale University, Department of Electrical Engineering}

\date{}

\begin{document}
\maketitle

\begin{abstract}	
We propose subsampling as a unified algorithmic technique for submodular maximization in centralized and online settings.
The idea is simple: independently sample elements from the ground set, and use simple combinatorial techniques (such as greedy or local search) on these sampled elements. We show that this approach leads to optimal/state-of-the-art results despite being much simpler than existing methods.
In the usual offline setting, we present \algrd, which obtains a $\left( p+2 + o(1) \right)$-approximation for maximizing a submodular function subject to a $p$-extendible system using $O(n + nk/p)$ evaluation and feasibility queries, where $k$ is the size of the largest feasible set.
The approximation ratio improves to $p+1$ and $p$ for monotone submodular and linear objectives, 
respectively.
In the streaming setting, we present \AlgSampling, which obtains a $\left(4p +2 - o(1)\right)$-approximation for maximizing a submodular function subject to a $p$-matchoid using $O(k)$ memory and $O(km/p)$ evaluation and feasibility queries per element, where $m$ is the number of matroids defining the $p$-matchoid. 
The approximation ratio improves to $4p$ for monotone submodular objectives.
We empirically demonstrate the effectiveness of our algorithms on video summarization, location 
summarization, and movie recommendation tasks.

\end{abstract}

\section{Introduction} \label{sec:intro}

Submodular functions have played a celebrated role in both the theory of discrete optimization and 
practical modeling scenarios.
Submodular functions are defined by a diminishing returns property, which makes this function class 
natural for modeling many applications in a wide variety of domains, from economics to machine 
learning.
Constrained submodular maximization has found numerous applications, including viral marketing 
\citep{kempe03}, network monitoring \citep{leskovec2007cost, gomez10}, sensor placement and 
information gathering \citep{guestrin2005near}, news article recommendation \citep{el2009turning}, movie recommendation \citep{mitrovic2019adaptive,haba2020streaming}, 
active set selection in non-parametric learning~\citep{MKSK16}, image summarization~\citep{TIWB14,kazemi2020regularized}, location summarization \citep{badanidiyuru2020submodular},
corpus summarization~\citep{LB11,kirchhoff2014submodularity,sipos2012temporal}, fMRI  parcellation 
\citep{SKSC17}, ensuring privacy and fairness \citep{kazemi2018scalable, 
mitrovic2017differentially}, two-stage sub-linear data 
summarization \citep{BMKS16, stanZK17, mitrovic2018data} and removing redundant elements from DNA 
sequencing \citep{LBN18}. For a more detailed description of theses applications 
in machine learning and signal processing, we refer the interested reader to 
\cite{tohidi2020submodularity}. 

Although producing an exactly optimal solution for constrained submodular maximization is 
computationally hard \citep{NW78}, the seminal work of \citet{Nemhauser1978} showed that 
the 
natural greedy algorithm 
produces a $(1 - e^{-1})^{-1}$-approximate solution when the objective function is monotone 
submodular 
and the constraint set is a cardinality constraint; however, the greedy algorithm may perform 
much worse in more complex scenarios, e.g., when the objective is non-monotone or the constraint is more involved.
Most existing algorithms developed for these more complicated settings can be grouped into a few 
categories:
repeated greedy procedures, local search 
techniques, and relax-and-round methods which go through a continuous relaxation of the problem. 
Unfortunately, these techniques tend to be quite slow, and moreover, some of them are quite 
complex,
making their implementation challenging \cite{buchbinder2018submodular}.

In this paper, we propose \emph{subsampling} as a simpler alternative to existing 
algorithmic techniques. 
In particular, we present two algorithms which use subsampling to achieve better approximation guarantees than 
existing techniques at a fraction of the computational costs.
At the heart of our algorithms is a carefully designed{\textemdash}but simple 
to implement{\textemdash}subsampling of the ground set.
Interestingly, our algorithms naturally produce state-of-the-art results for both monotone and non-monotone objective functions; which is rare for a submodular maximization algorithm.

Our first algorithm, \algrd, is designed for maximizing a submodular function subject to a 
$p$-extendible system. The algorithm achieves a $(p+1)^2 / p = p + 2 + o(1)$ approximation guarantee (which is 
nearly tight for this problem by a result of~\cite{FHK17}) and uses only $O(n + \frac{nk}{p})$ function 
evaluations and 
feasibility queries, where $n$ is the size of the ground set and $k$ is the largest feasible set.
The technique is simple: independently sample elements from the ground set and run the greedy 
algorithm. 
Moreover, the approximation guarantee improves to $p+1$ or $p$ when the function is monotone or 
linear, respectively.
Our second algorithm, \AlgSampling, is designed for maximizing a submodular function subject to a 
$p$-matchoid constraint in the streaming setting where elements arrive one at a time and only a small 
working memory is kept. \AlgSampling achieves a $4p +2 - o(1)$ approximation ratio in this setting
and uses $O(k)$ memory. To process the arrival of every element the algorithm uses, in expectation, 
$O(km/p)$ function evaluations and matroid feasibility queries, where $m$ is 
the number of matroids used to define the $p$-matchoid. We also note that the approximation ratio 
improves to $4p$ for monotone functions.

We empirically demonstrate the effectiveness of our subsampling based algorithms for video 
summarization and movie recommendation tasks with real datasets. We show that our algorithms are
competitive with respect to existing algorithms, but require a fraction of the computational cost.

\paragraph{Organization.}
Section~\ref{sec:related_works} contains a brief summary of related works on 
constrained submodular maximization as well as a comparison of our new results with these works. 
We review preliminary definitions in Section~\ref{sec:preliminaries_revised}.
Section~\ref{sec:main_results} presents and analyzes our subsampling algorithms, focusing on the 
offline algorithm in Sections~\ref{sec:centralized_alg} and the streaming algorithm in 
Section~\ref{sec:streaming_alg}.
In Section~\ref{sec:experimental_results}, we empirically evaluate the performance of our algorithms 
against existing methods on 
real datasets.
Finally, we provide concluding remarks in Section~\ref{sec:conclusion}.

\section{Related Work} \label{sec:related_works}

In this section, we briefly survey the most relevant related work on constrained submodular function 
maximization.
In what follows, $n$ is the size of the ground set, $p$ is in reference to $p$-extendible systems or 
$p$-matchoids, and $k$ is the size of the largest independent set.\footnote{These terms, along with a 
precise notion of oracle complexity, are defined in 
Section~\ref{sec:preliminaries_revised}.}

The most relevant works for comparing our \algrd algorithm are the repeated greedy algorithms, which have been 
historically developed for a slightly broader class of constraints known as $p$-systems.
\citet{FNW78} showed that the natural greedy algorithm achieves a $(p+1)$ approximation for 
maximizing a 
monotone submodular function subject to a $p$-system. 
Algorithms for the non-monotone variant of this problem were developed only much more recently and rely on repeated applications of the greedy algorithm. 
 \citet{GRST10} showed that iteratively running the greedy algorithm on the constrained problem 
 followed by 
an unconstrained optimization on the greedy solution results in an approximation guarantee of roughly $3p$ for 
general non-monotone submodular functions, while requiring $O(nkp)$ function evaluations and independence 
oracles queries.
Using a different analysis, \cite{MBK16} improved the approximation guarantee of this algorithm to roughly
$2p$.
\citet{FHK17} showed that an improved approximation guarantee of $p + O(\sqrt{p})$ is possible with 
fewer iterations of the repeated greedy procedure.
The main drawback of all these algorithms is the large number of function evaluation and independence 
oracle queries they require, which grows unfavorably with $p$. 
Our proposed subsampling based algorithm, \algrd, significantly improves upon these algorithms in the 
case of $p$-extendible systems in two ways: the oracle complexity is greatly reduced and the 
approximation guarantee is improved.
A summary and comparison of these algorithms are presented in Table~\ref{tbl:summary-greedy}.

\begin{table}
	\begin{center}
		\begin{tabular}{m{2.1cm}m{2.5cm}>{\centering\arraybackslash}m{2.8cm}>{\centering\arraybackslash}m{2.89cm}m{4.3cm}}
			\toprule
			\textbf{\footnotesize Algorithm} & \textbf{\footnotesize  Function} & \textbf{\footnotesize 
			Approx. Ratio} & \textbf{\footnotesize Query Complexity} & \textbf{\footnotesize Reference} \\
			\hline
			{\footnotesize Deterministic} & {\footnotesize Monotone} &  {\footnotesize $p + 1$} 
			& {\footnotesize $O(nk)$} & {\footnotesize \citet{FNW78} }\\
			{\footnotesize Deterministic} & {\footnotesize Non-monotone} &  {\footnotesize $\approx 2p$} 
			& {\footnotesize $O(nkp)$} & {\footnotesize \citet{MBK16}}\\
			{\footnotesize Deterministic} & {\footnotesize Non-monotone} &  {\footnotesize $\approx 3p$} 
			& {\footnotesize $O(nkp)$} & {\footnotesize \citet{GRST10}}\\
			{\footnotesize Deterministic} & {\footnotesize Non-monotone} &  {\footnotesize $p + 
			O(\sqrt{p})$} & {\footnotesize $O(nk \sqrt{p})$} & {\footnotesize \citet{FHK17}}\\
			\hline 
			\vspace{0.1cm}
			{\footnotesize Randomized} & {\footnotesize Non-monotone} &  {\footnotesize $\frac{(p+1)^2}{p} = 
			p + 2 + o(1)$} & {\footnotesize $O(n + nk / p)$} & {\footnotesize \algrd (this paper)}\\
			{\footnotesize Randomized} & {\footnotesize Monotone} &  {\footnotesize $p+1$} & 
			{\footnotesize $O(n + nk / p)$} & {\footnotesize \algrd (this paper)}\\
			\bottomrule
		\end{tabular}
	\end{center}
	\caption{Greedy algorithms for submodular maximization subject to a $p$-extendible constraint.} 
	\label{tbl:summary-greedy}
\end{table}

Local search algorithms have also been proposed for maximization over various subclasses of $p$-extendible systems. 
\citet{LSV10} developed a local search method which attains a $p + \epsilon$ 
approximation for maximizing a monotone submodular function subject to the intersection of $p$ 
matroids using a number of evaluations and independence oracle queries which is polynomial in $n$ 
and exponential in $\frac{1}{\epsilon}$.
They also showed how to use this algorithm to obtain a $p + 1 + \frac{1}{p+1} + \epsilon$ 
approximation for non-monotone objectives, improving over a $(p + 2 + \frac{1}{p} + \epsilon)$ approximation due to \cite{LMNS10}.
 \citet{FNSW11} showed that the same results can also be obtained for maximization over a $p$-exchange system, which is a different subclass of $p$-extendible systems.
For $p \geq 4$, the last approximation guarantee was improved to $ (p+3)/2+ \epsilon $ by \citet{W12}.
Despite running in polynomial time, these local search algorithms have very large oracle queries complexity, and so they are mostly of theoretical interest.

There has also been a long sequence of works which aim to obtain a tighter approximation guarantees 
for the special case of matroid constraints. 
Such methods rely on approximately optimizing continuous 
extensions of the discrete submodular objective, followed by rounding to obtain a discrete 
solution. The seminal work of \citet{CCPV11} showed that this technique achieves the
tight $(1 - e^{-1})^{-1}$ approximation ratio for maximizing a monotone submodular function subject to 
a 
matroid constraint. In the non-monotone setting, a long series of work~\citep{V13,GV11,FNS11,EN16,BF19} 
has further developed 
these techniques to obtain a $2.59$ approximation ratio, but the best inapproximability result is 
still slightly further away at $2.09$~\citep{GV11}.
Although these algorithms achieve tighter approximation guarantees for the special case of matroids, 
they suffer from a high evaluation oracle complexity due to the sampling techniques they use to obtain 
gradient estimates for the continuous extension. A more recent line of work suggests some techniques 
to (partially) remedy this problem 
\citep{Badanidiyuru14,buchbinder2016comparing,mokhtari2017conditional,KMZ18,BFG19,EN18}.
Finally, we remark that \citep{Mirzasoleiman15} devised a different randomized subsampling 
technique which achieves a $(1 - e^{-1} - \epsilon)^{-1}$ approximation ratio using $O(n \log 
\frac{1}{\epsilon})$ evaluation queries for monotone submodular objectives under the cardinality 
constraint, improving upon the query complexity of the greedy algorithm.

Additional recent work in constrained submodular maximization has focused on the streaming 
environment, 
where data points appear one at a time and centralized storage capacity is limited. A streaming 
algorithm for monotone submodular maximization under a cardinality constraint was presented by 
\citet{BMKK14}, which achieves a $1/2 - \eps$ 
approximation using $O(\eps^{-1} k \log k)$ memory. 
 Recently, \citet{kazemi2019submodular} presented modification of this algorithm which reduces the memory complexity 
 to $O(k/\eps)$.
A different series of work \citep{CK15, CGQ15} used a different technique to provide a $4p$ approximation for monotone submodular maximization subject to 
$p$-matchoid constraints.
The first streaming algorithm for non-monotone submodular maximization was given by \citet{BFS15}, whose randomized algorithm achieves 11.197 approximation for non-monotone maximization 
under a 
cardinality constraint.
This was shortly after improved by \citet{CGQ15}, who presented a randomized streaming algorithm for non-monotone maximization under 
$p$-matchoid constraints which achieves an approximation ratio of $(5p + 2 + 1/p)/(1-\eps)$ and a deterministic 
algorithm which 
achieves a slightly worse approximation ratio of $(9p + O(\sqrt{p}))/(1-\eps)$ but is more memory and 
update efficient.
Recently, \citet{MJK17} proposed a deterministic algorithm which they claim achieves an approximation ratio of $4p + 4\sqrt{p} + 
1$ and uses $O(k\sqrt{p})$ memory; however, \cite{haba2020streaming} pointed out several errors in their analysis and so the guarantees of \cite{MJK17} may not hold.
While the monotone algorithms mentioned above are quite efficient in terms of memory and update cost, 
the non-monotone algorithms are much less efficient in these aspects, having unfavorable 
dependence on $p$ or $\epsilon$ terms.
In contrast, our randomized streaming algorithm, \AlgSampling, achieves an improved approximation ratio 
of $4p + 2 - o(1)$ for non-monotone maximization over a $p$-matchoid constraint using only $O(k)$ memory and 
$O(km/p)$ expected evaluation and independence queries per iteration. With a minor modification to the algorithm, this 
approximation ratio improves to $4p$ for monotone functions.
A summary and comparison of these algorithms is given in Table~\ref{table:summary-streaming}

\begin{table}
	\begin{center}
		\begin{tabular}{m{2.1cm}m{2.1cm}>{\centering\arraybackslash}m{2.1cm}>{\centering\arraybackslash}m{1.8cm}>{\centering\arraybackslash}m{2.0cm}m{4.77cm}}
			\toprule
			\textbf{\footnotesize Algorithm} & \textbf{\footnotesize  Function} & \textbf{\footnotesize 
			Approx. Ratio} & \textbf{\footnotesize Memory} & \textbf{\footnotesize Queries per Element} & 
			\textbf{\footnotesize Reference} \\
			\hline
			{\footnotesize Deterministic} & {\footnotesize Monotone} & {\footnotesize $4p$} & 
			{\footnotesize $O(k)$} & {\footnotesize $O(km)$} & {\footnotesize \citet{CGQ15}}\\
			{\footnotesize Randomized} & {\footnotesize Non-monotone} & {\footnotesize $\frac{5p + 2 + 
			1/p}{{1-\eps}}$} & {\footnotesize $O(\frac{k}{\eps^2} \log \frac{k}{\eps})$} & {\footnotesize 
			$O(\frac{k^2m}{\eps^2} \log \frac{k}{\eps})$} & {\footnotesize \citet{CGQ15}}\\
			{\footnotesize Deterministic} & {\footnotesize Non-monotone} & {\footnotesize $\frac{9p + 
			O(\sqrt{p})}{1-\eps}$} & {\footnotesize $O(\frac{k}{\eps}\log\frac{k}{\eps})$} & {\footnotesize 
			$O(\frac{km}{\eps}\log \frac{k}{\eps})$} & {\footnotesize \citet{CGQ15}}\\
			\hline
			{\footnotesize Randomized} & {\footnotesize Monotone} & {\footnotesize $4p$} & 
			{\footnotesize $O(k)$} & {\footnotesize $O(km/p)$} & {\footnotesize \AlgSampling  {\tiny (this paper)}}\\
			{\footnotesize Randomized} & {\footnotesize Non-monotone} & {\footnotesize $4p + 2 - o(1)$} 
			& {\footnotesize $O(k)$} & {\footnotesize $O(km/p)$} & {\footnotesize \AlgSampling {\tiny (this paper)}}\\
			\bottomrule
		\end{tabular}
	\end{center}
	\caption{Streaming algorithms for submodular maximization subject to a $p$-matchoid constraint.} 
	\label{table:summary-streaming}
\end{table}
\footnotetext{The memory and query complexities of the algorithm of~\citet{MJK17} have been calculated based on the corresponding complexities of the algorithm of~\citep{CGQ15} for monotone objectives and the properties of the reduction used by~\citep{MJK17}. We note that these complexities do not match the memory and query complexities stated by~\citep{MJK17} for their algorithm.}

\section{Preliminaries} \label{sec:preliminaries_revised}
In this section, we describe the mathematical formulation of the constrained submodular maximization 
problem and preliminary definitions.

Let $\ground$ be a finite set of size $n$, which we refer to as the \emph{ground set}.
The objective functions are real-valued set functions of the form $f:2^\ground \rightarrow \reals$, 
which assign a real number to each set $S \subseteq \ground$. 
Such a function $f$ is \emph{submodular} if 
\begin{align} 
f(A \cup \{e\}) - f(A) \geq f(B \cup \{e\}) - f(B) \label{eq:def_submodular}
\end{align}
for all sets $A \subseteq B \subseteq \ground$ and element $e \notin B$.
Inequality~\eqref{eq:def_submodular} is also referred to as the \emph{diminishing returns property}.
Indeed, when $f$ is interpreted as a utility, \eqref{eq:def_submodular} states that the marginal gain in 
utility of an element $e \in \ground$ decreases as the current set grows.
For shorthand, we write the marginal gain of an element as $f( e \mid S) \triangleq f(S \cup \{e\}) - f(S)$ 
and the marginal gain of adding an entire set as $f( A \mid S) \triangleq f(S \cup A) - f(S)$.
A function $f$ is \emph{monotone} if $f(A) \leq f(B)$ for all sets $A \subseteq B$.
A function $f$ is \emph{linear} if \eqref{eq:def_submodular} holds with equality for all $A \subseteq B$ 
and $e \notin B$.

We now describe the structure of the constraints we consider in this paper.
Given a ground set $\ground$ and a collection of sets $\cI \subseteq 2^\ground$, we say that the pair $( 
\ground, \cI)$ is an \emph{independence system} if $\varnothing \in \cI$ and $A \subseteq B$, $B \in \cI$ 
implies that $A \in \cI$.
A set $A \in \cI$ is called \emph{independent}, and a set $B \notin \cI$ is called \emph{dependent}.
An independent set $A \in \cI$ which is maximal with respect to inclusion is called a \emph{base}; that 
is, 
$B \in \cI$ is a base if $A \in \cI$ and $B \subseteq A$ imply that $B = A$.
Given an independent set $A \in \cI$, an \emph{extension} is an independent set $B \in \cI$ that 
contains $A$, i.e., $A \subseteq B$.
There is a hierarchy of classes of independence systems which are considered in the literature as constraint 
families. 
Our results only require two such classes:  $p$-extendible systems and 
$p$-matchoids. However, for the sake of context, we present here a few additional central classes from the hierarchy.

An independence system $(\ground, \cI)$ is a \emph{$p$-system} if for every set $S \subseteq \ground$, the ratio $ | B_1 | / |B_2|$ is upper bounded by $p$ for every two 
bases $B_1$ and $B_2$ of $(S, 2^S \cap \cI)$.
The class of $p$-systems is the most general class usually included in the hierarchy of independence systems. An important class included in it is the class of $p$-extendible systems.
An independence system is a  \emph{$p$-extendible system} if for all $A \in \cI$, extension $B \in 
\cI$ of $A$ and element $e \notin A$ such that $A \cup \{e\} \in \cI$, there exists a set $Y \subseteq B 
\setminus A$ with $|Y| 
\leq p$ such that $B \setminus Y \cup \{e\} \in \cI$.
Intuitively, an independence system is $p$-extendible if adding an element $e$ to an independent
set $B$ requires the removal of at most $p$ other elements in order to keep the resulting set 
independent.
Another important class of independence systems is the class of matroids. While the usual definition 
of matroids is based on linear algebra intuition, \citep{Mestre2006} showed that this definition is 
equivalent to the definition of a $1$-extendible system.
An independence system $(\ground, \cI)$ is a \emph{$p$-matchoid} if there exist $m$ matroids 
$(\ground_1, \cI_1), \dots (\ground_m, \cI_m)$ such that $\ground = \cup_{i=1}^m \ground_i$, each 
element $e \in \ground$ appears in no more than $p$ ground sets $\ground_1, \dots, \ground_m$ and 
$\cI = \{ S \in \ground \mid \forall_{i=1, \dots, m} \; \ground_i \cap S \in \cI_i \}$.
The hierarchy of the classes of independence systems mentioned above is presented below. We note that all the inclusions between these classes are known to be strict.
\begin{equation*}
\text{matroid} \subset 
\text{intersection of $p$ matroids} \subset 
\text{$p$-matchoid} \subset
\text{$p$-extendible} \subset 
\text{$p$-system} \enspace.
\end{equation*}

This hierarchy of independence systems is quite rich and expressive, containing many classic examples 
which are useful for modeling applications.
The simplest example is the \emph{$k$-cardinality constraint}, where $\cI = \{ S \mid S \subseteq \ground \text{ and } |S| \leq k \}$, 
which 
is also referred to as the \emph{uniform matroid}.
The \emph{partition matroid} is specified by a partition $P_1, P_2, \dots P_\ell \subseteq \ground$ ($\cup_{i=1}^\ell P_i = \ground$ and $P_i \cap P_j = \varnothing$ for $i \neq j$) and integers $k_1, k_2 \dotsc k_\ell$ such that $\cI = \{ S \mid |P_i \cap S| \leq k_i, i=1,2\dots \ell \}$.
The \emph{graphic matroid} is specified by an undirected graph $G = (V,E)$ where $\ground = E$ and $\cI = \{ S \subseteq E \mid S \text{ does not contain a cycle} \}$.
Matching constraints on subsets of edges of a graph{\textemdash}and more generally, $b$-matchings{\textemdash}form $2$-matchoids.
Moreover, a variety of scheduling constraints may be represented as $p$-extendible systems \citep{Mestre2006}.
Similarly to a partition matroid, the independence system given by subsets $P_1, P_2 \dotsc P_\ell 
\subseteq \ground$ (not necessarily a partition) and integers $k_1, k_2 \dots k_\ell$ such that $\cI = \{ S \mid |P_i \cap S| \leq k_i \ i=1, 2 \dots \ell \}$ is the intersection of $\ell$ matroids (and thus, also an $\ell$-matchoid).
Finally, we remark that an example of an independence system which is $p$-extendible but not a $p$-matchoid is a knapsack constraint in which the sizes of all the elements are between $1$ and $p$.
A visualization of the hierarchy of independence systems is presented in 
Figure~\ref{fig:heirarchy_examples}.
\begin{figure}
	\centering
	\includegraphics[width=0.6\linewidth]{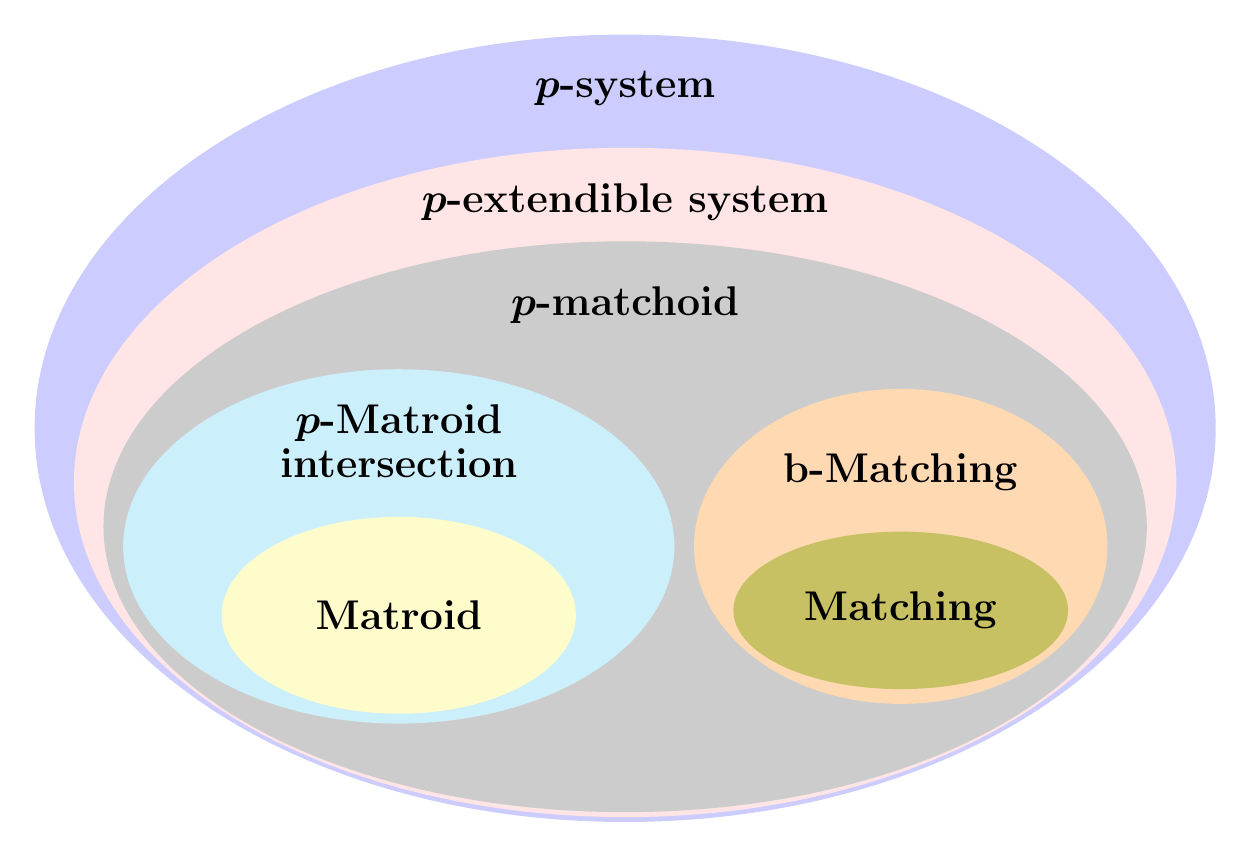}
	\caption{A visualization of the hierarchy of independence systems.}
	\label{fig:heirarchy_examples}
\end{figure}

The problem we are interested in in this paper is the mathematical program
\begin{equation} \label{eq:centralized_problem}
\max_{S \in \mathcal{I}} f(S) \enspace,
\end{equation}
where $f$ is non-negative and submodular on a ground set $\ground$ and $(\ground, \cI)$ is a $p$-extendible system.
We denote an arbitrary set achieving this maximum by \OPT.
We say that a set $S \in \cI$ is an \emph{$\alpha$-approximation} for some $\alpha \geq 1$ if 
\[ f(S) \geq \frac{1}{\alpha} \cdot f(\OPT) \enspace. \]
We assume that both the objective function $f$ and the independence system $\cI$ are accessed by 
the algorithms through oracles; that is, given a set $S$ there is a \emph{value oracle} which returns the value $f(S)$ and an \emph{independence oracle} which returns whether or not $S \in \cI$.
Our goal is to design an algorithm which makes few queries to the value and independence oracles and produces a set that is an $\alpha$-approximation for as small $\alpha \geq 1$ as possible.

In many practical applications, when the data is too large to be randomly accessed, offline algorithms are impractical.
The \emph{streaming model} of computation is an alternative computational paradigm for such settings. 
In this model, data points arrive in an arbitrary order and only a small amount of memory may be kept.
More formally, let the elements of the ground set be arbitrarily ordered as $\ground = \{u_1, \dotsc, u_n 
\}$, and let $\ground_t = \{ u_1, \dotsc, u_t \}$ be the first $t$ elements in this ordering.
An algorithm $\mathcal{A}$ for the model is presented with each element $u_1, u_2 \dots u_n$ in a 
sequential 
manner. The algorithm $\mathcal{A}$ maintains a set $M$ of the elements it currently keeps in its 
memory. Let us denote by $M_t$ the set $M$ immediately after the processing of $u_t$ by the algorithm. 
We note that $M_t$ must be a subset of $\ground_t$, and also that it must be a subset of $M_{t - 1} \cup 
\{u_t\}$ because once an element leaves $M$, it is forgotten and cannot be added to $M$ later on. 
Naturally, $\cA$ is allowed to query the value and independence oracle queries only with respect to 
subsets $S$ of the current $M$, because $M$ includes all the elements kept in $\cA$'s memory.

To be an $\alpha$-approximation algorithm, the algorithm $\mathcal{A}$ must be able---after viewing every element $e_t$---to produce a set $S_t \subseteq M_t$ that 
is independent ($S_t \in \cI$) and is an $\alpha$-approximation, i.e.,
\[ f(S_t) \geq \frac{1}{\alpha} \cdot \max_{\substack{S \subseteq \ground_t \\ S \in \cI }} f(S) \enspace. \]
The performance of a streaming algorithm is judged based on its approximation ratio $\alpha$, its 
update 
cost (which is the number of evaluation and independence oracle queries it makes after viewing each element), and the memory size $\max_t |M_t|$. 
In this paper, we consider streaming algorithms only for the case in which the constraint is defined by 
a $p$-matchoid, and we denote by $m$ the number of matroids used to defined the $p$-matchoid.
As is standard in the literature, we assume the streaming algorithm has access to an independence 
oracle for each of the $m$ defining matroids when considering a $p$-matchoid.

In the context of streaming algorithms, given an element $u_i \in N$ 
and sets $S, T \subseteq \cN$, we use the shorthands $f(u_i : S) = f(u_i \mid S \cap \{u_1, u_2, \dotsc, 
u_{i-1}\})$ and $f(T : S) = \sum_{u \in T} f(u : S)$. Intuitively, $f(u : S)$ is the marginal contribution of 
$u$ with respect to the part of $S$ that arrived before $u$ itself.

\section{Main Results} \label{sec:main_results}

In this section, we present two algorithms which use the subsampling technique for 
constrained 
submodular maximization.
The subsampling technique is simple: the algorithm only considers a random subset 
$\ground' \subseteq 
\ground$ of the 
ground set, where each element appears independently with
some probability $q$, which is determined solely by the complexity of the constraint set $p$.
\algrd is proposed for the usual offline setting where random access to the data is assumed 
and
\AlgSampling is proposed for the streaming setting, where elements arrive one at a time and 
only a 
small dataset is maintained in memory.
Our theoretical guarantees on the performance of these algorithms are given below.

\begin{theorem} \label{thm:centralized_alg}
	When $q = (p+1)^{-1}$, \algrd achieves a $\frac{(p+1)^2}{p}$-approximation ratio for the problem of 
	maximizing a non-negative submodular function $f$ subject to a $p$-extendible system.
	Moreover, this approximation ratio improves to $p+1$ when $f$ is monotone.
	When $f$ is linear and $q = p^{-1}$, \algrd achieves a further improved $p$-approximation ratio. 
	In all cases, \algrd uses in expectation $O(nk/p)$ calls to the evaluation and independence oracles.
\end{theorem}

\begin{theorem} \label{thm:streaming_alg}
	When $c = 1$ and $q = \left( (1+c)p + 1 \right)^{-1}$, \AlgSampling achieves an approximation ratio 
	of at most $(2p + 2\sqrt{p(p+1)} + 1) = 4p + 2 - o(1)$ for maximizing a non-negative submodular function $f$ 
	subject to a $p$-matchoid system $(\ground, \cI)$ in the streaming setting.
	When $c= \sqrt{1 + 1/p}$ and $f$ is monotone, \AlgSampling achieves an improved approximation 
	ratio of at most $4p$.
	In both cases, \AlgSampling requires $O(k)$ memory and $O(km/p)$ evaluation and 
	independence oracle queries in expectation when processing each arriving element.
\end{theorem}

We remark that the $(p+1)^{2}/p = p + 2 + o(1)$ approximation ratio of \algrd is nearly tight, due 
to 
Theorem 4 of~\citep{FHK17},
which shows that no randomized algorithm can achieve an approximation better than $p + 
1/2$ for the problem using 
polynomially many queries to the evaluation and independence oracles.
Although \AlgSampling has an approximation guarantee which is worse than \algrd and 
applies to a subclass of constraints, its approximation ratio is currently the best
among known algorithms for the problem setting.
Moreover, there is no known non-trivial streaming algorithm for maximizing even a monotone 
submodular function subject to a $p$-extendible constraint.
The difference between the guarantees of \algrd and \AlgSampling is of course aligned with the 
intuitive expectation that the  
streaming setting should be more challenging than the offline setting.

In addition to its simplicity, one of the more attractive aspects of the subsampling technique 
is the way 
in which it can be analyzed to provide approximation guarantees for both the monotone and 
non-monotone settings in a unified manner.
The main technical result which makes this possible is due to \cite{BFNS14}.
\begin{lemma}[Lemma 2.2 of~\citep{BFNS14}] \label{lem:buchbinder}
	Let $g\colon 2^\ground \to \nnR$ be a non-negative submodular function, and let $B$ be a 
	random 
	subset of $\ground$ containing every element of $\ground$ with probability at most $q$ 
	(not 
	necessarily 
	independently). Then,
	$\bE[g(B)] \geq (1 - q) \cdot g(\varnothing)$.
\end{lemma}
Suppose that $S$ is the random set returned by an algorithm.
When using the subsampling technique, this lemma provides a nontrivial lower bound on the 
term 
$\bE [f(S \cup \OPT)]$.
In particular, we can apply the lemma to the function defined by $g(T) = f(T \cup \OPT) \ 
\forall \ T \subseteq \ground$, whose
submodularity and non-negativity is guaranteed by the same conditions on $f$.
The subsampling technique implies that each element of $\ground$ appears in $S$ with 
probability at 
most $q$, and therefore,
\[ \bE[ f(S \cup \OPT) ] = \bE[g(S)] \geq (1 - q) \cdot g(\varnothing) = (1-q) \cdot f(\OPT) 
\enspace. \]
This allows for a basically interchangeable step when lower bounding the term $\bE [f(S \cup 
\OPT)]$ in 
the algorithmic analysis for monotone and non-monotone functions.
If $f$ is monotone, then we have $\bE [f(S \cup \OPT)] \geq f(\OPT)$ by monotonicity; 
otherwise, we 
invoke the subsampling technique to obtain $\bE [f(S \cup \OPT)] \geq (1 - q) \cdot f(\OPT)$ 
by 
Lemma~\ref{lem:buchbinder}.
Note that Lemma~\ref{lem:buchbinder} alone is not enough to guarantee any approximation 
factor.
Rather, Lemma~\ref{lem:buchbinder} shows that
\emph{subsampling} algorithms for monotone optimization can be converted into 
subsampling 
algorithms for non-monotone optimization with a controlled loss in the approximation ratio.
The powerful{\textemdash}and arguably, shocking{\textemdash}result is that this controlled 
loss in 
approximation yields nearly optimal approximation ratios in the offline setting.

Our subsampling technique is different from that proposed by 
\citep{Mirzasoleiman15}, which was 
developed for  the problem of maximizing a monotone submodular function under a cardinality 
constraint.
The subsampling technique of \citep{Mirzasoleiman15} works by sampling a new subset at each 
iteration and greedily choosing the element in the random subset with highest marginal gain.
The analysis works by guaranteeing that if the sample size is large enough, then it is likely that an 
element with sufficiently large marginal gain is chosen by the algorithm.
On the other hand, our subsampling technique in \algrd is quite different in that it requires only one 
subsampling of the ground set at the beginning of the algorithm.
Moreover, our proposed subsampling technique admits guarantees in more general 
problem settings, including non-monotone objectives and $p$-extendible system constraints.

\subsection{Offline Algorithm: \algrd} \label{sec:centralized_alg}

In this section, we present \algrd, a subsampling algorithm for the offline setting.
The idea is simple: first independently sample elements to obtain a subsampled ground set, then 
run the vanilla greedy algorithm.
We present \algrd as Algorithm~\ref{alg:sample_greedy} here.

\IncMargin{0.5em}
\SetKwIF{With}{OtherwiseWith}{Otherwise}{with}{do}{otherwise with}{otherwise}{}
\begin{algorithm}[H]
	\SetAlgoLined
	\SetInd{0.5em}{0.1em}
	\Indp
	\caption{\algrd$(q, f, \ground, \cI)$} \label{alg:sample_greedy}
	\DontPrintSemicolon
	Let $\cN' \leftarrow \varnothing$ and $S \leftarrow \varnothing$.\\
	\For{each $u \in \cN$}{
		\With{probability $q$\label{line:sample}}{Add $u$ to $\cN'$.}
	}
	\While{there exists $u \in \cN'$ such that $S + u \in \cI$ and $ f(u \mid S) > 0$\label{line:test_sg}}{
		Let $u \in \cN'$ be the element of this kind maximizing $ f(u \mid S)$.\\
		Add $u$ to $S$.
	}
	\Return{$S$}.
\end{algorithm}
\DecMargin{0.5em}

For analysis purposes, we introduce an auxiliary algorithm, Algorithm~\ref{alg:equiv_sample_greedy}.
While  \algrd first independently samples elements from the ground set $\ground$ and then runs a 
greedy maximization, Algorithm~\ref{alg:equiv_sample_greedy} runs a greedy 
maximization over the entire ground set and independently samples the greedily chosen element at 
each iteration.
We will show that both algorithms produce the same distribution over their output sets $S \subseteq 
\ground$ because the sampling of elements is independent from the greedy maximization. 
Thus, approximation guarantees obtained for Algorithm~\ref{alg:equiv_sample_greedy} also hold for 
\algrd.

\IncMargin{0.5em}
\begin{algorithm}[htb!]
	\caption{Equivalent-\algrd$(q, f, \ground, \cI)$} \label{alg:equiv_sample_greedy}
	\SetAlgoLined
	\SetInd{0.5em}{0.1em}
	\Indp
	\DontPrintSemicolon
	Let $\cN' \gets \cN$, $S \gets \varnothing$, $O \gets \OPT$, and $O_u \gets \varnothing$, $S_u 
	\gets 
	\varnothing$ for each $u \in \ground$.\\
	\While{there exists an element $u \in \cN'$ such that $S + u \in \cI$ and $ f(u \mid S) > 
	0$\label{line:while_sg_equiv}}
	{
		Let $u \in \cN'$ be the element of this kind maximizing $ f(u \mid S)$, and let $S_u \gets S$. 
		\label{line:chosen_u}\\
		\With{probability $q$\label{line:sample_equiv}}
		{
			Add $u$ to $S$ and $O$.\\
			Let $O_u \subseteq O \setminus S$ be the smallest set such that $O \setminus O_u \in \cI$.
		}
		\Otherwise
		{
			\lIf{$u \in O$\label{line:if_u_in}}{let $O_u \gets \{u\}$.}
			\lElse{let $O_u \gets \varnothing$.}
		}
		Remove the elements of $O_u$ from $O$.\\
		Remove $u$ from $\cN'$.
	}
	\Return{$S$.}
\end{algorithm}
\DecMargin{0.5em}

As in \algrd, $S$ is the current solution to which elements are incrementally added.
Algorithm~\ref{alg:equiv_sample_greedy} also maintains several auxiliary sets for analysis purposes, 
such as $O$ and the sets $O_u$, $S_u$ for each element $u$ in $\ground$.
We use these sets in the analysis of the algorithm below, however, they do not affect the output distribution (or other 
behavior) of Algorithm~\ref{alg:equiv_sample_greedy}. 

We now establish the equivalence between \algrd and Algorithm~\ref{alg:equiv_sample_greedy}.
It is important, however, to note that for the equivalence to hold we must make some technical 
assumption about the tie-breaking rule used by the two algorithms. In the proof below we assume 
that this tie-breaking rule is based on an ordering of the ground set. In other words, suppose that 
the ground set is (arbitrarily) ordered  $\ground = \{ u_1, u_2 \dots u_n 
\}$ and that both algorithms break ties in favor of the element appearing earlier in this 
order.\footnote{There 
are of course other natural ways in which the tie-breaking rule can be defined, and the proof can be 
made to work with many of them. For example, the proof goes through if the tie-breaking rule picks 
at every step a uniformly random element out of all the elements that can be picked at this point and 
maximize the marginal gain.}

\begin{lemma} \label{lemma:offline_equiv}
	Let $S^1$ and $S^2$ be the random sets produced by \algrd and 
	Algorithm~\ref{alg:equiv_sample_greedy}, respectively.
	These random sets have the same probability distribution.
\end{lemma}
\begin{proof}{Proof} 
	Let $u_1^1, u_2^1 \dots u_{k_1}^1$ be the random 
	sequence of elements chosen by \algrd, and let 
	$u_1^2, u_2^2 \dots \allowbreak u_{k_2}^2$ be the random sequence of elements chosen by 
	Algorithm~\ref{alg:equiv_sample_greedy}.
	We prove that these two random sequences have the same distribution, which implies the lemma 
	since $S^1 = \{ u_1^1, u_2^1 \dots u_{k_1}^1 \}$ and
	$S^2 = \{ u_1^2, u_2^2 \dots u_{k_2}^2 \}$.
	
	Observe that both \algrd and Algorithm~\ref{alg:equiv_sample_greedy} require at most $n$ random (biased) 
	bits, as each element 
	from the ground set is sampled at most once.
	This occurs at Line~\ref{line:sample} in \algrd and at Line~\ref{line:sample_equiv} in 
	Algorithm~\ref{alg:equiv_sample_greedy}.
	For each element $u_i \in \ground$, let $b_i \in \{ 0 , 1 \}$ be the corresponding random bit that 
	takes 
	the value $1$ if $u_i$ is accepted in the sampling step and $0$ if $u_i$ is rejected.
	Because the random sampling at Line~\ref{line:sample} in \algrd and at 
	Line~\ref{line:sample_equiv} in 
	Algorithm~\ref{alg:equiv_sample_greedy} are independent of the algorithms' previous and current 
	states, the bits $b_1, 
	b_2, 
	\dots b_n \in \{0,1\}^n$ are independent.
	Thus, we may assume for the sake of the proof that these bits are chosen before the execution of 
	the algorithms. Note also that, by definition of the 
	sampling probabilities in the algorithms, each bit $b_i$ takes the value $1$ with probability $q$ 
	and 
	$0$ with probability $1 - q$.
	
	We show below that conditioned on any fixed realization of the random bits $b_1, b_2 \dots 
	b_n$, the two 
	sequences $u_1^1, u_2^1 \dots \allowbreak u_{k_1}^1$ and $u_1^2, u_2^2 \dots u_{k_2}^2$, 
	produced by \algrd 
	and Algorithm~\ref{alg:equiv_sample_greedy} are the same. We do this by induction on the index of 
	the element in the 
	sequence. In other words, let us denote for every $i \in \{1,2\}$ and $0 \leq j \leq k_i$,
	\[
	S_j^i = \{ u_1^i, u_2^i \dots u_j^i \}
	\enspace.
	\]
	Then, we prove below by induction that $j \leq k_1$ if and only if $j \leq k_2$, and that $S_j^1 
	= S_j^2$ whenever $j \leq k_1$. Before starting the proof by induction, however, let us note that it 
	implies the lemma by the law of total probability, i.e., summing over all fixed realizations of the bits 
	$b_1, b_2, \dots b_n$.

	Clearly, $0 \leq k_1, k_2$ and $S_0^1 = S_0^2 = \varnothing$, which establishes the base of the 
	induction. Assume now that the induction hypothesis holds for some $j - 1 \geq 0$, and let us 
	prove it for $j$. If $j - 1 > k_1$, then the induction hypothesis implies $j > j -1 > \max\{k_1, k_2\}$, and 
	there is nothing left to prove. Thus, let us assume $j - 1 \leq k_1$, which  by the induction 
	hypothesis implies $j - 1 \leq k_2$ and $S_{j-1}^1 = S_{j-1}^2$. We denote the common value of the 
	last two sets by $A$,  i.e., $A = S_{j-1}^1 = S_{j-1}^2$. We note that {\algrd} chooses 
	its 
	next element from the set of 
	elements obeying $u \in \ground'$, $A + u \in \cI$ and $f(u \mid A)$ is positive 
	and maximal. Since $u \in \ground'$ if and only if $b_u = 1$, this implies that \algrd chooses $u_j$ 
	to be the earliest element in the tie-breaking order from the set
	\begin{equation}\label{eq:u_pick}
	\argmax_{\substack{ u_i \in \ground \\ b_i = 1 \\ A + u_i \in \cI }} \{ f(u \mid A) \mid  f(u \mid A) > 0 \} 
	\enspace,
	\end{equation}
	whenever this set is non-empty. Otherwise, if the set is empty, then \algrd terminates.
	Now, consider the element which is chosen next by Algorithm~\ref{alg:equiv_sample_greedy}.
	This element is chosen from the set of elements obeying $u \in \ground'$, $A + u \in \cI$ and $f(u 
	\mid 
	A)$ is positive and maximal. However, if $b_u = 0$, then the element is rejected and the algorithm 
	continues to the next element in the set (according to the tie-breaking order).
	Thus, Algorithm~\ref{alg:equiv_sample_greedy} also chooses the earliest element from the 
	set~\eqref{eq:u_pick} if this set 
	is non-empty; and terminates if the set is empty. Therefore, if the set~\eqref{eq:u_pick} is empty, 
	then $k_1 = k_2 = j - 1 < j$, and if it is not empty, then $j \leq \min\{k_1, k_2\}$ and $S_j^1 = S_j^2$. 
	In either case the induction step holds, which completes the proof by induction.
	\qedhere
\end{proof}

Now we return to the task of analyzing Algorithm~\ref{alg:equiv_sample_greedy}.
Let us explain the intuition behind the auxiliary sets $O$, $O_u$ and $S_u$ appearing in this algorithm.
The set $O$ begins as an optimal solution and is updated throughout the algorithm to maintain 
independence.
We say that an element $u \in \ground$ is \emph{considered} by 
Algorithm~\ref{alg:equiv_sample_greedy} if it is chosen in Line~\ref{line:chosen_u} at some iteration.
An element $u \in \ground$ is considered at most once, and perhaps not at all.
If $u$ is considered, then $S_u$ is the solution at the iteration in which this happens and $O_u$ is the 
subset of
$O$ which must be removed at this iteration to maintain independence and a few other properties; otherwise, if $u$ is 
not considered, then $S_u$ and $O_u$ are empty.
More formally, Lemma~\ref{lem:structural_properties} gives several key properties of these auxiliary 
sets.

\begin{lemma} \label{lem:structural_properties}
	The following three properties hold throughout Algorithm~\ref{alg:equiv_sample_greedy}.
	\begin{enumerate}[label=\textit{(P\arabic*)},leftmargin=1.5cm]
		\item $O$ is an independent set. \label{itm:O_ind}
		\item Every element of $S$ is an element of $O$. \label{itm:S_sub_O}
		\item Every element of $O \setminus S$ is an element not yet considered by 
		Algorithm~\ref{alg:equiv_sample_greedy}. \label{itm:not_considered}
	\end{enumerate}
\end{lemma}
\begin{proof}{Proof}
	It is clear that all three properties hold at the beginning of the algorithm when $O = \OPT$ and $S = 
	\varnothing$.
	Let us now show that these properties are maintained throughout 
	Algorithm~\ref{alg:equiv_sample_greedy} by induction over the 
	iterations of the algorithm.
	Suppose that $u$ is the element being considered at some iteration, $S$ and $O$ are the sets at the 
	beginning of this iteration (satisfying properties \ref{itm:O_ind}, \ref{itm:S_sub_O}, and 
	\ref{itm:not_considered} by the inductive hypothesis) and $S'$ and $O'$ are these sets at the end of the 
	iteration.
	
	\paragraph{Case 1} Suppose that $u$ is chosen to be added to the current solution.
	By \ref{itm:O_ind} and \ref{itm:S_sub_O}, $O$ is an extension of $S$, and by 
	Line~\ref{line:chosen_u}, $S \cup \{u\} 
	\in \cI$. Thus, the algorithm is able to find a set $O_u \subseteq (O \cup \{u\}) \setminus (S \cup \{u\})$ such that $O' = \left( O 
	\cup \{u\} \right) \setminus O_u \in \cI$ ($O_u = O \setminus (S \cup \{u\})$ is one possible option). 
	Thus, $O'$ remains independent, and \ref{itm:O_ind} is maintained.
	We have that $S \subseteq O$ by \ref{itm:S_sub_O}. The element $u$ is added to both $S$ and $O$, 
	and the only elements which are removed from $O$ are not in $S$.
	Thus, $S' \subseteq O'$ and so \ref{itm:S_sub_O} is maintained at well.
	By \ref{itm:not_considered}, all elements in $O \setminus S$ had not yet been considered by 
	Algorithm~\ref{alg:equiv_sample_greedy} at the beginning of the iteration. 
	The only element which is considered in this iteration is $u$, and it is in both $O'$ and $S'$, so 
	it is not in $O' \setminus S'$. Hence, the elements in $O' 	\setminus S'$ have still not been considered 
	by Algorithm~\ref{alg:equiv_sample_greedy} at the end of iteration, and \ref{itm:not_considered} is 
	maintained.
	
	\paragraph{Case 2} Consider now the case that $u$ is not added to the current solution. 
	In this case, $O' \subseteq O$, and therefore, it remains independent and \ref{itm:O_ind} is maintained.
	Because $S \subseteq O$ by \ref{itm:S_sub_O}, $u \notin S' = S$, and the only element which is 
	possibly removed from $O$ is $u$, we have that all elements in $S'$ belong to $O'$, and thus,
	\ref{itm:S_sub_O} is maintained.
	Finally, by \ref{itm:not_considered}, none of the elements in $O \setminus S$ were considered prior to this 
	iteration. By Line~\ref{line:if_u_in}, $u$ does not appear in $O' \setminus S'$. Since $u$ is the only element considered during the current iteration and $O' \setminus S' = (O \setminus S) \setminus \{u\}$, property  
	\ref{itm:not_considered} is maintained. 
\end{proof}

Throughout the remainder of this section, every expression involving $S$ or $O$ is assumed to refer to 
the final values of these sets. 
The following lemma provides a deterministic lower bound on $f(S)$.
Intuitively, this lemma follows from the observation that, when an element $u$ is considered by 
Algorithm~\ref{alg:equiv_sample_greedy}, its marginal contribution is at least as large as the marginal 
contribution of any element of $O \setminus S$. 

\begin{lemma} \label{lem:lower_bound_S}
	$f(S) \geq f(S \cup \OPT) - \sum \limits_{u \in \cN} |O_u \setminus S | \cdot f(u \mid S_u)$.
\end{lemma}
\begin{proof}{Proof}
	We first show that $f(S) \geq f(O)$, then we lower bound $f(O)$ to complete the proof. 
	By \ref{itm:O_ind} and \ref{itm:S_sub_O} of Lemma~\ref{lem:structural_properties}, we have $O \in \cI$ 
	and $S \subseteq O$, and thus, $S + v \in 
	\mathcal{I}$ for all $v \in O \setminus S$ because $(\cN, \cI)$ is an independence system. 
	Consequently, the termination condition of Algorithm~\ref{alg:equiv_sample_greedy} guarantees that 
	$\Delta f(v \mid S) \leq 0$ for all $v \in O \setminus S$. 
	To use these observations, let us denote the elements of $O \setminus S$ by $v_1, v_2, \dotsc, v_{|O 
	\setminus S|}$ in an arbitrary order. Then
	\[
	f(O) = f(S) + \sum \limits_{i=1}^{ |O \setminus S|} f \left( v_i \mid S \cup \{v_1, \dots , v_{i - 1} \} 
	\right) \leq f(S) + \sum \limits_{i=1}^{ |O \setminus S|} f \left( v_i \mid S \right) \leq f(S)\enspace,
	\]
	where the first inequality follows by the submodularity of $f$, and the second inequality follows from 
	the termination condition.
	
	It remains to prove the lower bound on $f(O)$. 
	By definition, $O$ is the set obtained from $\OPT$ after the elements of $\cup_{u \in \cN} O_u$ are 
	removed and the elements of $S$ are added. 
	Additionally, an element that is removed from $O$ is never added to $O$ again, unless it becomes a 
	part of $S$. 
	This implies that the sets $\left\{ O_u \setminus S \right\}_{u \in \cN}$ are disjoint, and that $O$ can also 
	be written as
	\begin{equation} \label{eq:rewriting_O}
	O = \left( S \cup \OPT \right) \setminus \cup_{u \in \ground} \left(O_u \setminus S \right)
	\enspace.
	\end{equation}
	Denoting the elements of $\ground$ by $u_1,u_2, \dotsc, u_n$ in an arbitrary order, and using 
	the 
	above, we get
	\begin{align*}
	f(O)
	&= f(S \cup \OPT) - \sum \limits_{i=1}^n f \left( O_{u_i} \setminus S \mid (S \cup \OPT) \setminus 
	\cup_{1 
	\leq 
	j \leq i} ( O_{u_j} \setminus S) \right) &\text{(Equality~\eqref{eq:rewriting_O})} \\
	&\geq f(S \cup \OPT) - \sum \limits_{i=1}^n  f(O_{u_i} \setminus S \mid S_{u_i}) \\
	&\geq f(S \cup \OPT) - \sum \limits_{i=1}^n \sum \limits_{v \in O_{u_i} \setminus S}  f(v \mid S_{u_i}) \\
	&= f(S \cup \OPT) - \sum \limits_{u \in \ground} \sum \limits_{v \in O_u \setminus S}  f(v \mid S_{u})
	\enspace,
	\end{align*}
	where the first inequality follows from the submodularity of $f$ because $S_{u_i} \subseteq S 
	\subseteq (S \cup \OPT) \setminus \cup_{u \in \ground} \left(O_u \setminus S \right)$, and the second 
	inequality follows from the submodularity of $f$ as well.
	
	To complete the proof of the lemma, we need one more observation. 
	Consider an element $u$ for which $O_u$ is not empty. 
	Since $O_u$ is not empty, we know that $u$ was considered by the algorithm at some iteration. 
	Moreover, every element of $O_u$ was also a possible candidate for consideration at this iteration, and 
	thus, it must be the case that $u$ was selected for consideration because its marginal contribution 
	with respect to $S_u$ is at least as large as the marginal contribution of every element of $O_u$. 
	Plugging this observation into the last inequality, we get the following desired lower bound on $f(O)$.
	\begin{align*}
	f(O)
	\geq{} &
	f(S \cup \OPT) - \sum \limits_{u \in \cN} \sum \limits_{v \in O_u \setminus S}  f(v \mid S_{u})  \\
	\geq{} &
	f(S \cup \OPT) - \sum \limits_{u \in \cN} \sum \limits_{v \in O_u \setminus S}  f(u \mid S_{u}) \\
	={} &
	f(S \cup \OPT) - \sum \limits_{u \in \cN} |O_u \setminus S| \cdot f(u \mid S_{u})
	\enspace. 
	\qedhere
	\end{align*}
\end{proof}

While the previous lemma was true deterministically, the next two lemmata are statements about 
expected values. At this point, it is convenient to define some random variables. For every element $u 
\in 
\ground$, let $X_u$ be an indicator for the event that $u$ is considered by 
Algorithm~\ref{alg:equiv_sample_greedy} in one of its iterations.

\begin{lemma} \label{lem:link}
	Suppose that the sampling probability is $q = \frac{1}{p+1}$. Then,
	for every element $u \in \cN$,
	\begin{equation}\label{eq:linking_inequality}
	\E{ |O_u \setminus S | \cdot  f(u \mid S_u) } \leq \frac{p}{p+1} \cdot \E{X_u  f(u \mid S_u)}\enspace.
	\end{equation}
\end{lemma}
\begin{proof}{Proof}
	Let $\mathcal{E}_u$ be an arbitrary event specifying all 
	random decisions made by Algorithm~\ref{alg:equiv_sample_greedy} up until the iteration in which it 
	considers $u$ if $u$ is considered, or all random decisions made by 
	Algorithm~\ref{alg:equiv_sample_greedy} throughout its execution if it never considers $u$. 
	By the law of total probability, since these events are disjoint, it is enough to prove 
	Inequality~(\ref{eq:linking_inequality}) conditioned on every such event $\mathcal{E}_u$.
	If $\mathcal{E}_u$ implies that $u$ is not considered, then both $|O_u|$ and $X_u$ are $0$ 
	conditioned on $\mathcal{E}_u$, and thus, the inequality holds as an equality. 
	Thus, we may assume in the rest of the proof that $\mathcal{E}_u$ implies that $u$ is considered by 
	Algorithm~\ref{alg:equiv_sample_greedy}. 
	Notice that conditioned on $\cE_u$ the set $S_u$ is deterministic and $X_u$ takes the value $1$. 
	Denoting the deterministic value of $S_u$ conditioned on $\cE_u$ by $S'_u$, 
	Inequality~(\ref{eq:linking_inequality}) reduces to 
	\[\E{|O_u \setminus S | \mid \mathcal{E}_u} \cdot f(u \mid S'_u) \leq \frac{p}{p+1} \cdot f(u \mid S'_u) 
	\enspace. 
	\]
	Since $u$ is being considered, it must hold that $ f(u \mid S'_u) > 0$, and thus, it suffices to show 
	that $\E{|O_u \setminus S | \mid \mathcal{E}_u} \leq \frac{p}{p+1} $. 
	There are now two cases to consider.
	
	\paragraph{Case 1}
	If $\mathcal{E}_u$ implies that $u \in O$ at the beginning of the iteration in which 
	Algorithm~\ref{alg:equiv_sample_greedy} considers $u$, then $O_u = \varnothing$ if $u$ is added to 
	$S$ 
	and $ O_u = \{u\}$ if $u$ is not added to $S$. 
	As $u$ is added to $S$ with probability $\frac{1}{p+1}$, this gives
	\[\E{|O_u \setminus S | \mid \mathcal{E}_u} = \frac{1}{p+1} \cdot | \varnothing | + \left( 1- 
	\frac{1}{p+1}\right) \cdot |\{u\}| = \frac{p}{p+1} \enspace, \]
	and we are done. 
	
	\paragraph{Case 2}
	Consider now the case that $\mathcal{E}_u$ implies that $u \not \in O$ at the beginning of the iteration 
	in which Algorithm~\ref{alg:equiv_sample_greedy} considers $u$. 
	Because $u$ is being considered, $S \cup \{u\} $ is independent and by \ref{itm:S_sub_O} and 
	\ref{itm:O_ind} of Lemma~\ref{lem:structural_properties},  $O$ is an extension of $S$.
	This implies that $O_u$ has size at most $p$ because $(\ground, \cI)$ is $p$-extendible.
	As $u$ is added to $S$ with probability $\frac{1}{p+1}$, we get in this case
	\[ 
	\E{|O_u \setminus S | \mid \mathcal{E}_u} \leq \frac{1}{p+1} \cdot p + \left( 1- \frac{1}{p+1}\right) \cdot 
	|\varnothing| = \frac{p}{p+1} \enspace. 
	\qedhere
	 \] 	
\end{proof}

The next lemma relates the expected marginal gains of considered elements in individual 
iterations to the final expected value of $f(S)$ produced by the algorithm.

\begin{lemma} \label{lem:exp_S}
	$q \cdot \sum \limits_{u \in \cN} \E{X_u  f(u \mid S_u)} \leq \E{f(S)}$.
\end{lemma}
\begin{proof}{Proof}
	For each $u \in \cN$, let $G_u$ be a random variable whose value is equal to the increase in the 
	value of $S$ when $u$ is added to $S$ by Algorithm~\ref{alg:equiv_sample_greedy}.
	If $u$ is never added to $S$ by Algorithm~\ref{alg:equiv_sample_greedy}, then the value of $G_u$ is 
	simply $0$. Clearly,
	\[f(S) = f(\varnothing) + \sum \limits_{u \in \cN} G_u \geq \sum \limits_{u \in \cN} G_u \enspace,\]
	where $f(\varnothing) \geq 0$ follows from non-negativity of $f$.
	By the linearity of expectation, it suffices to show that
	\begin{equation} \label{eq:exp_G}
	\E{G_u} = q \cdot \E{X_u f(u \mid S_u)} \enspace.
	\end{equation}
	As in the proof of Lemma~\ref{lem:link}, let $\mathcal{E}_u$ be an arbitrary event specifying all random 
	decisions made by Algorithm~\ref{alg:equiv_sample_greedy} up until the iteration in which it considers 
	$u$ if $u$ is considered, or all random decisions made by Algorithm~\ref{alg:equiv_sample_greedy} 
	throughout its execution if $u$ is never considered. 
	By the law of total probability, since these events are disjoint, it is enough to prove that 
	Equality~(\ref{eq:exp_G}) holds when conditioned on every such event $\mathcal{E}_u$. 
	If $\mathcal{E}_u$ is an event that implies that Algorithm~\ref{alg:equiv_sample_greedy} does not 
	consider $u$, then, by conditioning on $\mathcal{E}_u$, we obtain
	\[\E{G_u \mid \mathcal{E}_u}  = 0 = q \cdot \E{0 \cdot  f(u \mid S_u) \mid \mathcal{E}_u} = q \cdot \E{X_u 
	f(u \mid S_u) \mid \mathcal{E}_u} \enspace.\]
	On the other hand, if $\mathcal{E}_u$ implies that Algorithm~\ref{alg:equiv_sample_greedy} does 
	consider $u$, then we observe that $S_u$ is a deterministic set given $\cE_u$. 
	Denoting this set by $S'_u$, we obtain
	\[\E{G_u \mid \mathcal{E}_u} = \Pr\left[ u \in S \mid \cE_u\right] \cdot f(u \mid S'_u) = q \cdot f(u \mid 
	S'_u) = 
	q \cdot \E{X_u  f(u \mid S_u) \mid \mathcal{E}_u} \enspace,\]
	where the second equality holds since an element considered by 
	Algorithm~\ref{alg:equiv_sample_greedy} is added to $S$ with probability $q$. 
\end{proof}

With these lemmata, we are now ready to prove Theorem~\ref{thm:centralized_alg} in the case of 
submodular (not necessarily linear) objectives.

\subsubsection{Proof of Theorem~\ref{thm:centralized_alg}, Submodular Objectives.}
We prove the first part of Theorem~\ref{thm:centralized_alg} concerning submodular functions in this 
section. The improved approximation guarantees for linear functions requires a few tighter lemmas, and 
so we prove the case of linear functions in the next section (Section~\ref{sec:linear-centralized-proof}).

Theorem~\ref{thm:centralized_alg}. 
As discussed earlier, Algorithms~\ref{alg:sample_greedy} and~\ref{alg:equiv_sample_greedy} have 
identical output distributions, and so it suffices to show that Algorithm~\ref{alg:equiv_sample_greedy} 
achieves the desired approximation ratios. 
Note that $q = \frac{1}{p+1}$, and therefore,
\begin{align*}
\E{f(S)} &\geq \E{f(S \cup \OPT)} - \sum \limits_{u \in \cN} \E{|O_u \setminus S| \cdot f(u \mid S_u)} 
&\text{(Lemma~\ref{lem:lower_bound_S})} \\
&\geq \E{f(S \cup \OPT)} - \frac{p}{p+1} \sum \limits_{u \in \cN} \E{X_u \cdot f(u \mid S_u)} 
&\text{(Lemma~\ref{lem:link})} \\
&\geq \E{f(S \cup \OPT)} - p \cdot \E{f(S)}\enspace. &\text{(Lemma~\ref{lem:exp_S})}
\end{align*}
If $f$ is monotone, then by monotonicity we have that $\E{f(S \cup \OPT)} \geq f( \OPT)$. Substituting 
this 
in the expression above yields 
\[ \E{f(S)} \geq f(\OPT) - p \cdot \E{f(S)} \]
and rearranging this expression yields the desired approximation ratio of $p+1$.
Suppose now that $f$ is non-monotone. Note that each element appears in $S$ with probability at most $q = 
\frac{1}{p+1}$, and hence, by Lemma~\ref{lem:buchbinder}, we have that 
$\E{f(S \cup \OPT)} \geq \left( 1 - \frac{1}{p+1} \right) f( \OPT)$. Substituting this into the inequalities 
above yields
\[ \E{f(S)} \geq \left( 1 - \frac{1}{p+1} \right)  f(\OPT) - p \cdot \E{f(S)} \enspace,  \]
and rearranging this expression yields the desired approximation ratio of $(p+1)^2/p$.

It remains to bound the number of oracle calls required by \algrd. Because $\E{ |\ground'| } = n \cdot q$, 
iterating over each 
$u \in \ground'$ and testing $S + u\in \cI$ and $f( u \mid S) > 0$ requires $O \left( nq \right)$ 
calls to the evaluation and independence oracle. Moreover, because $|S|$ increases at each iteration, the 
while loop (Line~\ref{line:test_sg} in Algorithm~\ref{alg:sample_greedy}) is repeated at most $k$ times. 
Using that $q = O(p^{-1})$, we have shown that Algorithm~\ref{alg:sample_greedy} requires $O(knq) = 
O\left(  nk/p \right)$ calls to evaluation and independence oracles in expectation.
\qedhere

\subsubsection{Proof of Theorem~\ref{thm:centralized_alg}, Linear Objectives.}\label{sec:linear-centralized-proof}
The method for proving the improved approximation guarantees for linear objectives uses essentially the same ideas as in the general submodular setting. 
However, further care is required to obtain the $p$-approximation guarantee.
In this section, we prove two lemmata which are 
analogous to Lemmata~\ref{lem:lower_bound_S} and~\ref{lem:link}, but tighter in the case of linear 
functions. 

We begin with the following lemma, which corresponds to Lemma~\ref{lem:lower_bound_S}. For every 
$u \in \ground$, let $Y_u$ be a random variable which takes the value $1$ if $u \in S$ and, in addition, 
$u$ does not belong to $O$ at the beginning of the iteration in which $u$ is considered. In every other 
case the value of $Y_u$ is $0$.

\begin{lemma} \label{lem:linear_lower_bound_S}
	$f(S) \geq f(\OPT) - \sum \limits_{u \in \ground} [|O_u| - Y_u] f(u)$.
\end{lemma}
\begin{proof}{Proof}
	The proof of Lemma~\ref{lem:lower_bound_S} begins by showing that $f(S) \geq f(O)$. This part of 
	the proof is of course still true. Thus, we only need to show 
	that
	\[	f(O)	\geq	f(\OPT) - \sum \limits_{u \in \ground} [|O_u| - Y_u] f(u)	\enspace.
	\]
	Recall that $O$ begins as equal to $\OPT$. Thus, to prove the last inequality it is enough to show 
	that the second term on its right hand side is an upper bound on the decrease in the value of $O$ 
	over time. In the rest of the proof we do this by showing that $[|O_u| - Y_u] f(u)$ is an upper bound 
	on the decrease in the value of $O$ in the iteration in which $u$ is considered, and is equal to $0$ 
	when $u$ is not considered at all. 
	
	Let us first consider the case that $u$ is not considered at all. In this case, by definition, $O_u = 
	\varnothing$ and $Y_u = 0$, which imply together $[|O_u| - Y_u] f(u) = 0 \cdot f(u) = 0$. Consider 
	now the case that $u$ is considered by Algorithm~\ref{alg:equiv_sample_greedy}. In this case, $O$ is 
	changed during the iteration in which $u$ is considered in two ways. First, the elements of $O_u$ 
	are removed from $O$, and second, $u$ is added to $O$ if it is added to $S$ and it does not 
	already belong to $O$. Thus, the decrease in the value of $O$ during this iteration can be written as
	\[	\sum_{v \in O_u} f(v) - Y_u \cdot f(u)	\enspace.	\]
	To see why this expression is upper bounded by $[|O_u| - Y_u] f(u)$, we recall that in the proof of 
	Lemma~\ref{lem:lower_bound_S} we showed that $f(v\mid S_u) \leq f(u\mid S_u)$ for every $v \in O_u$, 
	which implies, since $f$ is linear, that $f(v) \leq f(u)$ for every such element $v$. 
\end{proof}

We need one more lemma which corresponds to Lemma~\ref{lem:link}.

\begin{lemma} \label{lem:linear_link}
	Suppose that the sampling probability is $q = \frac{1}{p}$. Then,
	for every element $u \in \ground$,
	\begin{equation} \label{eq:linear_linking_inequality}
	\E{|O_u| - Y_u} \leq \frac{p - 1}{p} \cdot \E{X_u}\enspace.
	\end{equation}
\end{lemma}
\begin{proof}{Proof}
	As in the proof of Lemma~\ref{lem:link}, let $\mathcal{E}_u$ be an arbitrary event specifying all 
	random decisions made by Algorithm~\ref{alg:equiv_sample_greedy} up until the iteration in which it 
	considers $u$ if $u$ is considered, or all random decisions made by 
	Algorithm~\ref{alg:equiv_sample_greedy} throughout its execution if it never considers $u$. By the 
	law of total probability, since these events are disjoint, it is enough to prove 
	Inequality~(\ref{eq:linear_linking_inequality}) conditioned on every such event $\mathcal{E}_u$. If 
	$\mathcal{E}_u$ implies that $u$ is not considered, then $|O_u|$, $X_u$ and $Y_u$ are all $0$ 
	conditioned on $\mathcal{E}_u$, and thus, the inequality holds as an equality. Thus, we may assume 
	in the rest of the proof that $\mathcal{E}_u$ implies that $u$ is considered by 
	Algorithm~\ref{alg:equiv_sample_greedy}. Notice that, conditioned on $\cE_u$, $X_u$ takes the 
	value $1$. Hence, Inequality~(\ref{eq:linear_linking_inequality}) reduces to 
	\[ \E{|O_u| - Y_u\mid \mathcal{E}_u} \leq \frac{p - 1}{p} \enspace. \]
	There are now two cases to consider. 
	
	\paragraph{Case 1}
	The first case is that $\mathcal{E}_u$ implies that $u \in O$ at 
	the beginning of the iteration in which Algorithm~\ref{alg:equiv_sample_greedy} considers $u$. In 
	this case $Y_u = 0$, and in addition, $O_u$ is empty if $u$ is added to $S$, and is $\{u\}$ if $u$ is 
	not added to $S$. As $u$ is added to $S$ with probability $\frac{1}{p}$, this gives
	\[\E{|O_u| - Y_u \mid \mathcal{E}_u} \leq \frac{1}{p} \cdot | \varnothing | + \left( 1- \frac{1}{p}\right) 
	\cdot
	|\{u\}| = \frac{p-1}{p} \enspace, \]
	and we are done. 
	
	\paragraph{Case 2}
	Consider now the case that $\mathcal{E}_u$ implies that $u \not \in O$ at the 
	beginning of the iteration in which Algorithm~\ref{alg:equiv_sample_greedy} considers $u$. In this 
	case, if $u$ is not added to $S$, then we get $Y_u = 0$ and $O_u = \varnothing$. In contrast, if $u$ 
	is added to $S$, then $Y_u = 1$ by definition and $|O_u| \leq p$ as in the proof of Lemma~\ref{lem:link}. As $u$ is added to $S$ with probability 
	$\frac{1}{p}$, we get in this case
	\[ \E{|O_u| - Y_u | \mid \mathcal{E}_u} \leq \frac{1}{p} \cdot (p - 1) + \left( 1- \frac{1}{p}\right) \cdot
	|\varnothing| = \frac{p - 1}{p} \enspace.
	\qedhere
	\]
\end{proof}

We are now ready to prove the guarantee of Theorem~\ref{thm:centralized_alg} for linear objectives.
\begin{proof}{Proof of Theorem~\ref{thm:centralized_alg} for linear objectives.}
	We prove here that the approximation ratio guaranteed by Theorem~\ref{thm:centralized_alg} for 
	linear objectives is obtained by Algorithm~\ref{alg:sample_greedy} for $q = 1/p$. As discussed 
	earlier, Algorithms~\ref{alg:sample_greedy} and~\ref{alg:equiv_sample_greedy} have identical output 
	distributions, and so it suffices to show that Algorithm~\ref{alg:equiv_sample_greedy} achieves this 
	approximation ratio. Since we assume $q = \frac{1}{p}$,
	\begin{align*}
	\E{f(S)} &\geq f(\OPT) - \sum \limits_{u \in \ground} \bE[|O_u| - Y_u] f(u) 
	&\text{(Lemma~\ref{lem:linear_lower_bound_S})} \\
	&\geq f(\OPT) - \frac{p-1}{p}\sum \limits_{u \in \ground} \bE[X_u] f(u) 
	&\text{(Lemma~\ref{lem:linear_link})} 
	\\
	&\geq f(\OPT) - (p - 1) \E{f(S)}\enspace. &\text{(Lemma~\ref{lem:exp_S})}
	\end{align*}
	Rearranging the above inequality completes the proof, as the oracle complexity is unchanged.
\end{proof}

\subsection{Streaming Algorithm} \label{sec:streaming_alg}

In this section, we present \AlgSampling, a subsampling algorithm for the streaming setting.
\AlgSampling has two parameters: a sampling probability $q \in (0,1]$ and an acceptance parameter 
$c>0$.
At a given iteration $i = 1, \dotsc, n$, the arriving element $u_i$ is considered for exchange with probability 
$q$, and rejected without being considered for an exchange with probability $1 - q$. 
This step acts as an independent subsampling of elements in the stream, in an analogous manner to 
the subsampling in \algrd.
If the element $u_i$ is considered for exchange, a subroutine \ExchangeAlg 
\footnote{The subroutine {\ExchangeAlg} has appeared in a previous work \citep{CGQ15} as a method 
for exchanging in a $p$-matchoid.}
produces a set $U_i \subseteq S_i$ of low marginal contribution such that 
$\left( S \setminus U_i \right) \cup \{u_i\}$ is independent.
If the marginal contribution of adding $u_i$ to the current solution is large enough compared to the 
value of the elements of $U$, then $u$ is added to the solution and the elements of $U$ are removed. 
\AlgSampling and the subroutine \ExchangeAlg are presented here as Algorithms~\ref{alg:actual} and~\ref{alg:exchange_alg}, respectively.

\IncMargin{0.5em}
\SetKwIF{With}{OtherwiseWith}{Otherwise}{with}{do}{otherwise with}{otherwise}{end}
\begin{algorithm}[htb!]
	\caption{\AlgSampling$(q, f, \cI_1, \dotsc, \cI_m)$} \label{alg:actual}
		\SetAlgoLined
	\SetInd{0.5em}{0.1em}
	\Indp
	\DontPrintSemicolon
	Let $S_0 \gets \varnothing$.\\
	\For{every arriving element $u_i$}
	{
		Let $S_i \gets S_{i-1}$.\\
		\With{probability $q$\label{line:rand_consider}}
		{
			Let $U_i \gets {\ExchangeAlg}(S_{i-1}, u_i)$.\\
			\lIf{$f(u_i \mid S_{i-1}) \geq (1 + c) \cdot f(U_i : S_{i-1})$}
			{
				let $S_i \gets (S_{i-1} \setminus U_i) \cup \{ u_i \}$.
			}
		}
	}
	\Return{$S_n$}.
\end{algorithm}
\DecMargin{0.5em}

\IncMargin{0.5em}
\begin{algorithm}[htb!]
	\caption{{\ExchangeAlg} $(S, u)$}\label{alg:exchange_alg}
		\SetAlgoLined
	\SetInd{0.5em}{0.1em}
	\Indp
	\DontPrintSemicolon
	Let $U \gets \varnothing$.\\
	\For{$\ell = 1$ \KwTo $m$}
	{
		\If{$(S + u) \cap \cN_\ell \not \in \cI_\ell$}
		{
			Let $X_\ell \gets \{x \in S \mid ((S - x + u) \cap \cN_\ell) \in \cI_\ell\}$.\\
			Let $x_\ell \gets \arg \min_{x \in X_\ell} f(x : S)$. \label{line:tie_breaking}\\
			Add $x_\ell$ to $U$.
		}
	}
	\Return{U}.
\end{algorithm}
\DecMargin{0.5em}

\AlgSampling adds an element $u$ to the current solution if two conditions are satisfied:
first, the element is randomly sampled from the stream in Line~\ref{line:rand_consider} 
and second, the element has sufficient marginal contribution.
These two conditions are checked in this order because it is more computationally efficient as it avoids 
unnecessary oracle calls.
However, the order that these conditions are checked may be swapped without affecting the 
distribution of outcomes of the algorithm.
In fact, it is easier to analyze the algorithm when these conditions are reversed.
It is also convenient to assume that elements which have sufficiently large marginal contributions but 
are not subsampled from the stream are put into a set $R$.
We present Algorithm~\ref{alg:analysis} with these changes for the purpose of analysis. 

\IncMargin{0.5em}
\begin{algorithm}
	\caption{Equivalent-\AlgSampling$(q, f, \cI_1, \dotsc, \cI_m)$} \label{alg:analysis}
		\SetAlgoLined
	\SetInd{0.5em}{0.1em}
	\Indp
	\DontPrintSemicolon
	Let $S_0 \gets \varnothing$ and $R \gets \varnothing$.\\
	\For{every arriving element $u_i$}
	{
		Let $S_i \gets S_{i-1}$.\\
		Let $U_i \gets {\ExchangeAlg}(S_{i-1}, u_i)$.\\
		\If{$f(u_i \mid S_{i-1}) \geq (1 + c) \cdot f(U_i : S_{i-1})$}
		{
			\lWith{probability $q$}
			{
			Let $S_i \gets (S_{i-1} \setminus U_i) \cup \{ u_i \}$.
			}
			\lOtherwise
			{
			Add $u_i$ to $R$.
			}
		}
	}
	\Return{$S_n$}.
\end{algorithm} 
\DecMargin{0.5em}

We now formally show that \AlgSampling and Algorithm~\ref{alg:analysis} have the same distribution of 
returned sets.
The proof of equivalence in the streaming setting is simpler than in the offline setting.
This is due to the fact that an ordering of the ground set does not need to be chosen by the 
algorithm{\textemdash}it is already determined by the order of the stream.
We assume that the procedure \ExchangeAlg uses a tie-breaking rule at 
Line~\ref{line:tie_breaking} of Algorithm~\ref{alg:exchange_alg} which depends only on the set 
$S$ and the element $u$. For simplicity, we also assume that this rule is deterministic, so that the set 
$U_i$ is deterministic conditioned on the current solution 
$S_{i-1}$ and the new element $u_i$ which are given as input to \ExchangeAlg. This is useful 
because it implies that the event $f(u_i \mid S_{i-1}) \geq (1 + c) \cdot f(U_i : S_{i-1})$ is also 
deterministic conditioned on  $S_{i-1}$ and $u_i$. We note, however, that the proof can be easily 
made to work also with a randomized tie-breaking rule in the procedure \ExchangeAlg, as long as 
this rule only depends on the set $S$ and the element $u$.

\begin{lemma} \label{lemma:streaming_equiv}
	Let $S_i^1$ and $S_i^2$ be the random solution sets maintained by \AlgSampling and 
	Equivalent-\AlgSampling at iterations $i=0, 1, \dots ,n$, respectively.
	At each iteration $i=0, 1,  \dots, n$, the random sets $S_i^1$ and $S_i^2$ have the same 
	distribution.
\end{lemma}
\begin{proof}{Proof} 
	We prove the lemma by induction on the iteration $i$. For $i = 0$ the lemma is trivial since both 
	$S_0^1$ and $S_0^2$ are initialized to be empty.
	Suppose now that $S_k^1$ and $S_k^2$ have the same distributions for all iterations $k = 0, 1, 
	\dots, 
	i-1$ and let us prove that $S_i^1$ and $S_i^2$ also share the same distribution. In fact, we show 
	the even stronger property that for every set $A \subseteq \cN$ such that $\Pr[S_{i-1}^1 = A] = 
	\Pr[S_{i-1}^2 = A] > 0$, the sets $S_i^1$ and $S_i^2$ have the same distributions conditioned on 
	the events $S_{i-1}^1 = A$ and $S_{i-1}^2 = A$, respectively.
	
	Let $u_i$ be the $i$th element in the stream, encountered by both algorithms.
	Recall that since we condition on $S_{i - 1}^1 = A$ or $S_{i - 1}^2 = A$, the same set $U_i$ is 
	chosen by both algorithms (in the case of \AlgSampling, we mean here the set that is chosen if 
	the algorithm decides to pick an element in this iteration).
	Suppose now that $A$ is such that $f(u_i \mid S_{i-1}) \geq (1 + c) \cdot f(U_i : S_{i-1})$.
	Then in \AlgSampling, the probability of updating $S_{i}^1 \gets (A \setminus U_i) \cup \{u_i\}$ is 
	$q$ and the probability of keeping $S_{i}^1 \gets S_{i-1}^1$ is $1 - q$.
	This is also true for Algorithm~\ref{alg:analysis}; that is, 
	the probability of updating $S_{i}^2 \gets (A \setminus U_i) \cup \{u_i\}$ is $q$ and the probability 
	of keeping $S_{i}^2 \gets S_{i-1}^2$ is $1 - q$.
	This is due to the fact that sampling and exchange procedures are independent in both 
	algorithms.
	If $A$ is such that $f(u_i \mid S_{i-1}) < (1 + c) \cdot f(U_i : S_{i-1})$, then both algorithms keep 
	the current solution (that is, $S_{i} \gets S_{i-1}$) with probability 1.
	Thus, we have shown that for every $A , B \subseteq \ground$, if $\Pr[S_{i - 1}^1 = A] > 0$, then
	\[
	\CondProb{S_i^1 = B}{S_{i-1}^1 = A} 
	=
	\CondProb{S_i^2 = B}{S_{i-1}^2 = A} 
	\enspace.
	\]
	The lemma now follows by the law of total probability and the inductive hypothesis.
\end{proof}

Now that the equivalence of \AlgSampling and Algorithm~\ref{alg:analysis} has been established, we are guaranteed that any 
approximation guarantee for Algorithm~\ref{alg:analysis} also holds for \AlgSampling.
Accordingly, in the remainder of the section, we analyze Algorithm~\ref{alg:analysis}.
The following technical lemma shows that, for every two sets $A$ and $B$, the sum of the marginal contributions of the
elements of $B$ (as they arrive) to the already arrived elements of $A$ is larger than the total marginal 
contribution of $B$ to $A$.
\newcommand{\obsTechnical}{For every two sets $A,B \subseteq \cN$, $f(B \mid A \setminus B) \leq 
f(B : A)$.}
\begin{observation} \label{obs:technical}
	\obsTechnical
\end{observation}
\begin{proof}{Proof}
	Let us denote the elements of $B$ by $u_{i_1}, u_{i_2}, \dotsc, u_{i_{|B|}}$, where $i_1 < i_2 < \dotsb < 
	i_{|B|}$. Then,
	\begin{align*}
	f(B \mid A \setminus B)
	&= \sum_{j = 1}^{|B|} f(u_{i_j} \mid (A \cup B) \setminus \{u_{i_j}, u_{i_{j + 1}} \dotsc, u_{i_{|B|}}\}) \\
	&\leq \sum_{j = 1}^{|B|} f(u_{i_j} \mid A \setminus \{u_{i_j}, u_{i_j + 1} \dotsc, u_n\})\\
	&= \sum_{j = 1}^{|B} f(u_{i_j} \mid A \cap \{u_1, u_2, \dotsc, u_{i_j - 1}\}) \\
	&= \sum_{j = 1}^{|B|} f(u_{i_j} : A) \\
	&= f(B : A) \enspace,
	\end{align*}
	where the inequality follows from the submodularity of $f$. 
\end{proof}

Let us denote from this point on by $A$ the set of elements that ever appeared in the solution maintained by 
Algorithm~\ref{alg:analysis}---formally, $A = \bigcup_{i=1}^n S_i$. The following lemma and corollary 
show that the elements of $A \setminus S_n$ cannot contribute much to the output solution $S_n$ of 
Algorithm~\ref{alg:analysis}, and thus, their absence from $S_n$ does not make $S_n$ much less 
valuable than $A$.
\begin{lemma} \label{lem:marginals_sum}
	$f(A \setminus S_n : S_n) \leq \frac{f(S_n)}{c}$.
\end{lemma}
\begin{proof}{Proof}
Fix an element $u_i \in A$, then
\begin{align*}
f(S_i) - f(S_{i-1})
&= f(S_{i-1} \setminus U_i + u_i) - f(S_{i-1}) \\
&= f(u_i \mid S_{i-1} \setminus U_i) - f(U_i \mid S_{i-1} \setminus U_i) \\
&\geq f(u_i \mid S_{i-1}) - f(U_i : S_{i-1}) \\
&\geq c \cdot f(U_i : S_{i-1}) 
\enspace,
\end{align*}
where the first inequality follows from the submodularity of $f$ and Observation~\ref{obs:technical}, 
and the second inequality holds since the fact that Algorithm~\ref{alg:analysis} accepted $u_i$ into its 
solution implies $f(u_i \mid S_{i-1}) \geq (1 + c) \cdot f(U_i : S_{i-1})$.

Because every element of $A \setminus S_n$ has been removed exactly once from 
the solution of Algorithm~\ref{alg:analysis}, the sets $U_i$ such that $u_i \in A$ form a disjoint 
partition of $A \setminus S_n$. Thus,
\[
f(A \setminus S_n : S_n)
=
\sum_{u_i \in A} f(U_i : S_n)
\leq
\sum_{u_i \in A} \frac{f(S_i) - f(S_{i-1})}{c}
=
\frac{f(S_n) - f(\varnothing)}{c}
\leq
\frac{f(S_n)}{c}
\enspace,
\]
where the first inequality follows from the inequalities above, the second equality holds 
since $S_i = S_{i-1}$ whenever $u_i \not \in A$ and the second inequality follows from the 
non-negativity of $f$. 
\end{proof}

\begin{corollary} \label{cor:sets_ratio}
	$f(A) \leq \frac{c + 1}{c} \cdot f(S_n)$.
\end{corollary}
\begin{proof}{Proof} Observe that
\[ f(A) 
=  f(A \setminus S_n \mid S_n) + f(S_n) 
\leq f(A \setminus S_n : S_n) + f(S_n) 
\leq \frac{f(S_n)}{c} + f(S_n)
= \frac{c + 1}{c} \cdot f(S_n) \enspace,\]
where the first equality follows from $S_n \subseteq A$, 
the first inequality follows from Observation~\ref{obs:technical},
and the second inequality follows from Lemma~\ref{lem:marginals_sum}.
\end{proof}

Our next objective is to show that the value of the elements of the optimal solution that do not belong 
to $A$ is not too large compared to the value of $A$ itself. To this end, we need a mapping from the 
elements of the optimal solution to elements of $A$. Such a mapping is given by 
Proposition~\ref{prop:mapping}. However, before stating Proposition~\ref{prop:mapping}, we need to present a simplification given by 
Reduction~\ref{red:exact}.
\begin{reduction} \label{red:exact}
	For the sake of analyzing the approximation ratio of Algorithm~\ref{alg:analysis}, one may assume 
	that every element $u \in \cN$ belongs to \emph{exactly} $p$ out of the $m$ ground sets $\cN_1, 
	\cN_2, \dotsc, \cN_m$ of the matroids defining the $p$-matchoid $(\ground, \cI)$.
\end{reduction}
\begin{proof}{Proof}
For every element $u \in \cN$ that belongs to the ground sets of only $p' < p$ out of the $m$ matroids 
$(\cN_1, \cN_1), (\cN_2, \cN_2), \dotsc, (\cN_m, \cI_m)$, we can add $u$ to $p - p'$ additional matroids 
as a free element (\ie, an element whose addition to an independent set always keeps the set 
independent). One can observe that the addition of $u$ to these matroids does not affect the behavior 
of Algorithm~\ref{alg:analysis} at all, but makes $u$ obey the technical property of belonging to 
\emph{exactly} $p$ out of the ground sets $\cN_1, \cN_2, \dotsc, \cN_m$.
\end{proof}

From this point on we implicitly make the assumption allowed by Reduction~\ref{red:exact}. In 
particular, the proof of Proposition~\ref{prop:mapping} relies on this assumption. To state the proposition, we still need some additional notation. For every $1 \leq i \leq n$, we define
\[
d(i)
=
\begin{cases}
1 + \max \{i \leq j \leq n \mid u_i \in S_j\} & \text{if $u_i \in A$} \enspace,\\
i & \text{otherwise} \enspace.
\end{cases}
\]
In general, $d(i)$ is the index of the element whose arrival made Algorithm~\ref{alg:analysis} remove 
$u_i$ from its solution. Two exceptions to this rule are as follows. If $u_i$ was never added to the 
solution, then $d(i) = i$; and if $u_i$ was never removed from the solution, then $d(i) = n + 1$. 
\newcommand{\propMapping}{%
	For every set $T \in \cI$ which does not include elements of $R$, there exists a mapping $\phi_T$ 
	from elements of $T$ to multi-subsets of $A$ such that
	\begin{itemize}
		\item every element $u \in S_n$ appears at most $p$ times in the multi-sets of $\{\phi_T(u) \mid u 
		\in T\}$.
		\item every element $u \in A \setminus S_n$ appears at most $p - 1$ times in the multi-sets of 
		$\{\phi_T(u) \mid u \in T\}$.
		\item every element $u_i \in T \setminus A$ obeys
		$
		f(u_i \mid S_{i-1})
		\leq
		(1 + c) \cdot \sum_{u_j \in \phi_T(u_i)} f(u_j : S_{d(j) - 1})
		$.
		\item every element $u_i \in T \cap A$ obeys
		$
		f(u_i \mid S_{i-1})
		\leq
		f(u_j : S_{d(j)-1})
		$
		for every
		$
		u_j \in \phi_T(u_i)
		$,
		and the multi-set $\phi_T(u_i)$ contains exactly $p$ elements (including repetitions).
\end{itemize}}
\begin{proposition} \label{prop:mapping}
	\propMapping
\end{proposition}

Because the proof of Proposition~\ref{prop:mapping} is lengthy and detailed, we 
defer it to Section~\ref{ssc:mapping}. 
Instead, we prove now a useful technical observation.
Let $Z = \{u_i \in \cN \mid f(u_i \mid 
S_{i-1}) < 0\}$. 
\newcommand{\obsSetProperties}{Consider an arbitrary element $u_i \in \cN$.
	\begin{itemize}
		\item If $u_i \not \in Z$, then $f(u_i : S_{i'}) \geq 0$ for every $i' \geq i - 1$. In particular, since 
		$d(i) \geq i$, $f(u_i : S_{d(i) - 1}) \geq 0$.
		\item $A \cap (R \cup Z) = \varnothing$.
\end{itemize}}
\begin{observation} \label{obs:set_properties}
	\obsSetProperties
\end{observation}
\begin{proof}{Proof}
	To see why the first part of the observation is true, consider an arbitrary element $u_i \not \in Z$. 
	Then,
	\[
	0
	\leq
	f(u_i \mid S_{i - 1})
	\leq
	f(u \mid S_{i'} \cap \{u_1, u_2, \dotsc, u_{i - 1}\})
	=
	f(u : S_{i'})
	\enspace,
	\]
	where the second inequality follows from the submodularity of $f$ and the inclusion $S_{i'} \cap 
	\{u_1, u_2, \dotsc,\allowbreak u_{i - 1}\} \subseteq S_{i-1}$ (which holds because elements are only 
	added by Algorithm~\ref{alg:analysis} to its solution at the time of their arrival).
	
	It remains to prove the second part of the observation. Note that Algorithm~\ref{alg:analysis} adds 
	every arriving element to at most one of the sets $A$ and $R$, and thus, these sets are disjoint; 
	hence, to prove the observation it is enough to show that $A$ and $Z$ are also disjoint. Assume 
	towards a contradiction that this is not the case, and let $u_i$ be the first element to arrive which 
	belongs to both $A$ and $Z$. Then,
	\[
	f(u_i \mid S_{i - 1})
	\geq
	(1 + c) \cdot f(U_i : S_{i - 1})
	=
	(1 + c) \cdot \sum_{u_j \in U_i} f(u_j : S_{d(j) - 1})
	\enspace.
	\]
	To see why that inequality leads to a contradiction, notice that its leftmost hand side is negative by our 
	assumption that $u_i \in Z$, while its rightmost hand side is non-negative by the first part of this 
	observation since the choice of $u_i$ implies that no element of $U_i \subseteq S_{i - 1} \subseteq A 
	\cap \{u_1, u_2, \dotsc, u_{i-1}\}$ can belong to $Z$. 
\end{proof}

%
%

We are now ready to show that the value of the elements of the optimal solution that do not belong 
to $A$ is not too large compared to the value of $A$ itself when the sampling parameter $q$ is chosen 
appropriately.
\begin{lemma} \label{lem:raw}
	If $q = \left( (1 + c)p + 1 \right)^{-1}$,  then $\bE[f(S_n)] \geq \frac{c}{(1+c)^2p} \cdot \bE[f(A \cup 
	\OPT)]$.
\end{lemma}
\begin{proof}{Proof}
	Since $S_i \subseteq A$ for every $0 \leq i \leq n$, the submodularity of $f$ guarantees that
	\begin{align*}
	f(A \cup \OPT)
	&\leq
	f(A) + \sum_{u_i \in \OPT \setminus (R \cup A)}  \mspace{-36mu} f(u_i \mid A) + \sum_{u_i \in (\OPT 
	\setminus A) \cap R} \mspace{-18mu} f(u_i \mid A) \\
	&\leq
	f(A) + \sum_{u_i \in \OPT \setminus (R \cup A)}  \mspace{-36mu} f(u_i \mid S_{i-1}) + \sum_{u_i \in 
	(\OPT \setminus A) \cap R} \mspace{-18mu} f(u_i \mid S_{i - 1}) \\
	&\leq
	\frac{1+c}{c} \cdot f(S_n) + \sum_{u_i \in \OPT \setminus (R \cup A)} \mspace{-36mu} f(u_i \mid 
	S_{i-1})  + \sum_{u_i \in \OPT \cap R} \mspace{-18mu} f(u_i \mid S_{i-1})
	\enspace,
	\end{align*}
	where the third inequality follows from Corollary~\ref{cor:sets_ratio} and the fact that $A \cap R = 
	\varnothing$ by Observation~\ref{obs:set_properties}. Let us now consider the function $\phi_{\OPT 
	\setminus R}$ whose existence is guaranteed by Proposition~\ref{prop:mapping} when we choose 
	$T = \OPT \setminus R$. The property guaranteed by Proposition~\ref{prop:mapping} for 
	elements of $T \setminus A$ implies
	\[
	\sum_{u_i \in \OPT \setminus (R \cup A)} \mspace{-27mu} f(u_i \mid S_{i-1})
	\leq
	(1 + c ) \cdot \mspace{-18mu} 
	\sum_{\substack{u_i \in \OPT \setminus (R \cup A)\\u_j \in \phi_{\OPT 
	\setminus R}(u_i)}} 
	\mspace{-27mu} 
	f(u_j:S_{d(j) - 1} )
	\enspace.
	\]
	Additionally,
	\begin{align*}
	\sum_{\substack{u_i \in \OPT \setminus (R \cup A)\\u_j \in \phi_{\OPT \setminus R}(u_i)}} 
	\mspace{-36mu} f(u_j:S_{d(j) - 1} ) &+ p \cdot \mspace{-18mu} \sum_{u_i \in \OPT \cap A} 
	\mspace{-18mu} f(u_i \mid S_{i-1}) \\
	&\leq
	\mspace{-18mu} \sum_{\substack{u_i \in \OPT \setminus R\\u_j \in \phi_{\OPT \setminus R}(u_i)}} 
	\mspace{-36mu} f(u_j: S_{d(j)- 1} )\\
	&\leq
	p \cdot \sum_{u_j \in S_n} f(u_j:S_n) + (p-1) \cdot \mspace{-9mu} \sum_{u_j \in A \setminus S_n} 
	\mspace{-9mu} f(u_j: S_{d(j) - 1})\\
	&\leq
	p \cdot f(S_n) + \frac{p-1}{c} \cdot f(S_n) \\
	&=
	\frac{(1 + c) \cdot p - 1}{c}  \cdot f(S_n)
	\enspace,
	\end{align*}
	where the first inequality follows from the properties guaranteed by Proposition~\ref{prop:mapping} 
	for elements of $T \cap A$ (note that the sets $\OPT \setminus (R \cup A)$ and $OPT \cap A$ are a 
	disjoint partition of $\OPT \setminus R$ by Observation~\ref{obs:set_properties}), and the second 
	inequality follows from the properties guaranteed by Proposition~\ref{prop:mapping} for elements of 
	$A \setminus S_n$ and $S_n$ because every element $u_i$ in the multisets produced by $\phi_{\OPT 
	\setminus R}$ belongs to $A$, and thus, obeys $f(u_i : S_{d(i) - 1}) \geq 0$ by 
	Observation~\ref{obs:set_properties}. Finally, the last inequality follows from 
	Lemma~\ref{lem:marginals_sum} and the fact that $f(u_j : S_{d(j) - 1}) \leq f(u_j : S_n)$ for every $1 
	\leq j \leq n$. Combining all the above inequalities, we get
	\begin{align*} \label{eq:cup_bound}
	f(A&{} \cup \OPT) \\
	&\leq \frac{1 + c}{c} \cdot f(S_n) + (1 + c) \cdot \left[  \frac{(1 + c) \cdot p - 1}{c}  \cdot f(S_n) - p \cdot 
	\mspace{-9mu} \sum_{u_i \in 
	\OPT \cap A} \mspace{-18mu} f(u_i \mid S_{i-1}) \right] + \sum_{u_i \in \OPT \cap R} \mspace{-18mu} 
	f(u_i \mid S_{i-1}) \nonumber \\
	&=
	\frac{(1 + c)^2\cdot p}{c} \cdot f(S_n) - (1 + c)p \cdot \mspace{-9mu} \sum_{u_i \in \OPT \cap A} 
	\mspace{-18mu} f(u_i \mid S_{i-1}) + \sum_{u_i \in \OPT \cap R} \mspace{-18mu} f(u_i \mid S_{i-1})
	\enspace.
	\end{align*}
	
	By the linearity of expectation, to prove the lemma it suffices to show that the expectation of the last two terms is 
	non-positive.
	We will show the stronger statement that the expectation of the last two terms is zero.  
	To this end, consider an arbitrary element $u_i \in \OPT$. When $u_i$ 
	arrives, one of two things happens. The first option is that Algorithm~\ref{alg:analysis} discards 
	$u_i$ without adding it to either its solution or to $R$. The other option is that 
	Algorithm~\ref{alg:analysis} adds $u_i$ to its solution (and thus, to $A$) with probability $q$, and 
	to $R$ with probability $1 - q$. The crucial observation here is that at the time of $u_i$'s arrival the 
	set $S_{i - 1}$ is already determined, and thus, this set is independent of the decision of the 
	algorithm to add $u$ to $A$ or to $R$; which implies the following equality (given an event $\cE$, 
	we use here $\characteristic[\cE]$ to denote an indicator for it).
	\[
	\frac{\bE[\characteristic[u_i \in A] \cdot f(u_i \mid S_{i - 1})]}{q}
	=
	\frac{\bE[\characteristic[u_i \in R] \cdot f(u_i \mid S_{i - 1})]}{1 - q}
	\enspace.
	\]
	Rearranging the last equality, and summing it up over all elements $u_i \in \OPT$, we get
	\[
	\frac{1 - q}{q} \cdot \bE\left[\sum_{u_i \in \OPT \cap A} \mspace{-27mu} f(u_i \mid S_{i - 1})\right]
	=
	\bE\left[\sum_{u_i \in \OPT \cap R} \mspace{-18mu} f(u_i \mid S_{i - 1})\right]
	\enspace.
	\]
	By assumption, $q = \left( (1 + c)p + 1 \right)^{-1}$, which implies $(1 - q)/q = q^{-1} - 1 = (c + 1)p$. 
	Substituting this into the equality above completes the proof.
\end{proof}

Now we are ready to prove the approximation and efficiency guarantees of 
Theorem~\ref{thm:streaming_alg}.
\subsubsection{Proof of Theorem~\ref{thm:streaming_alg}.}
	We first prove that Algorithm~\ref{alg:actual} achieves the approximation ratios guaranteed by 
	Theorem~\ref{thm:streaming_alg}. 
	As discussed earlier, Algorithms~\ref{alg:actual} and~\ref{alg:analysis} have 
	identical output distributions, and so it suffices to show that Algorithm~\ref{alg:analysis} 
	achieves the desired approximation ratios. 
	Recall that $q= \left( (1+c)p + 1 \right)^{-1}$, so by Lemma~\ref{lem:raw},
	\begin{equation} \label{eq:final_inequality_c}
	\bE[f(S_n)] \geq \frac{c}{(1 + c)^2p} \cdot \bE[f(A \cup \OPT)] \enspace.
	\end{equation}
	Suppose that $f$ is monotone. Setting $c = 1$ yields $\frac{c}{(1 + c)^2p} = \nicefrac{1}{4p}$. 
	Additionally, the monotonicity of $f$ implies that $\bE[f(A \cup \OPT)] \geq f(\OPT)$.
	Substituting these two observations into \eqref{eq:final_inequality_c} yields
	\[ \bE[f(S_n)] \geq \frac{1}{4p} f(\OPT) \enspace, \]
	which establishes the approximation guarantee in the monotone case.
	Consider now the more general case in which $f$ is not necessarily monotone.
	Note that each element appears in $A$ with probability at most $q$ due to subsampling, thus, by 
	Lemma~\ref{lem:buchbinder}, we have 
	$\E{f(A \cup \OPT)} \geq \left( 1 - q \right) f( \OPT)$. Substituting this into 
	Inequality~\eqref{eq:final_inequality_c} and setting $c = \sqrt{1 + 1/p}$ yields
	\begin{align*}
	\bE[f(S_n)] &\geq (1 - q) \frac{c}{(1 + c)^2p} \cdot f(\OPT) \\
	&= \left( 1 - \frac{1}{(1+c)p + 1}\right) \frac{c}{(1 + c)^2p} \cdot f(\OPT) \\
	&= \left( 1 - \frac{1}{(1+\sqrt{1 + 1/p})p + 1} \right) \frac{\sqrt{1 + 1/p}}{(1 + \sqrt{1 + 1/p})^2p} \cdot 
	f(\OPT) \\
	&= \frac{1}{2p + 2 \sqrt{p(p+1)} + 1} \cdot f(\OPT),
	\end{align*}
	which yields the promised approximation guarantee for this case.
	
	Now we verify that \AlgSampling achieves the guaranteed memory and oracle complexities.
	Algorithm~\ref{alg:actual} has to keep the following three sets in memory: $S_i$, $U_i$ and $X_\ell$. 
	Since $U_i$ and $X_\ell$ are subsets of $S_{i-1}$, they are independent, and so is $S_i$. Hence, each one of the three sets $S_i$, $U_i$ and $X_\ell$ contains at most $k$ elements. Thus, 
	$O(k)$ memory suffices for the algorithm.
	When an arriving element is not sampled (which happens with probability $1 - q$), no queries to the 
	evaluation or independence oracles are required. 
	A sampled element requires $O(km)$ queries. Because $q = O(1/p)$, an arriving element requires $q 
	\cdot O(km) = O(km/p)$ oracle queries in expectation.  
	
\subsubsection{Proof of Proposition~\ref{prop:mapping}.} 
\label{ssc:mapping}
In this section we prove Proposition~\ref{prop:mapping}. 
Before doing so, we introduce some terminology regarding matroids which will be used in the proof. 
A \emph{circuit} is a dependent set which is minimal with respect to inclusion; that is, $C \notin 
\cI$ is a circuit if $A \notin \cI$ and $A \subseteq C$ imply $C = A$.
An element $u$ is \emph{spanned} by a set $S$ if the maximum size independent subsets of $S$ and $S+u$ are of the same size.
Note that it follows from these definitions that every element of $u$ of a circuit $C$ is spanned by $C - u$.

Let us also recall some of the sets involved in Algorithm~\ref{alg:analysis}.
The sequence $S_1, \dots S_n$ are the solutions constructed by the algorithm and $S_n$ is the final returned solution.
The set $A = \cup_{i=1}^n S_i$ is the set of all elements which were added to some solution and the set $R$ contains elements which had sufficiently large marginal gain but were rejected with probability $q$.
Lastly, $d(i)$ is the index of the element whose arrival made Algorithm~\ref{alg:analysis} remove $u_i$ from its solution.
Now, we restate the proposition itself.

\begin{manualtheorem}{\ref{prop:mapping}}
	\propMapping
\end{manualtheorem}

We begin the proof of Proposition~\ref{prop:mapping} by constructing $m$ graphs, one for each of the matroids defining $\cM$. 
For every $1 \leq \ell \leq m$, the graph $G_\ell$ contains two types of vertices: its internal vertices are the elements of $A \cap \cN_\ell$, and its external vertices are the elements of $\{u_i \in \cN_\ell \setminus (R \cup A) \mid (S_{i-1} + u_i) \cap \cN_\ell \not \in \cI_\ell\}$. 
Informally, the external elements of $G_\ell$ are the elements of $\cN_\ell$ which were rejected upon arrival by Algorithm~\ref{alg:analysis} and the matroid $\cM_\ell = (\cN_\ell, \cI_\ell)$ can be (partially) blamed for this rejection. 

The arcs of $G_\ell$ are created using the following iterative process that 
creates some arcs of $G_\ell$ in response to every arriving element. For every $1 \leq i \leq n$, consider 
the element $x_\ell$ selected by the execution of {\ExchangeAlg} on the element $u_i$ and the set 
$S_{i-1}$. From this point on we denote this element by $x_{i, \ell}$. If no $x_{i, \ell}$ element was 
selected by the above execution of {\ExchangeAlg}, or $u_i \in R$, then no $G_\ell$ arcs are created in 
response to $u_i$. Otherwise, let $C_{i, \ell}$ be the single circuit of the matroid $\cM_\ell$ in the set 
$(S_{i-1} + u_i) \cap \cN_\ell$---there is exactly one circuit of $\cM_\ell$ in this set because $S_{i-1}$ is 
independent, but $(S_{i-1} + u_i) \cap \cN_\ell$ is not independent in $\cM_\ell$. One can observe that 
$C_{i, \ell} - u_i$ is equal to the set $X_\ell$ in the above-mentioned execution of {\ExchangeAlg}, and 
thus, $x_{i, \ell} \in C_{i, \ell}$. We now denote by $u'_{i, \ell}$ the vertex out of $\{u_i, x_{i, \ell}\}$ that 
does not belong to $S_i$---notice that there is exactly one such vertex since $x_{i, \ell} \in U_i$, which 
implies that it appears in $S_i$ if $S_i = S_{i-1}$ and does not appear in $S_i$ if $S_i = S_{i-1} \setminus 
U_i + u_i$. Regardless of the node chosen as $u'_{i, \ell}$, the arcs of $G_\ell$ created in response to $u_i$ 
are all the possible arcs from $u'_{i, \ell}$ to the other vertices of $C_{i, \ell}$. Observe that these are 
valid arcs for $G_\ell$ in the sense that their endpoints (\ie, the elements of $C_{i, \ell}$) are all vertices 
of $G_\ell$ (for the elements of $C_{i, \ell} - u_i$ this is true since $C_{i, \ell} - u_i \subseteq S_{i-1} 
\cap \cN_\ell \subseteq A \cap \cN_\ell$, and for the element $u_i$ this is true since the existence of 
$x_{i, \ell}$ implies $(S_{i-1} + u_i) \cap \cN_\ell \not \in \cI_\ell$).
See Figure~\ref{fig:adding-arcs} for a sketch of how arcs are added to $G_\ell$.

\begin{figure}
	\centering
	\includegraphics[width=0.75\textwidth]{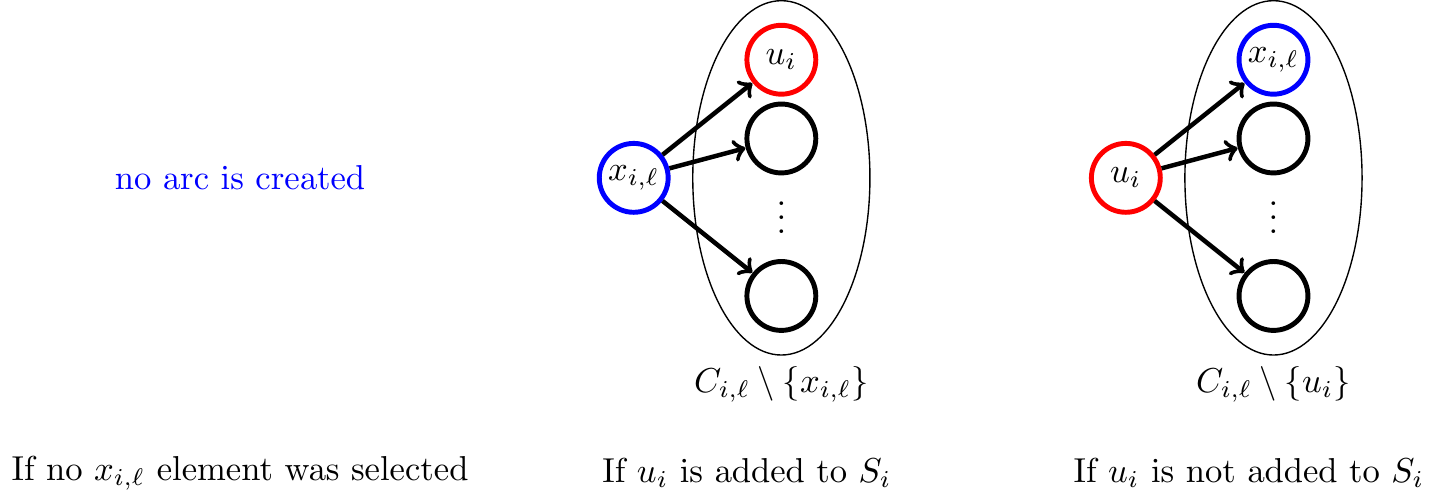}
	\caption{A sketch of how arcs are constructed in the graph $G_\ell$ at each 
	iteration.}
	\label{fig:adding-arcs}
\end{figure}

Some properties of $G_\ell$ are given by the following observation. Given a graph $G$ and a vertex 
$u$, we denote by $\delta^+_G(u)$ the set of vertices to which there is a direct arc from $u$ in $G$.
\newcommand{\obsGraphProperties}{%
	For every $1 \leq \ell \leq m$,
	\begin{itemize}
		\item every non-sink vertex $u$ of $G_\ell$ is spanned by the set $\delta^+_{G_\ell}(u)$.
		\item for every two indexes $1 \leq i, j \leq n$, if $u'_{i, \ell}$ and $u'_{j, \ell}$ both exist and $i \neq 
		j$, then $u'_{i, \ell} \neq u'_{j, \ell}$.
		\item $G_\ell$ is a directed acyclic graph.
\end{itemize}}
\begin{observation} \label{obs:graph_properties}
	\obsGraphProperties
\end{observation}
\begin{proof}{Proof}
Consider an arbitrary non-sink node $u$ of $G_\ell$. Since there are arcs leaving $u$, $u$ must be 
equal to $u'_{i, \ell}$ for some $1 \leq i \leq n$. This implies that $u$ belongs to the circuit $C_{i, \ell}$, 
and that there are arcs from $u$ to every other vertex of $C_{i, \ell}$. Thus, $u$ is spanned by the 
vertices of $\delta^+_{G_\ell}(u) \supseteq C_{i, \ell} - u$ because the fact that $C_{i, \ell}$ is a circuit 
containing $u$ implies that $C_{i, \ell} - u$ spans $u$. This completes the proof of the first part of the 
observation.

Let us prove now a useful technical claim. Consider an index $1 \leq i \leq n$ such that $u'_{i, \ell}$ 
exists, and let $j$ be an arbitrary value $i < j \leq n$. We will prove that $u'_{i, \ell}$ does not belong to $C_{j, 
\ell}$. By definition, $u'_{i, \ell}$ is either $u_i$ or the vertex $x_{i, \ell}$ that belongs to $S_{i - 1}$, and thus, 
arrived before $u_i$ and is not equal to $u_j$; hence, in both cases, we have that $u'_{i, \ell} \neq u_j$. 
Moreover, 
combining the fact that $u'_{i, \ell}$ is either $u_i$ or arrived before $u_i$ and the observation that $u'_{i, \ell}$ is 
never a part of $S_i$, we get that $u'_{i, \ell}$ cannot belong to $S_j \supseteq C_{j, \ell} - u_j$, which implies 
the claim together with the previous observation that $u'_{i, \ell} \neq u_j$.

The technical claim that we proved above implies the second part of the lemma, namely that for every 
two indexes $1 \leq i, j \leq n$, if $u'_{i, \ell}$ and $u'_{j, \ell}$ both exist and $i \neq j$, then $u'_{i, \ell} 
\neq u'_{j, \ell}$. To see why that is the case, assume without loss of generality $i < j$. Then, the above 
technical claim implies that $u'_{i, \ell} \not \in C_{j, \ell}$, which implies $u'_{i, \ell} \neq u'_{j, \ell}$ 
because $u'_{j, \ell} \in C_{j, \ell}$.

At this point, let us assume towards a contradiction that the third part of the observation is not true, 
\ie, that there exists a circuit $L$ in $G_\ell$. Since every vertex of $L$ has a non-zero out degree, every 
such vertex must be equal to $u'_{i, \ell}$ for some $1 \leq i \leq n$. Thus, there must be indexes $1 \leq 
i_1 < i_2 \leq n$ such that $L$ contains an arc from $u'_{i_2,\ell}$ to $u'_{i_1, \ell}$. Since we already 
proved that $u'_{i_2,\ell}$ cannot be equal to $u'_{j, \ell}$ for any $j \neq i_2$, the arc from 
$u'_{i_2,\ell}$ to $u'_{i_1, \ell}$ must have been created in response to $u_{i_2}$, hence, $u'_{i_1, \ell} \in 
C_{i_2, \ell}$, which contradicts the technical claim we have proved. 
\end{proof}

One consequence of the properties of $G_\ell$ proved by the last observation is given by the following 
lemma. A slightly weaker version of this lemma was proved implicitly by~\cite{V11}, and was stated as 
an explicit lemma by~\cite{CGQ15}.

\newcommand{\lemFunctionGraph}{%
	Consider an arbitrary directed acyclic graph $G = (V, E)$ whose vertices are elements of some 
	matroid $\cM'$. If every non-sink vertex $u$ of $G$ is spanned by $\delta^+_G(u)$ in $\cM'$, then 
	for every set $S$ of vertices of $G$ which is independent in $\cM'$ there must exist an injective 
	function $\psi_S: S \rightarrow V$ such that, for every vertex $u \in S$, $\psi_S(u)$ is a sink of $G$ which is 
	reachable from $u$.}
\begin{lemma} \label{lem:function_graph}
	\lemFunctionGraph
\end{lemma}
\begin{proof}{Proof}
Let us define the \emph{width} of a set $S$ of vertices of $G$ as the number of arcs that appear on 
some path starting at a vertex of $S$ (more formally, the width of $S$ is the size of the set $\{e \in E 
\mid \text{there is a path in $G$ that starts in a vertex of $S$ and includes $e$}\}$). We prove the 
lemma by induction of the width of $S$. 
We begin by considering the base case where $S$ is of width $0$. In this case, 
the vertices of $S$ cannot have any outgoing arcs because such arcs would have contributed to the 
width of $S$, and thus, they are all sinks of $G$. Thus, the lemma holds for the trivial function 
$\psi_S$ mapping every element of $S$ to itself. Assume now that the width $w$ of $S$ is larger than 
$0$, and assume that the lemma holds for every set of width smaller than $w$. Let $u$ be a non-sink 
vertex of $S$ such that there is no path in $G$ from any other vertex of $S$ to $u$. Notice that such 
a vertex must exist since $G$ is acyclic. By the assumption of the lemma, $\delta^+(u)$ spans $u$. In 
contrast, since $S$ is independent, $S - u$ does not span $u$, and thus, there must exist an element 
$v \in \delta^+(u) \setminus S$ such that the set $S' = S - u + v$ is independent.

Let us explain why the width of $S'$ must be strictly smaller than the width of $S$. First, consider an 
arbitrary arc $e$ which is on a path starting at a vertex $u' \in S'$. If $u' \in S$, then $e$ is also on a 
path starting in a vertex of $S$. On the other hand, if $u' \not \in S$, then $u'$ must be the vertex $v$. 
Thus, $e$ must be on a path $P$ starting in $v$. Adding $uv$ to the beginning of the path $P$, we 
get a path from $u$ which includes $e$. Hence, in conclusion, we have got that every arc $e$ which 
appears on a path starting in a vertex of $S'$  (and thus, contributes to the width of $S'$) also 
appears on a path starting in a vertex of $S$ (and thus, also contributes to the width of $S$); which 
implies that the width of $S'$ is not larger than the width of $S$. To see that the width of $S'$ is 
actually strictly smaller than the width of $S$, it only remains to find an arc which contributes to the 
width of $S$, but not to the width of $S'$. Towards this goal, consider the arc $uv$. Since $u$ is a 
vertex of $S$, the arc $uv$ must be on some path starting in $u$ (for example, the path including only 
this arc), and thus, contributes to the width of $S$. Assume now towards a contradiction that $uv$ 
contributes also to the width of $S'$, \ie, that there is a path $P$ starting at a vertex $w \in S'$ which 
includes $uv$. If $w = v$, then this leads to a contradiction since it implies the existence of a circuit in 
$G$. On the other hand, if $w \neq v$, then this implies a path in $G$ from a vertex $w \neq u$ of $S$ 
to $u$, which contradicts the definition of $u$. This completes the proof that the width of $S'$ is 
strictly smaller than the width of $S$.

Using the induction hypothesis, we now get that there exists an injective function $\psi_{S'}$ mapping 
every vertex of $S'$ to a sink of $G$. Using $\psi_{S'}$, we can define $\psi_S$ as follows. For every 
$w \in S$,
\[
\psi_S(w) =
\begin{cases}
\psi_{S'}(v) & \text{if $w = u$} \enspace,\\
\psi_{S'}(w) & \text{otherwise} \enspace.
\end{cases}
\]
Since $u$ appears in $S$ but not in $S'$, and $v$ appears in $S'$ but not in $S$, the injectiveness of 
$\psi_S$ follows from the injectiveness of $\psi_{S'}$. Moreover, $\psi_S$ clearly maps every vertex of 
$S$ to a sink of $G$ since $\psi_{S'}$ maps every vertex of $S'$ to such a sink. Finally, one can 
observe that $\psi_S(w)$ is reachable from $w$ for every $w \in S$ because $\psi_S(u) = \psi_{S'}(v)$ is 
reachable from $v$ by the definition of $\psi_{S'}$, and thus, also from $u$ due to the existence of the 
arc $uv$. 
\end{proof}

For every $1 \leq \ell \leq m$, let $T_\ell$ be the set of elements of $T$ that appear as vertices of 
$G_\ell$. Since $T$ is independent and $T_\ell$ contains only elements of $\cN_\ell$, 
Observation~\ref{obs:graph_properties} and Lemma~\ref{lem:function_graph} imply together the 
existence of an injective function $\psi_{T_\ell}$ mapping the elements of $T_\ell$ to sink vertices of 
$G_\ell$. We can now define the function $\phi_T$ promised by Proposition~\ref{prop:mapping}. For 
every element $u \in T$, the function $\phi_T$ maps $u$ to the multi-set $\{\psi_{T_\ell}(u) \mid 1 \leq 
\ell \leq m \text{ and } u \in T_\ell\}$, where we assume that repetitions are kept when the expression 
$\psi_{T_\ell}(u)$ evaluates to the same element for different choices of $\ell$. Let us explain why the 
elements in the multi-sets produced by $\phi_T$ are indeed all elements of $A$, as is required by the 
proposition. Consider an element $u_i \not \in A$, and let us show that it does not appear in the range 
of $\psi_{T_\ell}$ for any $1 \leq \ell \leq m$. If $u_i$ does not appear as a vertex in $G_\ell$, then this is 
obvious. Otherwise, the fact that $u_i \not \in A$ implies $u'_{i, \ell} = u_i$, and thus, the arcs of 
$G_\ell$ created in response to $u_i$ are arcs leaving $u_i$, which implies that $u_i$ is not a sink of 
$G_\ell$, and hence, does not appear in the range of $\psi_{T_\ell}$. 

Recall that every element $u \in \cN$ belongs to at most $p$ out of the ground sets $\cN_1, \cN_2, 
\dotsc,\allowbreak \cN_m$, and thus, is a vertex in at most $p$ out of the graphs $G_1, G_2, \dotsc, 
G_m$. Since $\psi_{T_\ell}$ maps every element to vertices of $G_\ell$, this implies that $u$ is in the 
range of at most $p$ out of the functions $\psi_{T_1}, \psi_{T_2}, \dotsc, \psi_{T_m}$. Moreover, since 
these functions are injective, every one of these functions that have $u$ in its range maps at most one 
element to $u$. Thus, the multi-sets produced by $\phi_T$ contain $u$ at most $p$ times. Since this 
is true for every element of $\cN$, it is true in particular for the elements of $S_n$, which is the first 
property of $\phi_T$ that we needed to prove.

Consider now an element $u \in A \setminus S_n$. Our next objective is to prove that $u$ appears at 
most $p - 1$ times in the multi-sets produced by $\phi_T$, which is the second property of $\phi_T$ 
that we need to prove. Above, we proved that $u$ appears at most $p$ times in these multi-sets by 
arguing that every such appearance must be due to a function $\psi_{T_\ell}$ that has $u$ in its range, 
and that the function $\psi_{T_\ell}$ can have this property only for the $p$ values of $\ell$ for which 
$u \in \cN_\ell$. Thus, to prove that $u$ in fact appears at most $p - 1$ times in the multi-sets produced 
by $\phi_T$, it is enough to argue that there exists a value $\ell$ such that $e \in \cN_\ell$, but 
$\psi_{T_\ell}$ does not have $u$ in its range. Let us prove that this follows from the membership of 
$u$ in $A \setminus S_n$. Since $u$ was removed from the solution of Algorithm~\ref{alg:analysis} at 
some point, there must be some index $1 \leq i \leq n$ such that both $u \in U_i$ and $u_i$ was added 
to the solution of Algorithm~\ref{alg:analysis}. Since $u \in U_i$, there must be a value $1 \leq \ell \leq 
m$ such that $u = x_{i, \ell}$, and since $u_i$ was added to the solution of 
Algorithm~\ref{alg:analysis}, $u'_{i, \ell} = x_{i, \ell}$. These equalities imply together that there are arcs 
leaving $u$ in $G_\ell$ (which were created in response to $u_i$). Thus, the function $\psi_{T_\ell}$ 
does not map any element to $u$ because $u$ is not a sink of $G_\ell$, despite the fact that $u \in 
\cN_\ell$.

To prove the other guaranteed properties of $\phi_T$, we need the following lemma.
\newcommand{\lemStepInternal}{%
	Consider two vertices $u_i$ and $u_j$ such that $u_j$ is reachable from $u_i$ in $G_\ell$. If $u_i \in 
	A$, then $f(u_i : S_{d(i) - 1}) \leq f(u_j : S_{d(j) - 1})$, otherwise, $f(x_{i, \ell} : S_{i-1}) \leq f(u_j : 
	S_{d(j)-1})$.}
\begin{lemma} \label{lem:step_internal}
	\lemStepInternal
\end{lemma}
\begin{proof}{Proof}
We begin by proving a weaker version of this lemma that makes the following two simplifying assumptions: (i) $u_i \in A$, and (ii) there is a direct arc from $u_i$ to $u_j$. 
The existence of the last arc implies that there is some value $1 \leq h \leq n$ such that $u'_{h, \ell} = u_i$ 
and $u_j \in C_{h, \ell}$. Since $u_i \in A$ is an internal vertex of $G_\ell$, it cannot be equal to $u_h$ because this would have 
implied that $u_h$ was rejected immediately by Algorithm~\ref{alg:analysis}, and is thus, not internal. 
Thus, $u_i = x_{h, \ell}$. Recall now that $C_{h, \ell} - u_h$ is equal to the set $X_\ell$ chosen by 
${\ExchangeAlg}$ when it is executed with the element $u_h$ and the set $S_{h-1}$;
and let us consider two cases.

\paragraph{Case 1} 
In the case of $u_j \neq u_h$, the fact that $u_i = x_{h, \ell}$ and the way $x_{h, \ell}$ is chosen out of $X_\ell$ imply
\[
f(u_i : S_{d(i) - 1})
=
f(u_i : S_{h-1})
\leq
f(u_j : S_{h-1})
\leq
f(u_j : S_{d(j)-1})
\enspace,
\]
where the equality holds since $u'_{h, \ell} = u_i$ implies $d(i) = h$ and the last inequality holds since 
$f(u_j : S_{r - 1})$ is a non-decreasing function of $r$ when $r \geq j$ and the membership of $u_j$ in 
$C_{h, \ell}$ implies $j \leq h \leq d(j)$.

\paragraph{Case 2}
It remains to consider the case $u_j = u_h$. 
In this case, the fact that $u_j = u_h$ is accepted into the 
solution of Algorithm~\ref{alg:analysis} implies
\begin{align*}
f(u_j : S_{d(j)-1})
\geq{} &
f(u_j : S_{j-1})
=
f(u_j \mid S_{j-1} \cap \{u_1, u_2,\dotsc, u_{j-1}\})
=
f(u_j \mid S_{j-1})\\
={} &
f(u_h \mid S_{h-1})
\geq
(1 + c) \cdot f(U_h : S_{h - 1})
\geq
f(U_h : S_{h - 1})\\
\geq{} &
f(x_{h, \ell} : S_{h-1})
=
f(u_i : S_{h-1})
=
f(u_i : S_{d(i) - 1})
\enspace,
\end{align*}
where the first inequality holds since $d(j) \geq j$ by definition, the last equality holds since $u'_{h, \ell} 
= u_i$ implies $d(i) = h$ and the last two inequalities follow from the fact that the elements of $U_h 
\subseteq A$ do not belong to $Z$ by Observation~\ref{obs:set_properties}, which implies (again, by 
Observation~\ref{obs:set_properties}) that $f(u : S_{h-1}) \geq 0$ for every $u \in U_h$. This completes 
the proof of the weaker version of the lemma.

We begin the proof of the full version of the lemma by showing that no arc of $G_\ell$ goes from an internal vertex to an external one. Assume this is 
not the case, and that there exists an arc $uv$ of $G_\ell$ from an internal vertex $u$ to an external 
vertex $v$. By definition, there must be a value $1 \leq h \leq n$ such that $v$ belongs to the circuit 
$C_{h, \ell}$ and $u'_{h, \ell} = u$. The fact that $u$ is an internal vertex implies that $u_h$ must have 
been accepted by Algorithm~\ref{alg:analysis} upon arrival because otherwise we would have gotten $u 
= u'_{h, \ell} = u_{h}$, which implies that $u$ is external, and thus, leads to a contradiction. 
Consequently, we get $C_{h, \ell} \subseteq A$ because every element of $C_{h, \ell}$ must either be 
$u_{h}$ or belong to $S_{h - 1}$. In particular, $v \in A$, which contradicts our assumption that $v$ is 
an external vertex.

We are now ready to prove the lemma for $u_i \in A$.
Consider some path $P$ from $u_i$ to $u_j$, and let us denote the 
vertices of this path by $u_{r_0}, u_{r_1}, \dotsc, u_{r_{|P|}}$. Since $u_i$ is an internal vertex of 
$G_\ell$ and we already proved that no arc of $G_\ell$ goes from an internal vertex to an external one, 
all the vertices of $P$ must be internal. Thus, by applying the weaker version of the lemma that we have 
already proved to every pair of adjacent vertices along the path $P$, we get that the expression 
$f(u_{r_k} : S_{d(r_k) - 1})$ is a non-decreasing function of $k$, and in particular,
\[
f(u_i : S_{d(i) - 1})
=
f(u_{r_0} : S_{d(r_0) - 1})
\leq
f(u_{r_k} : S_{d(r_k) - 1})
=
f(u_j : S_{d(j) - 1})
\enspace.
\]

It remains to prove the lemma for $u_i \not \in A$. Let $u_h$ denote the vertex immediately after $u_i$ on some 
path from $u_i$ to $u_j$ in $G_\ell$. Since $u_i \not \in A$, we get that $u'_{i, \ell} = u_i$, which implies 
that the arcs of $G_\ell$ that were created in response to $u_i$ go from $u_i$ to the vertices of $C_{i, 
\ell} - u_i$. Since Observation~\ref{obs:graph_properties} guarantees that $u_i = u'_{i, \ell} \neq u'_{j, 
\ell}$ for every value $1 \leq j \leq n$ which is different from $i$, there cannot be any other arcs in 
$G_\ell$ leaving $u_i$, and thus, the existence of an arc from $u_i$ to $u_h$ implies $u_h \in C_{i, \ell} - 
u_i$. Recall now that $C_{i, \ell} - u_i$ is equal to the set $X_\ell$ in the execution of {\ExchangeAlg} 
corresponding to the element $u_i$ and the set $S_{i-1}$, and thus, by the definition of $x_{i, \ell}$, 
$f(x_{i, \ell} : S_{i-1}) \leq f(u_h : S_{i-1})$. Additionally, as an element of $C_{i, \ell} - u_i$, $u_h$ must 
be a member of $S_{i-1} \subseteq A$, and thus, by the part of the lemma we have already proved, we 
get $f(u_h : S_{d(h)-1}) \leq f(u_j : S_{d(j) - 1})$ because $u_j$ is reachable from $u_h$. Combining the 
two inequalities we have proved, we get
\[
f(x_{i, \ell} : S_{i-1})
\leq
f(u_h : S_{i-1})
\leq
f(u_h : S_{d(h)-1})
\leq
f(u_j : S_{d(j) - 1})
\enspace,
\]
where the second inequality holds since the fact that $u_h \in C_{i, \ell} - u_i \subseteq S_{i-1}$ implies 
$d(h) \geq i$. 
\end{proof}

Consider now an arbitrary element $u_i \in T \setminus A$. Let us denote by $u_{r_\ell}$ the element 
$u_{r_\ell} = \psi_{T_\ell}(u_i)$ if it exists, and recall that this element is reachable from $u_i$ in $G_\ell$. 
Thus, the fact that $u_i$ is not in $A$ implies
\begin{align*}
f(u_i \mid S_{i-1})
\leq{} &
(1 + c) \cdot \sum_{u \in U_i} f(u : S_{i-1})
=
(1 + c) \cdot \mspace{-18mu} \sum_{\substack{1 \leq \ell \leq m\\(S_{i-1} + u_i) \cap \cN_\ell \not \in 
\cI_\ell}} \mspace{-36mu} f(x_{i, \ell} : S_{i-1})\\
\leq{} &
(1 + c) \cdot \mspace{-18mu} \sum_{\substack{1 \leq \ell \leq m\\(S_{i-1} + u_i) \cap \cN_\ell \not \in 
\cI_\ell}} \mspace{-36mu} f(u_{r_\ell} : S_{d(r_\ell)})
=
(1 + c) \cdot \mspace{-18mu} \sum_{u_j \in \phi_T(u_i)} f(u_j : S_{d(j)})
\enspace,
\end{align*}
where the second inequality follows from Lemma~\ref{lem:step_internal} and the last equality holds since the 
values of $\ell$ for which $(S_{i-1} + u_i) \cap \cN_\ell \not \in \cI_\ell$ are exactly the values for which 
$u_i \in T_\ell$, and thus, they are all also exactly the values for which the multi-set $\phi_T(u_i)$ 
includes the value of $\psi_{T_\ell}(u_i)$. This completes the proof of the third property of $\phi_T$ 
that we need to prove.

Finally, consider an arbitrary element $u_i \in A \cap T$. Every element $u_j \in \phi_T(u_i)$ can be 
reached from $u_i$ in some graph $G_\ell$, and thus, by Lemma~\ref{lem:step_internal},
\begin{align*}
f(u_i \mid S_{i - 1})
={} &
f(u_i \mid S_{i - 1} \cap \{u_1, u_2, \dotsc, u_{i-1}\})
=
f(u_i : S_{i-1})\\
\leq{} &
f(u_i : S_{d(i)-1})
\leq
f(u_j : S_{d(j)-1})
\enspace,
\end{align*}
where the first inequality holds since $d(i) \geq i$ by definition and $f(u_i : S_{r-1})$ is a 
non-decreasing function of $r$ for $r \geq i$. Additionally, we observe that $u_i$, as an element of $T 
\cap A$, belongs to $T_\ell$ for every value $1 \leq \ell \leq m$ for which $u_i \in \cN_\ell$, and thus, the 
size of the multi-set $\phi_T(u_i)$ is equal to the number of ground sets out of $\cN_1, \cN_2, \dotsc, 
\cN_m$ that include $u_i$. Since we assume by Reduction~\ref{red:exact} that every element belongs 
to exactly $p$ out of these ground sets, we get that the multi-set $\phi_T(u_i)$ contains exactly $p$ 
elements (including repetitions), which completes the proof of Proposition~\ref{prop:mapping}.

\section{Experimental Results}\label{sec:experimental_results}
In this section, we compare the performance of our proposed offline and streaming algorithms with the 
state-of-the-art methods for maximizing non-monotone submodular functions on real-world datasets .
We empirically demonstrate that by using the subsampling technique, our algorithms return solutions 
with similar utility values as other methods, but at a fraction of computational cost.
We first compare the performance of \algrd with other competitive offline algorithms. 
Then, we evaluate the effectiveness of \AlgSampling in terms of utility and scalability.

\subsection{Offline Algorithms} \label{sec:experiments-greedy}
In this section, we compare the performance of \algrd with two other offline 
algorithms:  \algfantom \citep{MBK16} and \algdt  \citep{FHK17}, both of which are repeated greedy 
techniques requiring $O(nkp)$ and $O(nk\sqrt{p})$ oracle queries, respectively. 
In order to improve the performance of \algrd in terms of utility, we also consider a boosted version of 
\algrd by taking the best of four runs, denoted \msg.
We test these algorithms on a personalized movie recommendation system and find that while  \algrd 
and its boosted variant return comparable solutions to \algdt and \algfantom, they run orders of 
magnitude faster. 

In the movie recommendation system application, we observe movie ratings from users, and our objective is to recommend movies to users based on their reported favorite genres. 
In particular, given a user-specified input of favorite genres, we would like to recommend a short list of movies that are diverse, and yet representative, of those genres. 
The similarity score between movies that we use is derived from user ratings similar to the methods used in \citep{Lindgren2015}.
Next, we describe the experiment setting in more detail.
 
Let $\ground$ be a set of movies, and $\mathcal{G} = \{G_1, \dots G_p \}$ be the set of all movie 
genres, where each genre is a subset $G_i \subseteq \ground$.
Note that each movie may be identified with multiple genres.
Let $s_{i,j}$ be a non-negative similarity score between movies $i,j \in \cN$, and suppose a user $u$ 
seeks a representative set of movies from genres $\mathcal{G}_u \subseteq \mathcal{G}$. 
Note that the set of movies from these genres is $\ground_u = \cup_{G_i \in \mathcal{G}_u} G_i$. A 
reasonable utility function for choosing a diverse yet representative set of movies $S$ for $u$ is 
\begin{equation} \label{eq:cut_fun}
f_u(S) = \sum \limits_{i \in S} \sum \limits_{j \in \cN_u} s_{i,j} - \lambda \sum 
\limits_{i \in S} \sum \limits_{j \in S} s_{i,j},
\end{equation} 
for some parameter $0 \leq \lambda \leq 1$. Observe that the first term is a sum-coverage function that 
captures the representativeness of $S$, and the second term is a dispersion function penalizing 
similarity within $S$. When $\lambda > 0$, this function is non-monotone and for $\lambda = 1$, this 
utility is the usual cut function.

The user may specify an upper limit $k$ on the number of movies in his recommended set $S$, as well 
as an upper limit $k_i$ on the number of movies from each genre $G_i$ appearing in $S$ 
(we call the parameters $k_1, \dots k_p$ \emph{genre limits}). 
One can show that these constraints on the recommended movie set $S$ form a $p$-extendible 
system. 

For our experiments, we use the MovieLens 20M dataset, which features 20 million ratings of \num[group-separator={,}]{27000} movies by \num[group-separator={,}]{138000} users. 
To obtain a similarity score between movies, we take an approach developed in \citep{Lindgren2015}. 
First, we fill missing entries of an incomplete movie-user matrix $M \in \mathbb{R}^{n \times m}$ via low-rank matrix completion \citep{Candes2008, Hastie2015}, where $m$ is the total number of users.
Then we randomly sample to obtain a matrix $\tilde{M} \in \mathbb{R}^{n \times \ell}$ where $\ell \ll m$ and the inner products between rows is preserved. 
The similarity score between movies $i$ and $j$ is then defined as the inner product of their corresponding rows in $\tilde{M}$. 
In our experiment, we set the total recommended movies limit to $k=10$. The genre limits $k_g$ are set 
to be equal across genres, and we vary them from 1 to 6. 
For all experiments, we use $\lambda = 0.9$.
Finally, we set our favorite genres $G_u$ to be \textit{Adventure}, \textit{Animation} and \textit{Fantasy}; thus, the corresponding constraint set is 3-extendible.
For each algorithm and test instance, we record the function value of the returned solution $S$ and 
the number of evaluations of the objective $f$, which is a machine-independent measure of run-time.

\begin{figure}[htb!] 
\begin{center}
	\subfloat[Solution Quality]{\includegraphics[width=80mm]{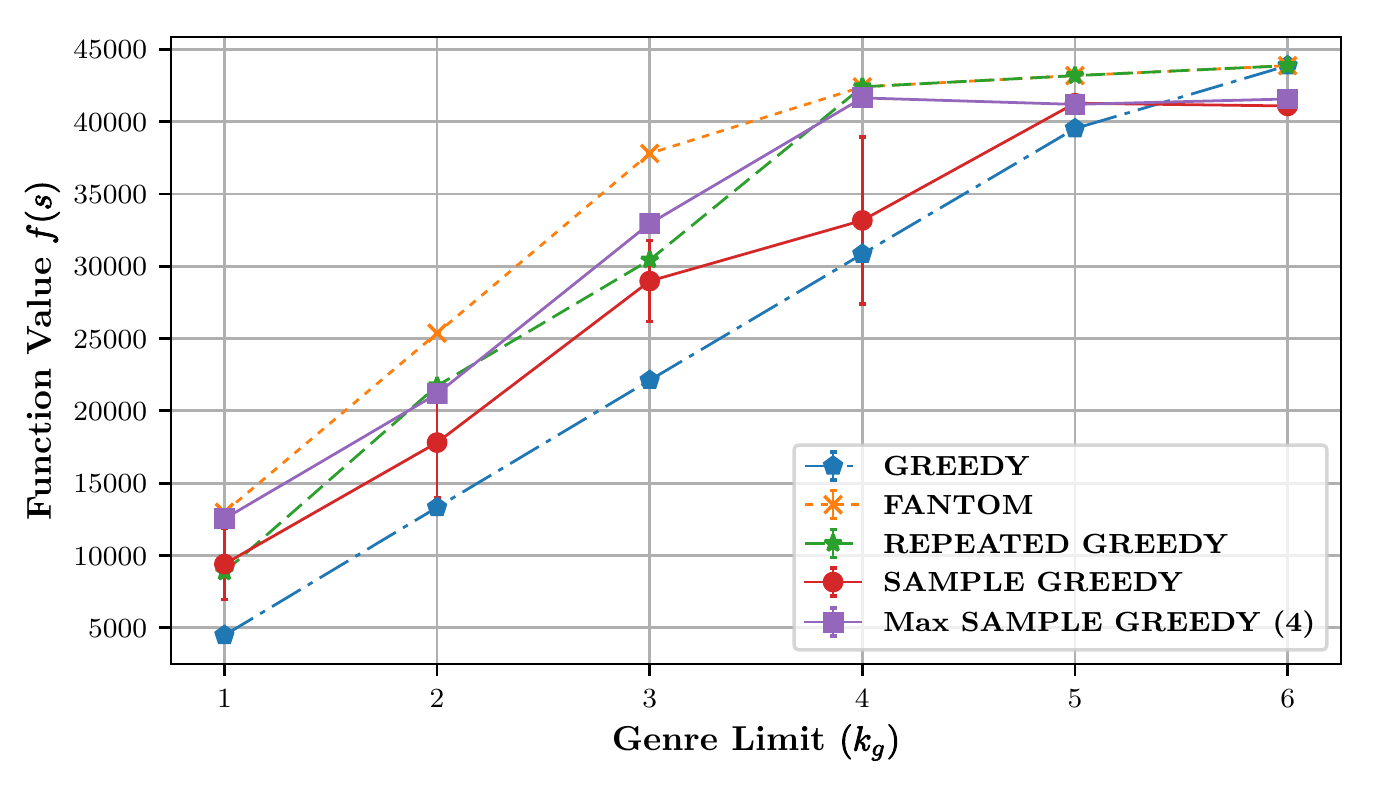}\label{fig:fun_val}}
	\subfloat[Run Time]{\includegraphics[width=80mm]{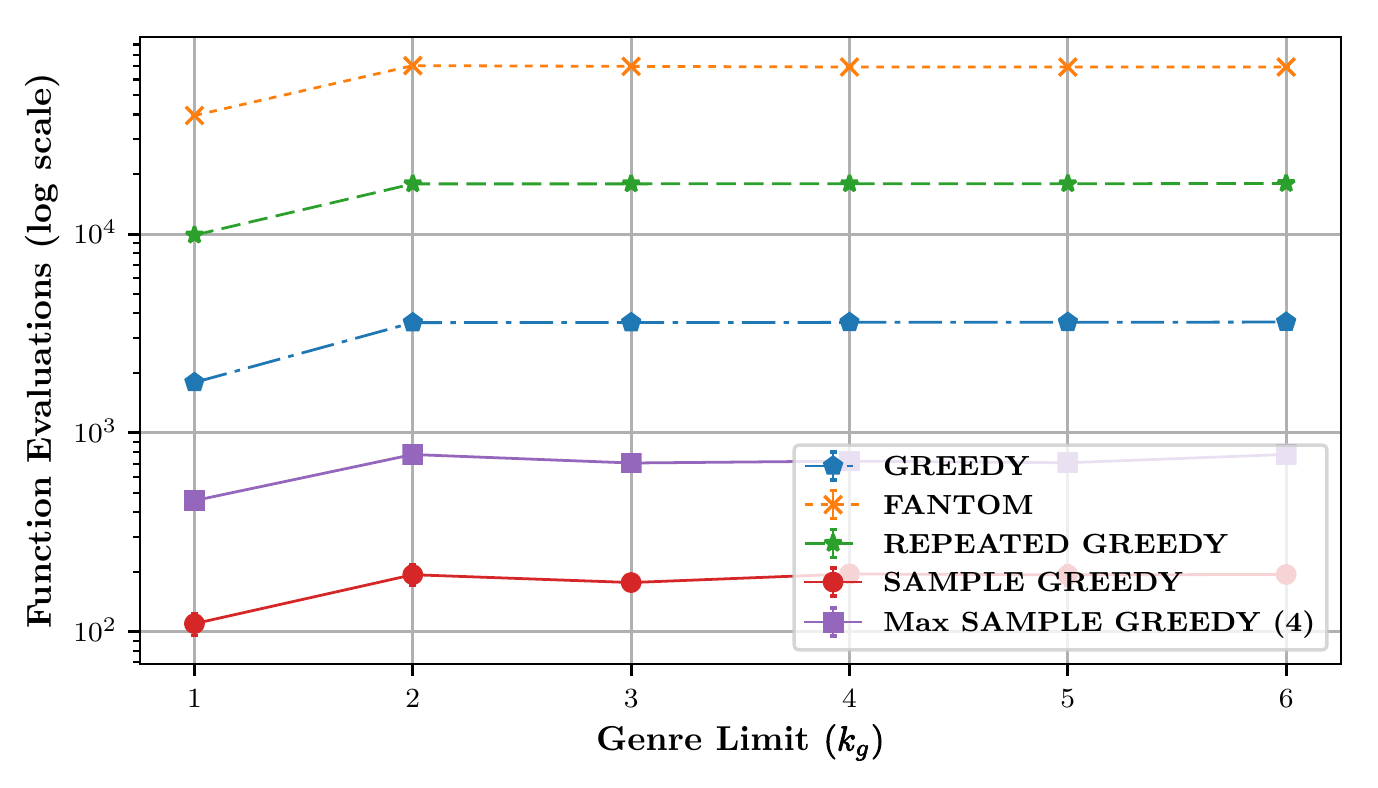} \label{fig:run_time}} \\
	\subfloat[Ratio Comparison $k_g = 1$]{\includegraphics[width=80mm]{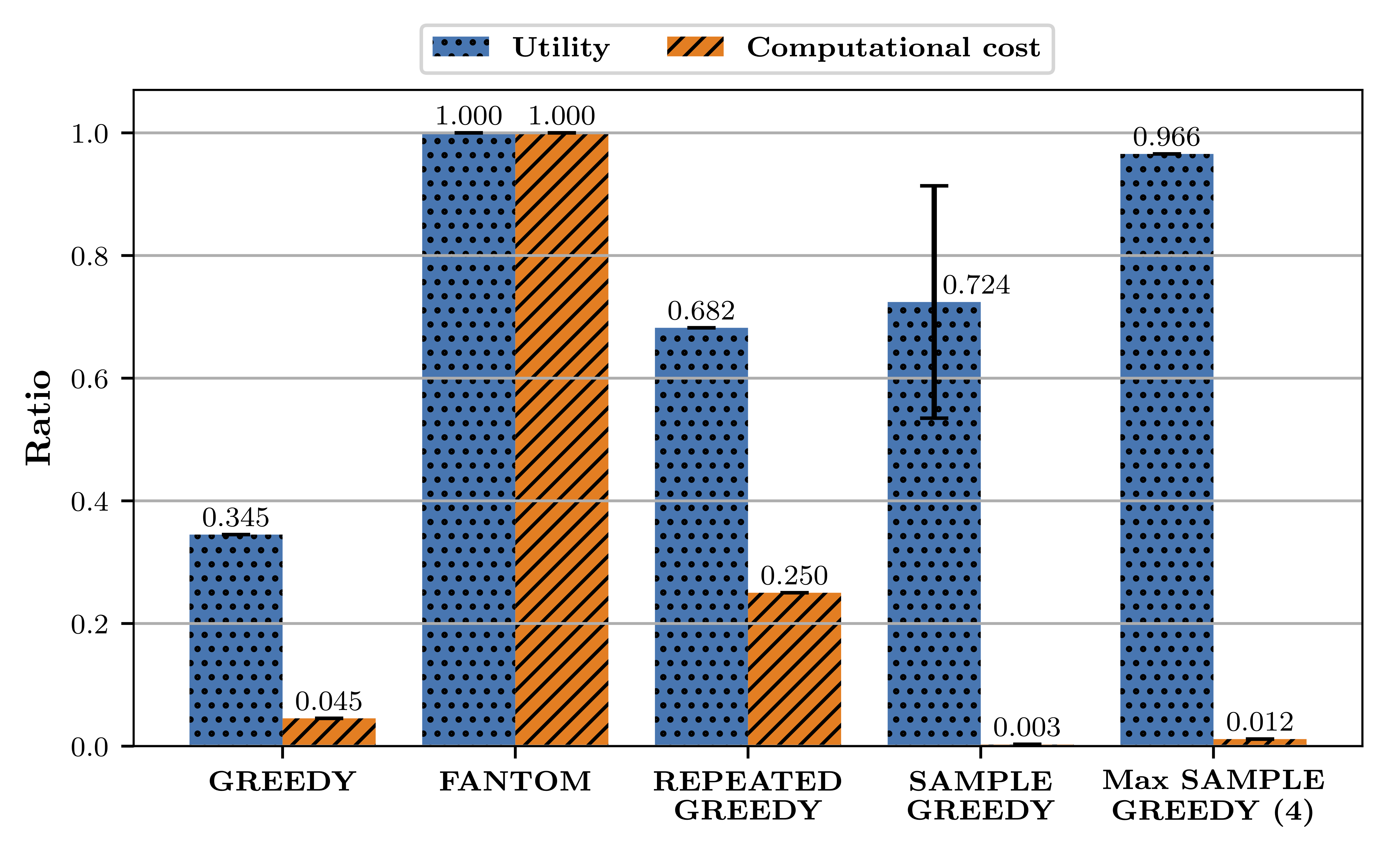}\label{fig:ratio_comp_1}}
	\subfloat[Ratio Comparison $k_g =3$]{\includegraphics[width=80mm]{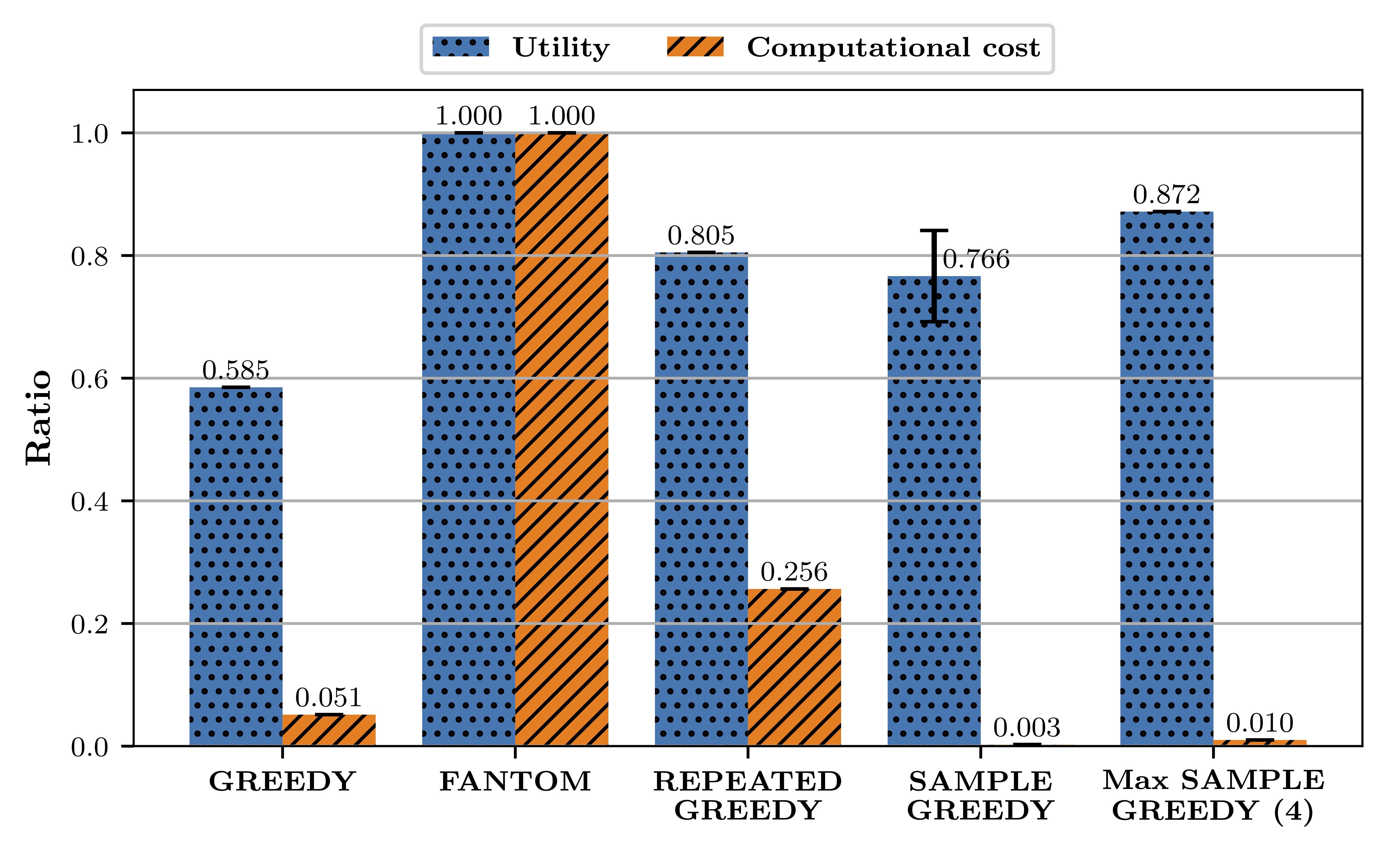}\label{fig:ratio_comp_4}}	
	\caption{Performance Comparison. \ref{fig:fun_val} shows the function value of the returned 
	solutions for tested algorithms with varying genre limit $k_g$. \ref{fig:run_time} shows the number 
	of function evaluations on a logarithmic scale with varying genre limit $k_g$. \ref{fig:ratio_comp_1} 
	and \ref{fig:ratio_comp_4} show the ratio of solution quality and cost with \algfantom serving as a baseline.}
\label{fig:performance}
\end{center}
\end{figure}

The result of this experiment is depicted in Figure~\ref{fig:performance}.
First, Figure~\ref{fig:fun_val} shows the utility value of the solution sets for the various algorithms.
 As we see from Figure~\ref{fig:fun_val}, \algfantom consistently returns a solution set with the highest 
 function value. However, \algrd and \msg return solution sets with similarly high 
 function values. 
 Even for four runs, \msg noticeably increases the utility of the returned solution.
Both \algrd and \msg return solutions with larger utility than Greedy. 
Figure~\ref{fig:run_time} shows the number of function calls made by each algorithm as the genre limit 
$k_g$ is varied. 
 For each algorithm, the number of function calls appears roughly constant as $k_g$ is varied---this is 
 likely due to the lazy greedy implementation and also to the small values used for $k_g$.
 We observe that
  \algrd runs roughly two orders magnitude faster than \algdt and three orders of magnitude faster 
  than \algfantom. 
	Moreover, even after we boost \algrd by executing it a few times, its computational cost remains much lower than that of the other algorithms.

To better analyze the trade-off between the utility of the solution value  and the computational cost, in 
Figures~\ref{fig:ratio_comp_1} and \ref{fig:ratio_comp_4}, we compare the ratio of these measurements 
for the various algorithms using \algfantom as a baseline. 
For both cases of $k_g=1$ and $k_g=3$, we see that \msg provides 
nearly the same utility as \algfantom, while only incurring around 1\% of the computational cost. 
For the case of $k_g =3$, \algrd achieves $76.6\%$ of the utility of \algfantom, while incurring only $0.3\%$  
of 
the computational cost.
Thus, we may conclude that our algorithm provides solutions whose quality is on par with current state-of-the-art, and yet they run in a small fraction of the time.

While Greedy may commit to poor solutions early on which may not be improved, \algrd avoids this (in 
expectation) by only 
considering a fraction of the ground set. 
Fortunately, the movie recommendation system has a very interpretable solution so we can observe 
this phenomenon in a qualitative manner. See Figure~\ref{fig:solution_sets} for the movies 
recommended by the different algorithms. 
Because $k_g = 1$, we are constrained here to have at most one movie from \textit{Adventure}, \textit{Animation} and \textit{Fantasy}. 
As seen in Figure~\ref{fig:solution_sets}, \algfantom and \algrd return maximum size solution sets that 
are both diverse and representative of these genres. 
On the other hand, Greedy gets stuck choosing a single movie that belongs to all three genres, 
precluding any other choice of movie from the solution set.

\begin{figure}[htb!] 
	\centering  
	\includegraphics[width=\textwidth]{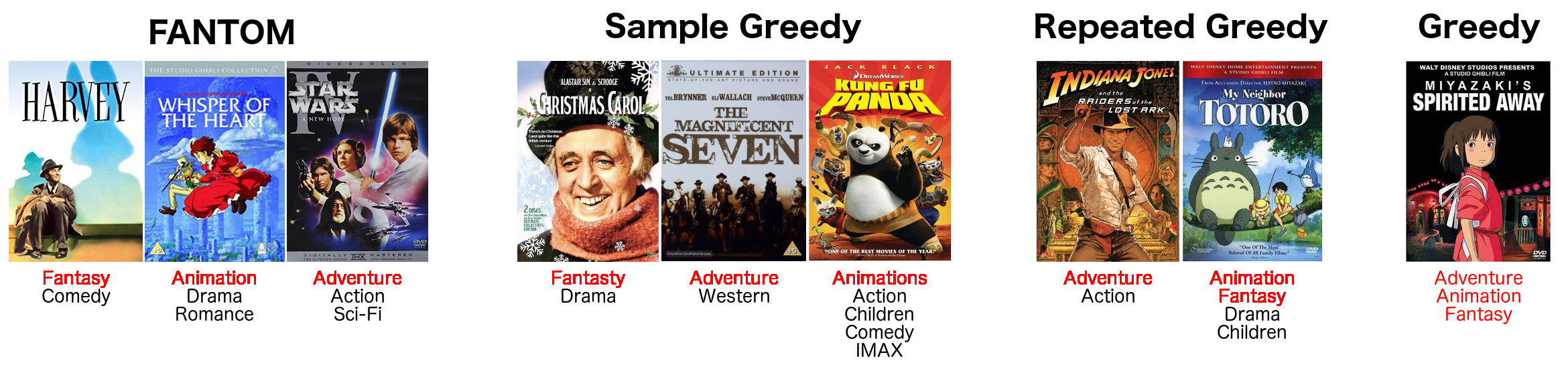} 
	\caption{Solution sets: the movies in the solution sets for $k_g=1$ returned by \algfantom, Sample 
	Greedy, Repeated Greedy and Greedy are listed here, along with genre information. The favorite 
	genres ($G_u$) are in red.} \label{fig:solution_sets}
\end{figure}

\subsection{Streaming Algorithms} \label{sec:experiments-streaming}
In  this section, we investigate the performance of our proposed streaming algorithm on two data 
summarization applications.
We demonstrate that by subsampling elements of a stream, it is possible to dramatically 
reduce the computational cost while providing solutions with similar utility as state-of-the-art 
algorithms.

For the first task, in Section~\ref{sec:stream-video}, we replicate the experiment of \citet{MJK17} and compare the 
performance of our algorithm (\AlgSampling) with the performance of the algorithm of 
\citet{MJK17}. 
Unfortunately, to allow such a comparison we had to resort to the relatively small datasets that existing algorithms can handle. 
Despite the small size of these datasets, we could still observe the superiority of our method against 
this state-of-the-art. 
In Section~\ref{sec:stream-location}, we investigate the scalability of our algorithm to larger datasets in a 
location summarization task where the other streaming algorithms are not applicable due to their larger 
computational cost.

\subsubsection{Video Summarization.} \label{sec:stream-video}

We evaluate the performance of \AlgSampling on a video summarization task and
compare it with \AlgSeq \citep{GCGS14}\footnote{\url{https://github.com/pujols/Video-summarization}} and  \AlgLocal \citep{MJK17}.\footnote{\url{https://github.com/baharanm/non-mon-stream}}
 For our experiments, we use the Open Video Project (OVP) and the YouTube datasets, which have 50 and 39 videos, respectively \citep{DLDA11}. 

Determinantal point process (DPP) is a powerful method to capture diversity in datasets \citep{M75,kulesza2012determinantal}. 
Let $\cN = \{1,2, \cdots,  n\}$ be a ground set of $n$ items.
A DPP defines a probability distribution over all subsets of $\cN$, distributed as $\Pr[Y = S] = \frac{\det 
(L_S)}{\det(I + L)}$ for every set $S \subseteq \ground$,
where $L$ is a positive semidefinite kernel matrix, $L_S$  is the principal sub-matrix of $L$ indexed by 
$S$, and $I$ is the $n \times n$ identity matrix. 
The most diverse subset of $\cN$ is the one with the maximum probability in this distribution. 
Although finding this set is NP-hard \citep{KLQ95}, the function $f(S) = \log \det (L_S)$ is a 
non-monotone submodular function \citep{kulesza2012determinantal}.

We follow the experimental setup of  \citep{GCGS14} for extracting frames from videos, finding a linear kernel matrix  $L$ and evaluating the quality of produced summaries based on their F-score.
\citet{GCGS14} define a sequential DPP, where each video sequence is partitioned into disjoint segments of equal sizes. For selecting a subset $S_t$ from each segment $t$ (i.e., set $\cP_t$), a DPP is defined on the union of the frames in this segment and the selected frames $S_{t-1}$  from the previous segment. 
Therefore, the conditional distribution of $S_t$ is given by, $\Pr[S_t|S_{t-1}] =  \frac{\det (L_{S_t \cup S_{t-1}})}{\det(I_t + L)},$ where $L$ is the kernel matrix defined over $\cP_t \cup S_{t-1}$, and $I_t$ is a diagonal matrix of the same size as $\cP_t \cup S_{t-1} $ in which the elements corresponding to $S_{t-1}$ are zeros and
the elements corresponding to $\cP_t$ are $1$.
For the detailed explanation, please refer to \citep{GCGS14}.
In our experiments, we focus on maximizing the non-monotone submodular function $f(S_t) = \log \det (L_{S_t \cup S_{t-1}})$. 
We remark that this function can take negative values, violating the non-negativity condition required
for our theoretical guarantees.

We first compare the objective values  (F-scores) of \AlgSampling and \AlgLocal for different segment sizes over YouTube and OVP datasets. In each experiment, the values are normalized to the F-score of summaries generated by \AlgSeq. 
While \AlgSeq has the best performance in terms of maximizing the objective value, in Figures~\ref{fig:stream-youtube} and \ref{fig:stream-ovp}, we observe that both  \AlgSampling and  \AlgLocal produce summaries with very high qualities. 
Figure~\ref{fig:video-60} shows the summary produced by our algorithm for OVP video number 60.
\citet{MJK17} showed that their algorithm (\AlgLocal) runs three orders of magnitude faster than \AlgSeq \citep{GCGS14}.
 In our experiments (see Figure~\ref{fig:stream-running}), we observed that \AlgSampling is 40 and 50 
 times faster than \AlgLocal  for the YouTube and OVP datasets, respectively, which demonstrates the 
 power of subsampling  in reducing the computational cost. 
 Note that  for different segment sizes the number of frames remains constant; therefore, the time complexities for both \AlgSampling and \AlgLocal do not change.

\begin{figure}[htb!] 
	\begin{center}
		\subfloat[YouTube videos] {\includegraphics[width=80mm]{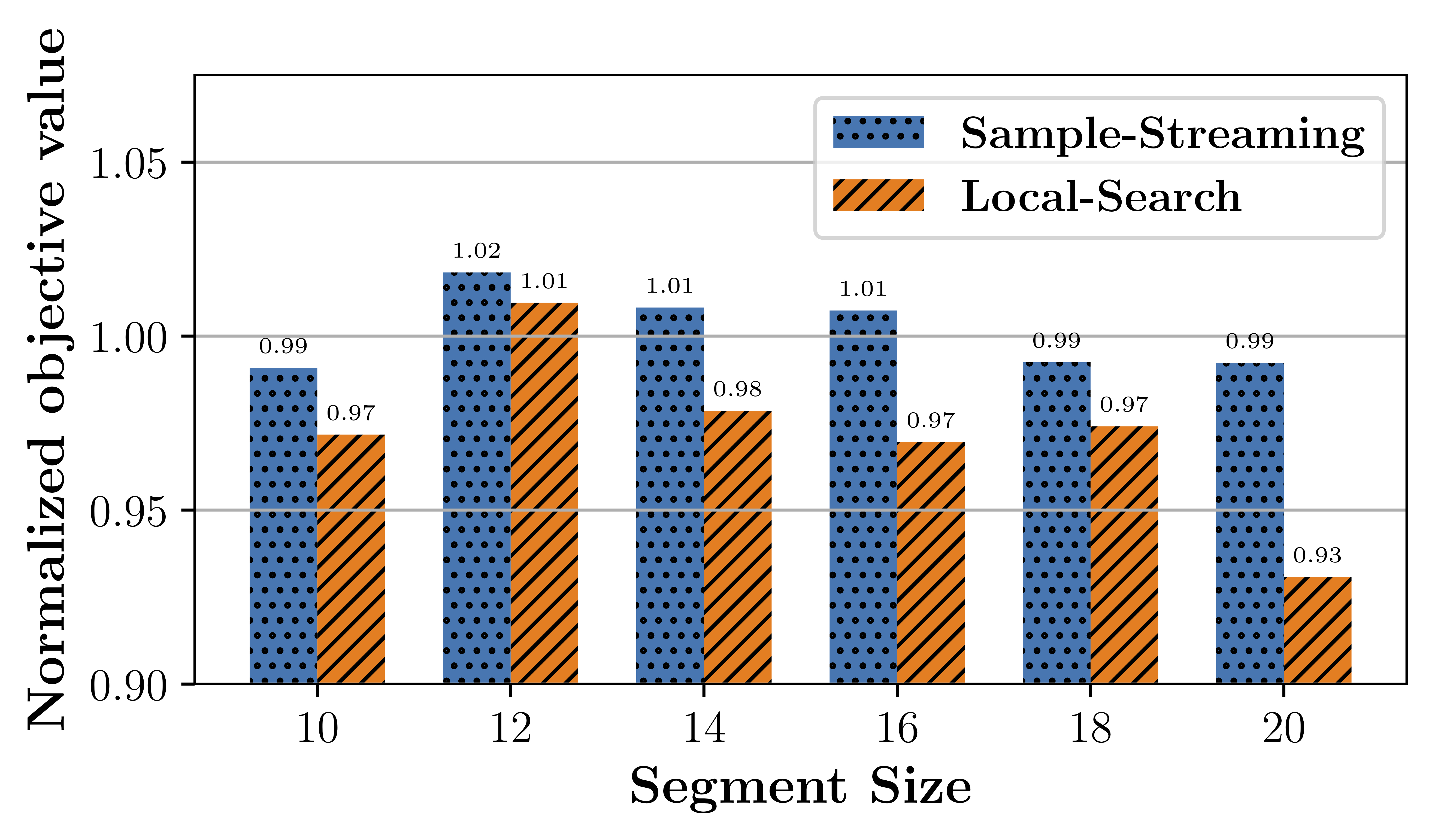} \label{fig:stream-youtube}}
		\subfloat[OVP videos]{\includegraphics[width=80mm]{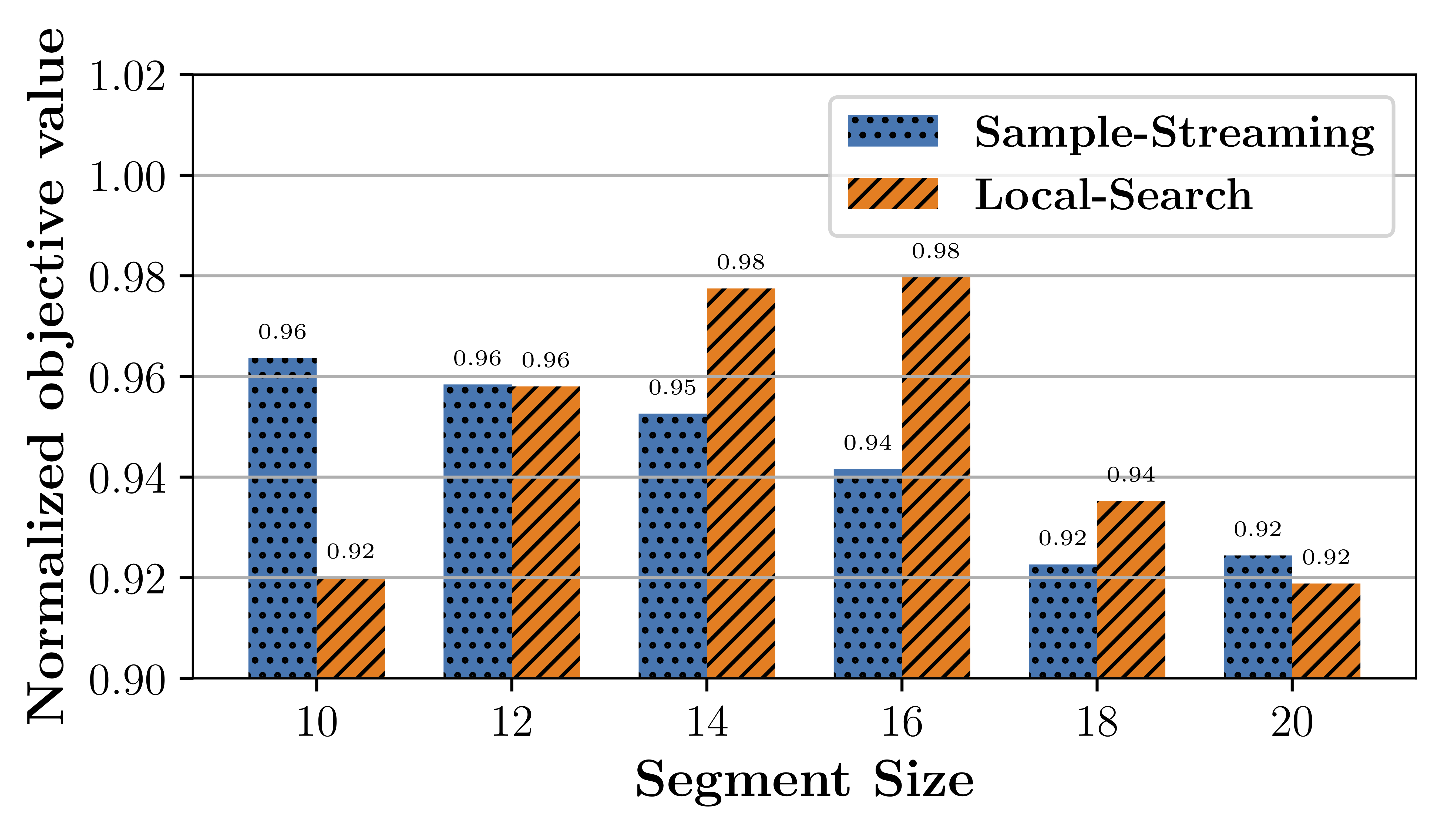} \label{fig:stream-ovp}} \\
		\subfloat[Running time]{\includegraphics[width=60mm]{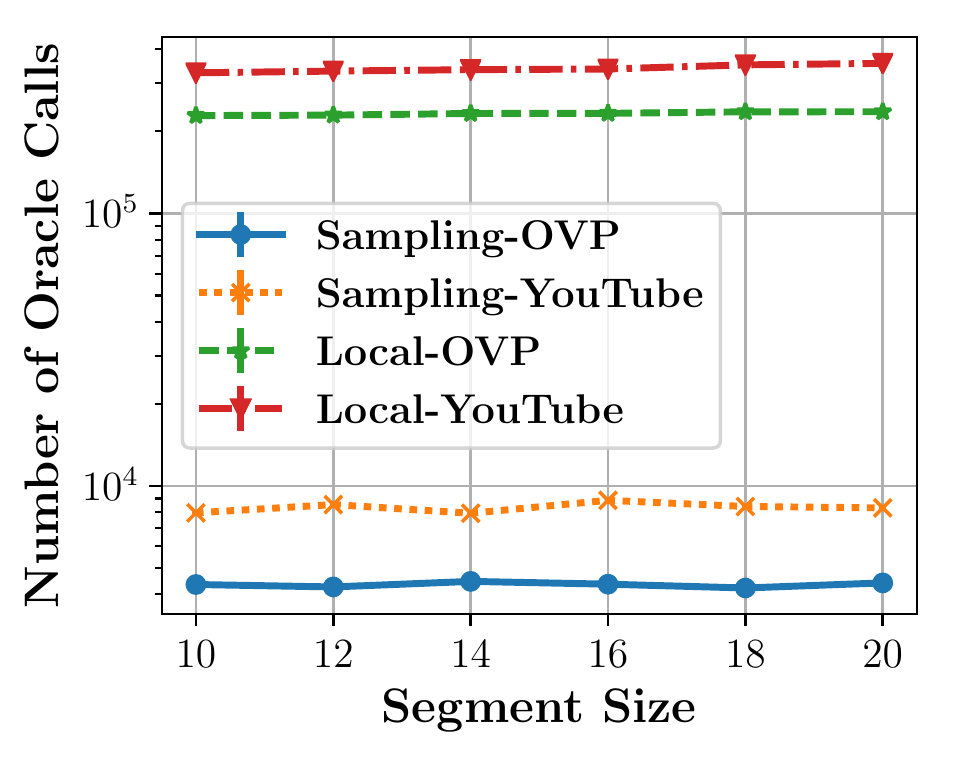}\label{fig:stream-running}}
		\caption{Comparing the normalized objective value and running time
			of \AlgSampling and \AlgLocal for different segment sizes.}
	\end{center}
\end{figure}

\begin{figure}[htb!]
	\centering
	\includegraphics[height=.5in]{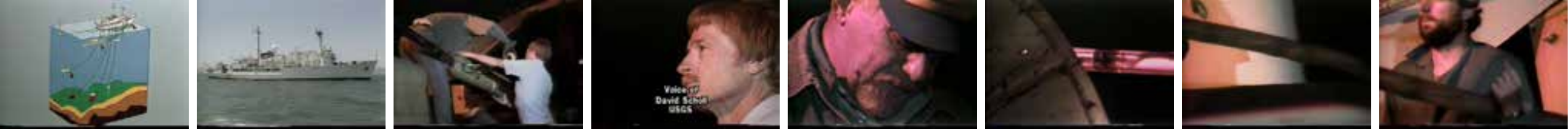}
	\caption{Summary generated by \AlgSampling for OVP video number 60.}
	\label{fig:video-60}
\end{figure}

\begin{figure}[t!]
	\centering
	\includegraphics[height=1.7in]{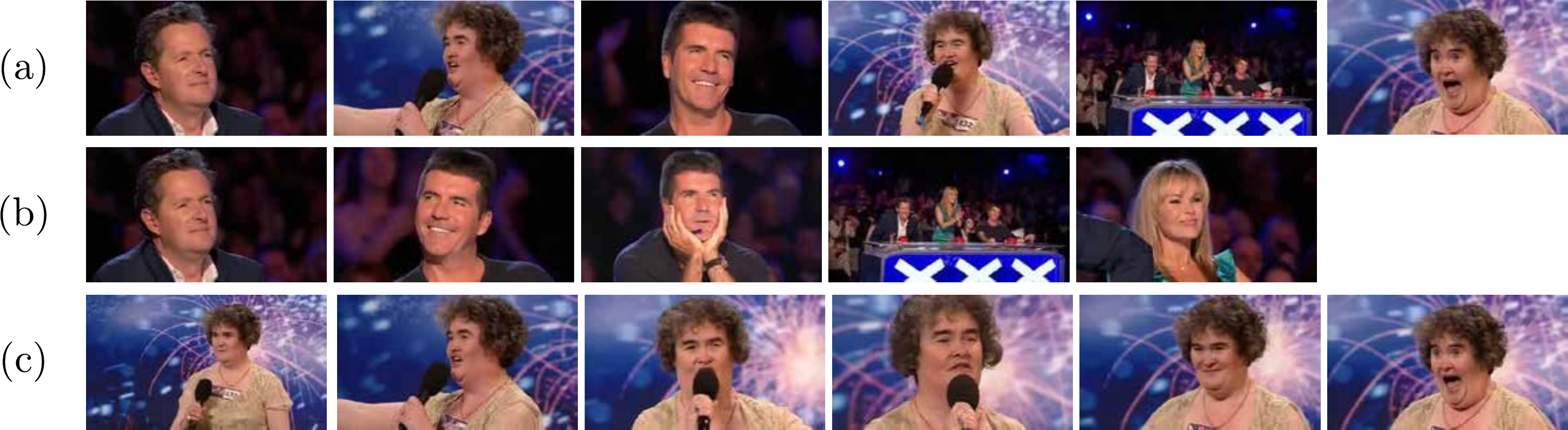}
	\caption{Summaries generated by \AlgSampling for YouTube video number 106: (a)  a $6$-matchoid constraint, (b) a $3$-matchoid constraint and (c) a partition matroid constraint.}
	\label{fig:constraint}
\end{figure}

In the second experiment, we study the effect of imposing different constraints on video 
summarization task for YouTube video number 106, which is a part of the television series,  ``Britain's 
Got Talent''.
We consider a constraint on the faces which appear in the summary's frames, using the same methods as 
described by \citet{MJK17} for face recognition.
Note that a frame may contain more than one face.
In the first set of constraints, we restrict to $3$ the number of appearances in the summary of each face, which results in a 
6-matchoid since there are six unique faces in the video. Figure~\ref{fig:constraint}(a) shows the summary produced for this task. 
We also produce a summary limiting the number of frames which contain each judge.
Such a constraint forms a 3-matchoid and this summary is shown in  Figure~\ref{fig:constraint}(b).
Finally, Figure~\ref{fig:constraint}(c) shows a summary with a constraint allowing only the singer's face.

\subsubsection{Location Summarization.} \label{sec:stream-location}

In the second task, we consider a massive ridesharing dataset, for which current state-of-the-art 
streaming algorithms are infeasible to run.
Given a dataset of \num[group-separator={,}]{504247} Uber pick ups in Manhattan, New York in April 
2014 \cite{uberpickup}, our goal is to find a set of the most representative locations. 
This dataset allows us to study the power of subsampling and the effect of $p$ and $k$ (the size of the largest feasible solution) on the performance of our algorithm. 

To do so, the entire area of the given pick ups is covered by $m = 166$ overlapping circular regions of radius $r$ (the centers of these regions provided a \num{1}km-cover of all the area, \ie, for each location in the dataset there was at least one center within a distance of \num{1}km from it), and the algorithm was allowed to choose at most $\ell$ locations out of each one of these regions. One can observe that by using a single matroid for limiting the number of locations chosen within each one of the regions, the above constraint can be expressed as a $p$-matchoid constraint, where $p$ is the maximum number of regions a single location can belong to (notice that $p$ could be much smaller than the total number $m$ of regions). 

In order to find a representative set $S$, we use the following monotone submodular objective function: 
$
f(S) = \log \det(I + \alpha K_{S,S}), 
$
where the matrix $K$ encodes the similarities between data points, $K_{S,S}$ is the principal sub-matrix of $K$ indexed by $S$ and $\alpha > 0$ is a regularization parameter \citep{herbrich2003fast,seeger2004greedy,krause2005near}. The similarity of two location samples $i$ and $j$ is defined by a Gaussian kernel
 $K_{i,j} = \exp{( - d_{i,j}^2 / h^2)}$,  where the distance $d_{i,j}$ (in meters) is calculated from the coordinates and $h$ is set to $5000$. 

In the first location summarization experiment, we set the radius of regions to $r = 1.5$km. In this setting, we observed that a point belongs to at most \num{7}  regions; hence, the constraint is a $7$-matchoid. For $\ell = 5$, it took \num{116} seconds\footnote{In these experiments, we used a machine powered by Intel i5, 3.2 GHz processor and 16 GB of RAM.}  (and \num[group-separator={,}]{693717} oracle calls) for our algorithm to find a summary of size $k = 153$. Additionally, for $\ell = 10$ and $\ell = 20$ it took $294$ seconds (and \num[group-separator={,}]{1306957} oracle calls) and \num{1004} seconds (and  \num[group-separator={,}]{2367389} oracle calls), respectively, for the algorithm to produce summaries of sizes \num{301} and \num{541}, respectively. 

In the second location summarization experiment, we set the radius of regions to $r = 2.5$km to investigate the performance of our algorithm on $p$-matchoids with larger values of $p$.
 In this setting, we observed that a point belongs to at most \num{17} regions, which made the constraint a \num{17}-matchoid. This time, for $\ell = 5$, it took only \num{35} seconds (and \num[group-separator={,}]{296023} oracle calls) for our algorithm to find a summary of size $k = 54$. 
Additionally, for $\ell = 10$ and $\ell = 20$ it took 80 seconds (and \num[group-separator={,}]{526839} oracle calls) and \num{176} seconds (and \num[group-separator={,}]{958549} oracle calls), respectively, for the algorithm to produce summaries of sizes \num{106} and \num{198}, respectively. 
As one can observe, our algorithm scales well to larger datasets. 
Also, for $p$-matchoids with larger $p$ (which results in a smaller sampling probability $q$) the performance gets even better. 

\section{Conclusion}  \label{sec:conclusion}

We presented subsampling as an alternative algorithmic technique for constrained submodular 
maximization.
We argued that subsampling provides a unified framework for achieving tight or nearly-tight approximation guarantees 
for both monotone and non-monotone objectives in several computational settings.
To this end, we proposed two algorithms which use the subsampling technique:
\algrd for submodular maximization under a $p$-extendible system in the offline setting,
and \AlgSampling for submodular maximization under a $p$-matchoid constraint in the streaming 
setting.
Our theoretical analysis shows that subsampling can be used to achieve state-of-the-art approximation 
guarantees for these settings at significantly lower computational cost than previous methods.
Our experimental results demonstrate that for a variety of practical problems, algorithms featuring the 
subsampling technique produce solutions with similar quality to other existing algorithms at a fraction of the computational cost.
This work shows that subsampling is a powerful technique which can be used to scale submodular 
optimization to larger problem instances with more complex constraints.

\section*{Acknowledgments.}
This work was supported in part by
NSF (IIS- 1845032), ONR (N00014-19-1-2406), and AFOSR (FA9550-18-1-0160) awarded to Amin Karbasi 
and
an NSF Graduate Research Fellowship (DGE1122492) awarded to Christopher Harshaw. 

\bibliography{tex/main}

\begin{thebibliography}{68}
\providecommand{\natexlab}[1]{#1}
\providecommand{\url}[1]{\texttt{#1}}
\expandafter\ifx\csname urlstyle\endcsname\relax
  \providecommand{\doi}[1]{doi: #1}\else
  \providecommand{\doi}{doi: \begingroup \urlstyle{rm}\Url}\fi

\bibitem[Badanidiyuru(2011)]{V11}
Ashwinkumar Badanidiyuru.
\newblock Buyback problem - approximate matroid intersection with cancellation
  costs.
\newblock In \emph{ICALP}, pages 379--390, 2011.

\bibitem[Badanidiyuru and Vondr{\'{a}}k(2014)]{Badanidiyuru14}
Ashwinkumar Badanidiyuru and Jan Vondr{\'{a}}k.
\newblock Fast algorithms for maximizing submodular functions.
\newblock In \emph{SODA}, pages 1497--1514, 2014.

\bibitem[Badanidiyuru et~al.(2014)Badanidiyuru, Mirzasoleiman, Karbasi, and
  Krause]{BMKK14}
Ashwinkumar Badanidiyuru, Baharan Mirzasoleiman, Amin Karbasi, and Andreas
  Krause.
\newblock Streaming submodular maximization: massive data summarization on the
  fly.
\newblock In \emph{KDD}, pages 671--680, 2014.

\bibitem[Badanidiyuru et~al.(2020)Badanidiyuru, Karbasi, Kazemi, and
  Vondrak]{badanidiyuru2020submodular}
Ashwinkumar Badanidiyuru, Amin Karbasi, Ehsan Kazemi, and Jan Vondrak.
\newblock Submodular maximization through barrier functions.
\newblock In \emph{Advances in Neural Information Processing Systems}, pages
  524--534, 2020.

\bibitem[Balkanski et~al.(2016)Balkanski, Mirzasoleiman, Krause, and
  Singer]{BMKS16}
Eric Balkanski, Baharan Mirzasoleiman, Andreas Krause, and Yaron Singer.
\newblock Learning sparse combinatorial representations via two-stage
  submodular maximization.
\newblock In \emph{Proceedings of The 33rd International Conference on Machine
  Learning}, pages 2207--2216, 2016.

\bibitem[Buchbinder and Feldman(2018)]{buchbinder2018submodular}
Niv Buchbinder and Moran Feldman.
\newblock Submodular functions maximization problems., 2018.

\bibitem[Buchbinder and Feldman(2019)]{BF19}
Niv Buchbinder and Moran Feldman.
\newblock Constrained submodular maximization via a nonsymmetric technique.
\newblock \emph{Math. Oper. Res.}, 44\penalty0 (3):\penalty0 988--1005, 2019.
\newblock URL \url{https://doi.org/10.1287/moor.2018.0955}.

\bibitem[Buchbinder et~al.(2014)Buchbinder, Feldman, Naor, and
  Schwartz]{BFNS14}
Niv Buchbinder, Moran Feldman, Joseph Naor, and Roy Schwartz.
\newblock {Submodular Maximization with Cardinality Constraints}.
\newblock In \emph{SODA}, pages 1433--1452, 2014.

\bibitem[Buchbinder et~al.(2015)Buchbinder, Feldman, and Schwartz]{BFS15}
Niv Buchbinder, Moran Feldman, and Roy Schwartz.
\newblock Online submodular maximization with preemption.
\newblock In \emph{SODA}, pages 1202--1216, 2015.

\bibitem[Buchbinder et~al.(2016)Buchbinder, Feldman, and
  Schwartz]{buchbinder2016comparing}
Niv Buchbinder, Moran Feldman, and Roy Schwartz.
\newblock {Comparing Apples and Oranges: Query Trade-off in Submodular
  Maximization}.
\newblock \emph{Mathematics of Operations Research}, 2016.

\bibitem[Buchbinder et~al.(2019)Buchbinder, Feldman, Filmus, and Garg]{BFG19}
Niv Buchbinder, Moran Feldman, Yuval Filmus, and Mohit Garg.
\newblock Online submodular maximization: Beating 1/2 made simple.
\newblock In \emph{IPCO}, pages 101--114, 2019.

\bibitem[C{\u{a}}linescu et~al.(2011)C{\u{a}}linescu, Chekuri, P{\'{a}}l, and
  Vondr{\'{a}}k]{CCPV11}
Gruia C{\u{a}}linescu, Chandra Chekuri, Martin P{\'{a}}l, and Jan
  Vondr{\'{a}}k.
\newblock Maximizing a monotone submodular function subject to a matroid
  constraint.
\newblock \emph{{SIAM} J. Comput.}, 40\penalty0 (6):\penalty0 1740--1766, 2011.

\bibitem[Cand\'{e}s and Recht(2008)]{Candes2008}
E.~Cand\'{e}s and B.~Recht.
\newblock Exact matrix completion via convex optimization.
\newblock In \emph{Foundations of Computational Mathematics}, 2008.

\bibitem[Chakrabarti and Kale(2015)]{CK15}
Amit Chakrabarti and Sagar Kale.
\newblock Submodular maximization meets streaming: matchings, matroids, and
  more.
\newblock \emph{Math. Program.}, 154\penalty0 (1-2):\penalty0 225--247, 2015.

\bibitem[Chekuri et~al.(2015)Chekuri, Gupta, and Quanrud]{CGQ15}
Chandra Chekuri, Shalmoli Gupta, and Kent Quanrud.
\newblock Streaming algorithms for submodular function maximization.
\newblock In \emph{ICALP}, pages 318--330, 2015.

\bibitem[De~Avila et~al.(2011)De~Avila, Lopes, da~Luz~Jr, and
  de~Albuquerque~Ara{\'u}jo]{DLDA11}
Sandra Eliza~Fontes De~Avila, Ana Paula~Brand{\~a}o Lopes, Antonio da~Luz~Jr,
  and Arnaldo de~Albuquerque~Ara{\'u}jo.
\newblock {VSUMM:} {A} mechanism designed to produce static video summaries and
  a novel evaluation method.
\newblock \emph{Pattern Recognition Letters}, 32\penalty0 (1):\penalty0 56--68,
  2011.

\bibitem[El-Arini et~al.(2009)El-Arini, Veda, Shahaf, and
  Guestrin]{el2009turning}
Khalid El-Arini, Gaurav Veda, Dafna Shahaf, and Carlos Guestrin.
\newblock Turning down the noise in the blogosphere.
\newblock In \emph{international conference on Knowledge discovery and data
  mining (KDD)}, pages 289--298, 2009.

\bibitem[Ene and Nguyen(2016)]{EN16}
Alina Ene and Huy~L. Nguyen.
\newblock Constrained submodular maximization: Beyond 1/e.
\newblock In \emph{FOCS}, pages 248--257, 2016.

\bibitem[Ene and Nguyen(2018)]{EN18}
Alina Ene and Huy~L. Nguyen.
\newblock Towards nearly-linear time algorithms for submodular maximization
  with a matroid constraint.
\newblock \emph{CoRR}, abs/1811.07464, 2018.
\newblock URL \url{http://arxiv.org/abs/1811.07464}.

\bibitem[Feldman et~al.(2011{\natexlab{a}})Feldman, Naor, and Schwartz]{FNS11}
Moran Feldman, Joseph Naor, and Roy Schwartz.
\newblock A unified continuous greedy algorithm for submodular maximization.
\newblock In \emph{FOCS}, pages 570--579, 2011{\natexlab{a}}.

\bibitem[Feldman et~al.(2011{\natexlab{b}})Feldman, Naor, Schwartz, and
  Ward]{FNSW11}
Moran Feldman, Joseph Naor, Roy Schwartz, and Justin Ward.
\newblock Improved approximations for k-exchange systems - (extended abstract).
\newblock In \emph{ESA}, pages 784--798, 2011{\natexlab{b}}.

\bibitem[Feldman et~al.(2017)Feldman, Harshaw, and Karbasi]{FHK17}
Moran Feldman, Christopher Harshaw, and Amin Karbasi.
\newblock Greed is good: Near-optimal submodular maximization via greedy
  optimization.
\newblock In \emph{COLT}, pages 758--784, 2017.

\bibitem[Feldman et~al.(2018)Feldman, Karbasi, and Kazemi]{FKK18}
Moran Feldman, Amin Karbasi, and Ehsan Kazemi.
\newblock Do less, get more: Streaming submodular maximization with
  subsampling.
\newblock In \emph{Proceedings of the 32nd International Conference on Neural
  Information Processing Systems}, NIPS’18, page 730–740, 2018.

\bibitem[Fisher et~al.(1978)Fisher, Nemhauser, and Wolsey]{FNW78}
M.~L. Fisher, G.~L. Nemhauser, and L.~A. Wolsey.
\newblock An analysis of approximations for maximizing submodular set functions
  -- {II}.
\newblock \emph{Mathematical Programming Study}, 8:\penalty0 73--87, 1978.

\bibitem[Gomez-Rodriguez et~al.(2010)Gomez-Rodriguez, Leskovec, and
  Krause]{gomez10}
Manuel Gomez-Rodriguez, Jure Leskovec, and Andreas Krause.
\newblock Inferring networks of diffusion and influence.
\newblock In \emph{international conference on Knowledge discovery and data
  mining (KDD)}, 2010.

\bibitem[Gong et~al.(2014)Gong, Chao, Grauman, and Sha]{GCGS14}
Boqing Gong, Wei{-}Lun Chao, Kristen Grauman, and Fei Sha.
\newblock Diverse sequential subset selection for supervised video
  summarization.
\newblock In \emph{NIPS}, pages 2069--2077, 2014.

\bibitem[Guestrin et~al.(2005)Guestrin, Krause, and Singh]{guestrin2005near}
Carlos Guestrin, Andreas Krause, and Ajit~Paul Singh.
\newblock {Near-Optimal Sensor Placements in Gaussian Processes}.
\newblock In \emph{International Conference on Machine Learning (ICML)}, 2005.

\bibitem[Gupta et~al.(2010)Gupta, Roth, Schoenebeck, and Talwar]{GRST10}
Anupam Gupta, Aaron Roth, Grant Schoenebeck, and Kunal Talwar.
\newblock {Constrained Non-monotone Submodular Maximization: Offline and
  Secretary Algorithms}.
\newblock In \emph{WINE}, pages 246--257, 2010.

\bibitem[Haba et~al.(2020)Haba, Kazemi, Feldman, and
  Karbasi]{haba2020streaming}
Ran Haba, Ehsan Kazemi, Moran Feldman, and Amin Karbasi.
\newblock Streaming submodular maximization under a k-set system constraint.
\newblock In \emph{International Conference on Machine Learning}, 2020.

\bibitem[Hastie et~al.(2015)Hastie, Mazumder, Lee, and Zadeh]{Hastie2015}
Trevor Hastie, Rahul Mazumder, Jason~D. Lee, and Rexa Zadeh.
\newblock Matrix completion and low-rank svd via fast alternating least
  squares.
\newblock In \emph{Journal of Machine Learning Research}, 2015.

\bibitem[Herbrich et~al.(2003)Herbrich, Lawrence, and Seeger]{herbrich2003fast}
Ralf Herbrich, Neil~D Lawrence, and Matthias Seeger.
\newblock Fast sparse gaussian process methods: The informative vector machine.
\newblock In \emph{Advances in neural information processing systems}, pages
  625--632, 2003.

\bibitem[Kazemi et~al.(2018)Kazemi, Zadimoghaddam, and
  Karbasi]{kazemi2018scalable}
Ehsan Kazemi, Morteza Zadimoghaddam, and Amin Karbasi.
\newblock {Scalable Deletion-Robust Submodular Maximization: Data Summarization
  with Privacy and Fairness Constraints}.
\newblock In \emph{ICML}, pages 2549--2558, 2018.

\bibitem[Kazemi et~al.(2019)Kazemi, Mitrovic, Zadimoghaddam, Lattanzi, and
  Karbasi]{kazemi2019submodular}
Ehsan Kazemi, Marko Mitrovic, Morteza Zadimoghaddam, Silvio Lattanzi, and Amin
  Karbasi.
\newblock {Submodular Streaming in All Its Glory: Tight Approximation, Minimum
  Memory and Low Adaptive Complexity}.
\newblock In \emph{International Conference on Machine Learning (ICML)}, pages
  3311--3320, 2019.

\bibitem[Kazemi et~al.(2020)Kazemi, Minaee, Feldman, and
  Karbasi]{kazemi2020regularized}
Ehsan Kazemi, Shervin Minaee, Moran Feldman, and Amin Karbasi.
\newblock Regularized submodular maximization at scale.
\newblock \emph{arXiv preprint arXiv:2002.03503}, 2020.

\bibitem[Kempe et~al.(2003)Kempe, Kleinberg, and Tardos]{kempe03}
David Kempe, Jon Kleinberg, and \'{E}va Tardos.
\newblock Maximizing the spread of influence through a social network.
\newblock In \emph{international conference on Knowledge discovery and data
  mining (KDD)}, pages 137--146, 2003.

\bibitem[Kirchhoff and Bilmes(2014)]{kirchhoff2014submodularity}
Katrin Kirchhoff and Jeff Bilmes.
\newblock {Submodularity for data selection in statistical machine
  translation}.
\newblock In \emph{Proceedings of EMNLP}, 2014.

\bibitem[Ko et~al.(1995)Ko, Lee, and Queyranne]{KLQ95}
Chun{-}Wa Ko, Jon Lee, and Maurice Queyranne.
\newblock An exact algorithm for maximum entropy sampling.
\newblock \emph{Operations Research}, 43\penalty0 (4):\penalty0 684--691, 1995.

\bibitem[Korula et~al.(2018)Korula, Mirrokni, and Zadimoghaddam]{KMZ18}
Nitish Korula, Vahab~S. Mirrokni, and Morteza Zadimoghaddam.
\newblock Online submodular welfare maximization: Greedy beats 1/2 in random
  order.
\newblock \emph{{SIAM} J. Comput.}, 47\penalty0 (3):\penalty0 1056--1086, 2018.
\newblock URL \url{https://doi.org/10.1137/15M1051142}.

\bibitem[Krause and Guestrin(2005)]{krause2005near}
Andreas Krause and Carlos Guestrin.
\newblock Near-optimal nonmyopic value of information in graphical models.
\newblock In \emph{{UAI} '05, Proceedings of the 21st Conference in Uncertainty
  in Artificial Intelligence, Edinburgh, Scotland, July 26-29, 2005}, pages
  324--331, 2005.

\bibitem[Kulesza and Taskar(2012)]{kulesza2012determinantal}
Alex Kulesza and Ben Taskar.
\newblock Determinantal point processes for machine learning.
\newblock \emph{Foundations and Trends in Machine Learning}, 5\penalty0 (2--3),
  2012.

\bibitem[Lee et~al.(2010{\natexlab{a}})Lee, Mirrokni, Nagarajan, and
  Sviridenko]{LMNS10}
Jon Lee, Vahab~S. Mirrokni, Viswanath Nagarajan, and Maxim Sviridenko.
\newblock Maximizing nonmonotone submodular functions under matroid or knapsack
  constraints.
\newblock \emph{{SIAM} J. Discrete Math.}, 23\penalty0 (4):\penalty0
  2053--2078, 2010{\natexlab{a}}.

\bibitem[Lee et~al.(2010{\natexlab{b}})Lee, Sviridenko, and
  Vondr{\'{a}}k]{LSV10}
Jon Lee, Maxim Sviridenko, and Jan Vondr{\'{a}}k.
\newblock {Submodular Maximization over Multiple Matroids via Generalized
  Exchange Properties}.
\newblock \emph{Math. Oper. Res.}, 35\penalty0 (4):\penalty0 795--806,
  2010{\natexlab{b}}.

\bibitem[Leskovec et~al.(2007)Leskovec, Krause, Guestrin, Faloutsos,
  VanBriesen, and Glance]{leskovec2007cost}
Jure Leskovec, Andreas Krause, Carlos Guestrin, Christos Faloutsos, Jeanne
  VanBriesen, and Natalie Glance.
\newblock Cost-effective outbreak detection in networks.
\newblock In \emph{international conference on Knowledge discovery and data
  mining (KDD)}, pages 420--429, 2007.

\bibitem[Libbrecht et~al.(2018)Libbrecht, Bilmes, and Noble]{LBN18}
Maxwell~W. Libbrecht, Jeffrey~A. Bilmes, and William~Stafford Noble.
\newblock Choosing non-redundant representative subsets of protein sequence
  data sets using submodular optimization.
\newblock \emph{Proteins: Structure, Function, and Bioinformatics}, 2018.
\newblock ISSN 1097-0134.

\bibitem[Lin and Bilmes(2011)]{LB11}
Hui Lin and Jeff~A. Bilmes.
\newblock {A Class of Submodular Functions for Document Summarization}.
\newblock In \emph{HLT}, pages 510--520, 2011.

\bibitem[Lindgren et~al.(2015)Lindgren, Wu, and Dimakis]{Lindgren2015}
Erik~M Lindgren, Shanshan Wu, and Alexandros~G Dimakis.
\newblock {Sparse and greedy: Sparsifying submodular facility location
  problems}.
\newblock In \emph{NIPS Workshop on Optimization for Machine Learning}, 2015.

\bibitem[Macchi(1975)]{M75}
Odile Macchi.
\newblock The coincidence approach to stochastic point processes.
\newblock \emph{Advances in Applied Probability}, 7\penalty0 (1):\penalty0
  83--122, 1975.

\bibitem[Mestre(2006)]{Mestre2006}
Juli{\'a}n Mestre.
\newblock {Greedy in Approximation Algorithms}.
\newblock In \emph{European Symposium on Algorithms (ESA)}, pages 528--539.
  2006.

\bibitem[Mirzasoleiman et~al.(2015)Mirzasoleiman, Badanidiyuru, Karbasi,
  Vondr{\'{a}}k, and Krause]{Mirzasoleiman15}
Baharan Mirzasoleiman, Ashwinkumar Badanidiyuru, Amin Karbasi, Jan
  Vondr{\'{a}}k, and Andreas Krause.
\newblock {Lazier Than Lazy Greedy}.
\newblock In \emph{AAAI Conference on Artificial Intelligence}, pages
  1812--1818, 2015.

\bibitem[Mirzasoleiman et~al.(2016{\natexlab{a}})Mirzasoleiman, Badanidiyuru,
  and Karbasi]{MBK16}
Baharan Mirzasoleiman, Ashwinkumar Badanidiyuru, and Amin Karbasi.
\newblock {Fast Constrained Submodular Maximization: Personalized Data
  Summarization}.
\newblock In \emph{ICML}, pages 1358--1367, 2016{\natexlab{a}}.

\bibitem[Mirzasoleiman et~al.(2016{\natexlab{b}})Mirzasoleiman, Karbasi,
  Sarkar, and Krause]{MKSK16}
Baharan Mirzasoleiman, Amin Karbasi, Rik Sarkar, and Andreas Krause.
\newblock {Distributed Submodular Maximization}.
\newblock \emph{Journal of Machine Learning Research}, 17:\penalty0
  238:1--238:44, 2016{\natexlab{b}}.

\bibitem[Mirzasoleiman et~al.(2018)Mirzasoleiman, Jegelka, and Krause]{MJK17}
Baharan Mirzasoleiman, Stefanie Jegelka, and Andreas Krause.
\newblock {Streaming Non-Monotone Submodular Maximization: Personalized Video
  Summarization on the Fly}.
\newblock In \emph{{AAAI} Conference on Artificial Intelligence}, 2018.

\bibitem[Mitrovic et~al.(2017)Mitrovic, Bun, Krause, and
  Karbasi]{mitrovic2017differentially}
Marko Mitrovic, Mark Bun, Andreas Krause, and Amin Karbasi.
\newblock Differentially private submodular maximization: Data summarization in
  disguise.
\newblock In \emph{International Conference on Machine Learning}, pages
  2478--2487. PMLR, 2017.

\bibitem[Mitrovic et~al.(2018)Mitrovic, Kazemi, Zadimoghaddam, and
  Karbasi]{mitrovic2018data}
Marko Mitrovic, Ehsan Kazemi, Morteza Zadimoghaddam, and Amin Karbasi.
\newblock {Data Summarization at Scale: {A} Two-Stage Submodular Approach}.
\newblock In \emph{ICML}, pages 3593--3602, 2018.

\bibitem[Mitrovic et~al.(2019)Mitrovic, Kazemi, Feldman, Krause, and
  Karbasi]{mitrovic2019adaptive}
Marko Mitrovic, Ehsan Kazemi, Moran Feldman, Andreas Krause, and Amin Karbasi.
\newblock Adaptive sequence submodularity.
\newblock In \emph{Advances in Neural Information Processing Systems}, pages
  5352--5363, 2019.

\bibitem[Mokhtari et~al.(2018)Mokhtari, Hassani, and
  Karbasi]{mokhtari2017conditional}
Aryan Mokhtari, Hamed Hassani, and Amin Karbasi.
\newblock Conditional gradient method for stochastic submodular maximization:
  Closing the gap.
\newblock In \emph{AISTATS}, pages 1886--1895, 2018.

\bibitem[Nemhauser and Wolsey(1978)]{NW78}
G.~L. Nemhauser and L.~A. Wolsey.
\newblock Best algorithms for approximating the maximum of a submodular set
  function.
\newblock \emph{Mathematics of Operations Research}, 3\penalty0 (3):\penalty0
  177--188, 1978.

\bibitem[Nemhauser et~al.(1978)Nemhauser, Wolsey, and Fisher]{Nemhauser1978}
G.~L. Nemhauser, L.~A. Wolsey, and M.~L. Fisher.
\newblock An analysis of approximations for maximizing submodular set
  functions---i.
\newblock \emph{Mathematical Programming}, 14\penalty0 (1):\penalty0 265--294,
  1978.

\bibitem[Oveis~Gharan and Vondr{\'a}k(2011)]{GV11}
Shayan Oveis~Gharan and Jan Vondr{\'a}k.
\newblock Submodular maximization by simulated annealing.
\newblock In \emph{SODA}, pages 1098--1116, 2011.

\bibitem[Salehi et~al.(2017)Salehi, Karbasi, Scheinost, and Constable]{SKSC17}
Mehraveh Salehi, Amin Karbasi, Dustin Scheinost, and R.~Todd Constable.
\newblock {A Submodular Approach to Create Individualized Parcellations of the
  Human Brain}.
\newblock In \emph{MICCAI}, pages 478--485, 2017.

\bibitem[Seeger(2004)]{seeger2004greedy}
Matthias Seeger.
\newblock Greedy forward selection in the informative vector machine.
\newblock Technical report, Technical report, University of California at
  Berkeley, 2004.

\bibitem[Sipos et~al.(2012)Sipos, Swaminathan, Shivaswamy, and
  Joachims]{sipos2012temporal}
Ruben Sipos, Adith Swaminathan, Pannaga Shivaswamy, and Thorsten Joachims.
\newblock {Temporal corpus summarization using submodular word coverage}.
\newblock In \emph{International Conference on Information and Knowledge
  Management, (CIKM)}, pages 754--763, 2012.

\bibitem[Stan et~al.(2017)Stan, Zadimoghaddam, Krause, and Karbasi]{stanZK17}
Serban Stan, Morteza Zadimoghaddam, Andreas Krause, and Amin Karbasi.
\newblock {Probabilistic Submodular Maximization in Sub-Linear Time}.
\newblock In \emph{International Conference on Machine Learning (ICML)}, pages
  3241--3250, 2017.

\bibitem[Tohidi et~al.(2020)Tohidi, Amiri, Coutino, Gesbert, Leus, and
  Karbasi]{tohidi2020submodularity}
Ehsan Tohidi, Rouhollah Amiri, Mario Coutino, David Gesbert, Geert Leus, and
  Amin Karbasi.
\newblock Submodularity in action: From machine learning to signal processing
  applications.
\newblock \emph{IEEE Signal Processing Magazine}, 37\penalty0 (5):\penalty0
  120--133, 2020.

\bibitem[Tschiatschek et~al.(2014)Tschiatschek, Iyer, Wei, and Bilmes]{TIWB14}
Sebastian Tschiatschek, Rishabh~K. Iyer, Haochen Wei, and Jeff~A. Bilmes.
\newblock Learning mixtures of submodular functions for image collection
  summarization.
\newblock In \emph{NIPS}, pages 1413--1421, 2014.

\bibitem[UberDataset()]{uberpickup}
UberDataset.
\newblock Uber pickups in new york city, 2014.
\newblock URL
  \url{https://www.kaggle.com/fivethirtyeight/uber-pickups-in-new-york-city}.

\bibitem[Vondr{\'{a}}k(2013)]{V13}
Jan Vondr{\'{a}}k.
\newblock Symmetry and approximability of submodular maximization problems.
\newblock \emph{{SIAM} J. Comput.}, 42\penalty0 (1):\penalty0 265--304, 2013.

\bibitem[Ward(2012)]{W12}
Justin Ward.
\newblock A (k+3)/2-approximation algorithm for monotone submodular k-set
  packing and general k-exchange systems.
\newblock In \emph{STACS}, pages 42--53, 2012.

\end{thebibliography}

\end{document}